\documentclass[journal]{IEEEtran}

\usepackage{bibunits} 
\usepackage{amsmath,graphicx}
\usepackage{amssymb,amsfonts}
\usepackage{amsbsy,bbm}
\usepackage{amsthm}
\usepackage{array}
\usepackage{cite}
\usepackage{dsfont}
\usepackage[utf8]{inputenc}
\usepackage[T1]{fontenc}
\usepackage{hhline}
\usepackage{multicol,float}
\usepackage{pgfplots}
\usetikzlibrary{plotmarks}
\usepackage{tabularx}
\usepackage{tikz}
\usetikzlibrary{arrows.meta, positioning}
\usepackage{stmaryrd}
\usepackage{enumerate}
\pgfplotsset{compat=1.14}
\usepackage{url}

\newtheorem{theorem}{Theorem}
\newtheorem{lemma}{Lemma}

\newtheorem{remark}{Remark}

\newtheorem{assumption}{Assumption}

\begin{document}
\begin{bibunit}[IEEEtran] 

\title{A Robust Framework for Model Order Selection in Correlated Large-Dimensional CES Noise}

\author{\IEEEauthorblockN{Eugénie Terreaux}, \IEEEauthorblockN{Emmanuelle Jay}, \IEEEauthorblockN{Frédéric Pascal} and \IEEEauthorblockN{Jean-Philippe Ovarlez}
\thanks{Eugénie Terreaux is with DEMR, ONERA, Universit\'e Paris-Saclay, F-91120 Palaiseau, France (e-mail: eugenie.terreaux@onera.fr).}
\thanks{Emmanuelle Jay is with LBP AM, 36 Quai Henri IV, F75004, Paris, France (e-mail: emmanuelle.jay@@lbpam.com).}
\thanks{Frédéric Pascal is with L2S, CentraleSupélec, Université Paris-Saclay, F91190 Gif-sur-Yvette, France (e-mail: frederic.pascal@centralesupelec.fr).}
\thanks{Jean-Philippe Ovarlez is with SONDRA, CentraleSup\'elec, Universit\'e Paris-Saclay, F91190, Gif-sur-Yvette, France and also with DEMR, ONERA, Universit\'e Paris-Saclay, F-91123, Palaiseau, France (jean-philippe.ovarlez@onera.fr).}}

\maketitle

\begin{abstract}
This paper addresses model order selection under large-dimensional, correlated, non-Gaussian noise. Sources are assumed to be embedded in additive Complex Elliptically Symmetric (CES) noise with an unknown Toeplitz-structured scatter matrix. We propose a two-stage robust framework: (i) a noise-whitening step based on a Toeplitz-rectified $M$-estimator of the scatter matrix, and (ii) signal subspace rank inference via large-dimensional Random Matrix Theory (RMT). Almost sure consistency of the proposed estimators is established, together with explicit RMT eigenvalue upper bounds separating signal from noise components, in the regime where the observation dimension $m$ and the sample size $N$ grow proportionally. Three estimation branches are derived, based respectively on the sample covariance matrix (SCM), Maronna's $M$-estimator, and the distribution-free Tyler $M$-estimator for whitening. The methodology is validated on synthetic data, real hyperspectral images, EEG recordings, and financial data, with significant gains over AIC and unwhitened methods.
\end{abstract}


\IEEEpeerreviewmaketitle
\markboth{Submitted to IEEE Transactions on Signal Processing - T-SP-35633-2026}%
{A Robust Two-Stage Framework for Model Order Selection in Large-Dimensional Correlated CES Noise}
\section{Introduction}

\IEEEPARstart{M}{odel} order selection---the problem of estimating the
number of statistically significant components in a noisy observation---is fundamental in signal processing, with applications in wireless communications~\cite{Julia13}, array processing~\cite{Nadler10}, and many other areas~\cite{Ottersten92,Nadler11}. In the classical white-noise setting, the model order can be estimated from the eigenvalue structure of the sample covariance matrix via information-theoretic criteria such as the Akaike Information Criterion (AIC)~\cite{Akaike74}, and related subspace methods have been applied to source localization~\cite{schmidt86}, channel identification~\cite{Meraim97}, and waveform estimation~\cite{Liu96}. However, these methods become unreliable in the presence of large-dimensional and correlated data. Although some extensions have been proposed for correlated signals \cite{Bai98, Silverstein95}, they cannot be generalized to arbitrary scenarios, and whitening requires knowledge of the noise covariance structure, which is typically unknown~\cite{Cawse11}. Moreover, in high-dimensional regimes, the sample covariance matrix no longer exhibits the same properties as in the classical setting, leading to poor estimation performance \cite{Combernoux14, Nadler9, Farsi15}. In the large-dimensional regime, where both the number of snapshots $N$ and the signal dimension $m$ grow to infinity with a constant ratio, Random Matrix Theory provides powerful tools for model order selection. In particular, methods based on the asymptotic behavior of the largest eigenvalues of the covariance matrix have been proposed \cite{Couillet13}. RMT offers a rigorous probabilistic framework to analyze large random matrices; see \cite{Couillet11} for a comprehensive review and, e.g., \cite{Nadler9, Hachem13, Pascal16, Couillet15, Cawse10} for applications to detection, MUSIC, radar, and hyperspectral imaging, respectively. When the noise is correlated, model order estimation remains possible, for instance by analyzing the spacing between eigenvalues \cite{Vinogradova13, Vinogradova2015}. Moreover, many practical signal processing applications involve impulsive, heavy-tailed, or non-Gaussian perturbations that cannot be accurately modeled within the classical Gaussian framework. This motivates the use of more general distribution families together with robust covariance estimation techniques. When correlation and heavy-tailed behaviour occur jointly, both covariance estimation and eigenvalue-based model-order selection become significantly more challenging. To mitigate sensitivity to the noise distribution, robust model order selection methods have been proposed, notably in hyperspectral imaging \cite{Breloy16, Halimi16}. Nevertheless, these approaches often depend on unknown parameters \cite{Julia13} or are not well-suited for large-dimensional settings. Recent advances in RMT enable consistent estimation of covariance (or scatter) matrices for textured signals \cite{Couillet15b}. However, these results typically assume that the noise correlation structure is known and that the data can be whitened beforehand, which is rarely the case in practice. The present paper addresses both challenges simultaneously.

\medskip
In this work, we consider noise following a Complex Elliptically Symmetric (CES) distribution \cite{Kelker70, Yao73} (see also \cite{Ollila12} for a comprehensive overview). CES distributions are widely used in signal processing due to their flexibility, as they can model a broad class of random signals beyond the Gaussian assumption. A CES random vector can be decomposed into two components: a texture and a speckle. Its second-order structure is characterized by a scatter matrix, which coincides with the covariance matrix up to a scaling factor. CES models have been successfully applied in various fields, including hyperspectral imaging \cite{Ovarlez11}, radar clutter modeling \cite{Gini97}, and finance, particularly in portfolio optimization \cite{Yang15}. 

\medskip
This paper proposes a robust two-stage framework that combines Toeplitz-structured scatter matrix estimation, whitening, and RMT-based eigenvalue thresholding to estimate the number of sources in correlated CES noise, extending~\cite{Couillet15b, Vinogradova14, Vinogradova2015} to left-sided correlation with an \emph{unknown} scatter matrix under non-Gaussian noise. \\

\noindent\textbf{Contributions:}
\begin{enumerate}[(i)]\setlength{\itemsep}{0pt}\setlength{\parsep}{0pt}
  \item Almost sure consistency of Toeplitz-structured scatter matrix estimators under correlated CES noise, even in the presence of sources: SCM-based and Maronna-based (Theorems~\ref{thm:consist_scm} and~\ref{thm:consist_fp}, Supplementary Material) and Tyler-based (Theorem~\ref{thm:consist_tyl}~\cite{TerreauxICASSP18}).
  \item Robust whitening via Toeplitz rectification of the SCM (Section~\ref{sec::3}-A) or Maronna's $M$-estimator (Section~\ref{sec::4}-A), with asymptotic validity.
  \item Explicit closed-form RMT eigenvalue thresholds: Maronna-based threshold $t = \Phi_\infty(1+\sqrt{c})^2 /  [\xi\,(1-c\,\Phi_\infty)]$ (Theorems~\ref{thm:conv_scm} and~\ref{thm:conv_fp}) and Mar\v{c}enko--Pastur upper edge $(1+\sqrt{c})^2$ for the distribution-free Tyler branch 
  (Theorem~\ref{thm:tyler_est}).
\end{enumerate}
\noindent Detailed proofs are provided in the Supplementary Material
(Sections~S-A and S-B).

\medskip
\textit{Notations}: Matrices are in bold and capital, vectors in bold. Let $\mathbf{X}$ be a square matrix of size $s\times s$, $\left(\lambda_{i}(\mathbf{X})\right)_i$, $i \in \{1,\ldots,s\}$, are the eigenvalues of $\mathbf{X}$. $\operatorname{tr}(\mathbf{X})$ is the trace of the matrix $\mathbf{X}$. $\left\Vert \mathbf{X} \right\Vert$ stands for the spectral norm. Let $\mathbf{A}$ be a matrix, $\mathbf{A}^T$ is the transpose of $\mathbf{A}$ and $\mathbf{A}^H$ the Hermitian transpose of $\mathbf{A}$. $\mathbf{I}_p$ is the $p \times p$ identity matrix. For any $m$-vector $\mathbf{x}$, $\mathcal{L} :\, \mathbf{x} \mapsto \mathcal{L}(\mathbf{x})$ is the $m \times m$ matrix defined as the Toeplitz operator: $\left([\mathcal{L}(\mathbf{x})]_{i,j}\right)_{i\leq j} = x_{i-j}$ and $\left([\mathcal{L}(\mathbf{x})]_{i,j}\right)_{i>j} = x_{i-j}^\ast$. For $\mathbf{A} \in \mathbb{C}^{m \times m}$, the Toeplitz rectification $\mathcal{T}(\mathbf{A})$ is defined as $\mathcal{L}(\tilde{\mathbf{a}})$ with $\tilde{a}_k = \displaystyle\frac{1}{m} \sum_{i-j=k} A_{i,j}$ for $k \in \{0, \dots, m-1\}$. For any complex $z$, $z^{\star}$ is the conjugate of $z$. $\mathcal{R}e(\cdot)$ and $\mathcal{I}m(\cdot)$ denote the real and imaginary parts. The notation $\overset{a.s.}{\longrightarrow}$ means ``converges almost surely''. $C$ generally denotes a real positive constant.

\section{Model and Assumptions}
\label{sec::2}

This section presents the signal model and the assumptions required for the theoretical analysis. We consider the following standard source-plus-noise model. Let $\mathbf{y}_0, \ldots, \mathbf{y}_{N-1}$ be $N$ observations of dimension $m$, corresponding to $p$ mixed sources embedded in additive CES noise:
\begin{equation}
\mathbf{y}_i = \mathbf{M}\, \mathbf{s}_i + \sqrt{\tau_i} \,\mathbf{C}^{1/2} \, \mathbf{x}_i\,,
\quad i \in \{0,\ldots,N-1\} \,, \label{modele}
\end{equation}
The first term represents $p$ deterministic sources, while the second term models CES noise with texture $\tau_i$ and scatter 
matrix $\mathbf{C}$. We can rewrite this equation in a more compact form $\mathbf{Y} = \mathbf{M}\, \mathbf{S} + \mathbf{C}^{1/2}\, \mathbf{X}\, \mathbf{T}^{1/2}$, where
\begin{itemize}
  \item $\mathbf{Y}=[\mathbf{y}_0,\ldots,\mathbf{y}_{N-1}] \in \mathbb{C}^{m\times N}$,
  \item $\mathbf{M}\in \mathbb{C}^{m\times p}$ is the mixing matrix, whose columns $\left\{\mathbf{m}_j\right\}_{j\in\{1,\ldots,p\}}$ represent the source signatures.
  \item The matrix $\mathbf{S} = \left[\mathbf{s}_0, \ldots, \mathbf{s}_{N-1}\right] = \boldsymbol{\Gamma}^{1/2}\, \boldsymbol{\delta}^H \in \mathbb{C}^{p \times N}$ with $\boldsymbol{\delta} \in\mathbb{C}^{N\times p}$ having i.i.d.\ $\mathcal{CN}(0,1)$ entries, independent of $\mathbf{X}$, contains the source amplitude where each entry $s_{i,j}$ represents the contribution of source $j$ in observation $i$. 
  \item $\boldsymbol{\Gamma}\in\mathbb{C}^{p\times p}$ is Hermitian nonnegative definite,
  \item $\mathbf{X}=[\mathbf{x}_0,\ldots,\mathbf{x}_{N-1}]\in\mathbb{C}^{m\times N}$ has i.i.d.\ entries with zero mean, unit variance, and sub-Gaussian tails,
  \item $\mathbf{T}=\mathrm{diag}(\tau_0,\ldots,\tau_{N-1})$, with $\{\tau_i\}$ i.i.d.\ positive random variables referred to as textures, with an unspecified distribution.
  \item $\mathbf{C} = \mathcal{L}\bigl([c_0,\ldots,c_{m-1}]^T\bigr)\in \mathbb{C}^{m\times m}$ is a Hermitian nonnegative definite Toeplitz scatter matrix. This structural assumption facilitates estimation and is consistent with many practical applications. 
\end{itemize}

In the sequel, we will consider the following assumptions: 

\begin{assumption}[Large-dimensional regime]\label{ass:regime}
One assumes the usual random matrix regime, \textit{i.e.}, $N \rightarrow \infty$, $m \rightarrow \infty$ and $c_N = \displaystyle \frac{m}{N} \rightarrow c > 0$. 
\end{assumption}

\begin{assumption}[Model parameters]\label{ass:model}
\begin{enumerate}[(i)]
\item $p \in \mathbb{N}^{*}$ is fixed.
\item The coefficients $\{c_k\}_{k=0}^{m-1}$ are absolutely summable with $c_0 \neq 0$, and the empirical spectral distribution (ESD) of $\mathbf{C}$ converges a.s. to $\mu_{\mathbf{C}}$ as $m \to \infty$.
\item We impose the following condition on the texture $\tau$. We assume that $\tau$ has a regularly varying tail with index $\alpha_\tau > 2$, i.e., 
\[
P(\tau > x) = x^{-\alpha_\tau}\,\ell(x), \quad x \to \infty\, , \]
where $\ell:(0,\infty)\to (0,\infty)$ is slowly varying, i.e., $\ell(tx)/\ell(x)\to 1$ as $x\to\infty$ for every $t>0$ \cite{Bingham1987}. This assumption implies in particular that $\mathbb{E}[\tau^p] < \infty$ for all $p < \alpha_\tau$, 
and hence $\mathbb{E}[\tau^2] < \infty$. We normalize $\mathbb{E}[\tau]=1$ without loss of generality, since any positive constant factor in $\tau$ can be absorbed into $\mathbf{C}$. This ensures that $\check{\mathbf{C}}_{\mathrm{SCM}}$ in Theorem~\ref{thm:consist_scm} (Section~\ref{sec::3})  consistently estimates $\mathbf{C}$ (rather than a scaled version). Such distributions belong to the \emph{max-domain of attraction of the Fr\'echet distribution} with index $\alpha_\tau$ \cite{deHaan2006}, and are characterized by a polynomial tail decay modulated by a slowly varying function $\ell$. This class encompasses the two texture distributions used in the simulations. Examples include Inverse-Gamma$(\alpha,\beta)$ ($\alpha_\tau=\alpha>2$) and Student-$t^2(\nu_s)$ ($\alpha_\tau=\nu_s/2>2$), used in the simulations below.
\item $\mathbf{X}$ has i.i.d.\ complex entries with zero mean, unit variance, and sub-Gaussian tails, while heavy-tailed behavior is captured by the texture $\tau$. This assumption allows the use of concentration inequalities for quadratic forms (e.g., Hanson--Wright inequality \cite{Hanson1971, Vershynin2018, Rudelson2013}).
\item The limiting ratio $c$ satisfies $c < \Phi_\infty^{-1}$, where $\Phi_\infty \triangleq \lim_{x\to\infty} x\,u(x)$ is the asymptotic value of the weight function $u$ defined in Section~\ref{sec::3}-B, Eq.~\eqref{eqsigma1}. This condition ensures the existence and uniqueness of Maronna's $M$-estimator \cite{Maronna76}. (This assumption is vacuously satisfied for the Tyler estimator of Section~\ref{sec::5}, where $u(x)=m/x$ and no such constraint is needed.)
\end{enumerate}
\end{assumption}

\begin{assumption}[Signal parameters]\label{ass:signal}
\begin{enumerate}[(i)]
\item In each column of $\mathbf{M}$, the coefficients are absolutely summable, that is, for all fixed $j$, $\displaystyle\sum_{i=1}^m | [\mathbf{M}]_{i,j} | < \infty$, and $\left\Vert \mathbf{M} \right\Vert < \infty$. This is a common assumption in several applications, and especially in hyperspectral imaging.
\item $\boldsymbol{\Gamma} \in \mathbb{C}^{p\times p}$ has absolutely summable coefficients (i.e., $\|\boldsymbol{\Gamma}\|_1
  < \infty$), which implies $\|\boldsymbol{\Gamma}\| < \infty$.
\end{enumerate}
\end{assumption}

The overall procedure is summarized in Fig.~\ref{fig:pipeline}. The three whitening strategies correspond to the SCM-based, Maronna-based, and Tyler-based estimators, which respectively provide a baseline, a robust parametric, and a distribution-free approach.

\begin{figure}
\centering
\includegraphics[width=\columnwidth]{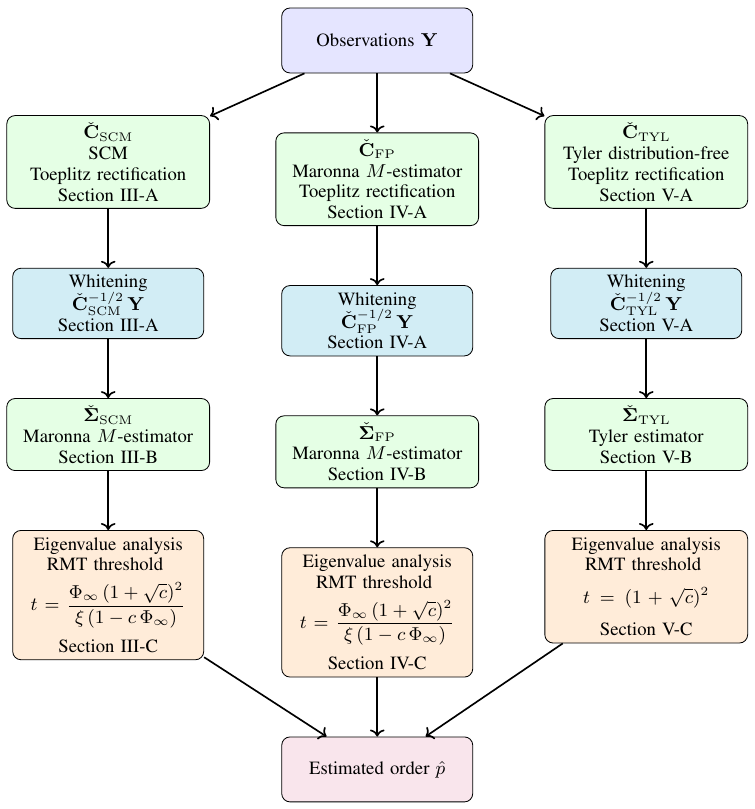}
\caption{Three-branch pipeline: each branch uses a specific whitening estimator (\textit{top}), the matched scatter estimator (\textit{middle}), and RMT-based eigenvalue thresholding (\textit{bottom}) to produce $\hat{p}$}
\label{fig:pipeline}
\end{figure}

\section{SCM-Based Whitening with Robust Scatter Estimation}
\label{sec::3}

In this section, we exploit the consistency of the Sample Covariance Matrix (SCM) to perform a whitening step and subsequently estimate the model order using Maronna's $M$-estimator. Model order estimation via Maronna's $M$-estimator was studied in~\cite{Couillet15b} for the spiked model with CES \emph{white} noise.
Here we address the more challenging scenario of \emph{correlated} CES noise. For clarity, Table~\ref{tab:matrices} summarizes the main matrices used throughout the paper.

\begin{table}[htbp]
\centering
\caption{Summary of main matrices and their roles}
\begin{tabular}{ll}
\hline
\textbf{Notation} & \textbf{Description} \\
\hline
$\mathbf{Y}$ & Observation matrix \\
$\mathbf{M}$ & Mixing (signal) matrix \\
$\mathbf{S}$ & Source matrix \\
$\mathbf{X}$ & Noise (speckle) matrix \\
$\mathbf{T}$ & Texture diagonal matrix \\
$\mathbf{C}$ & True noise scatter matrix (Toeplitz) \\
$\hat{\mathbf{C}}_{SCM}$ & Sample covariance matrix \\
$\check{\mathbf{C}}_{SCM}$ & Toeplitz-rectified SCM (consistent estimator of $\mathbf{C}$) \\
$\hat{\mathbf{C}}_{FP}$ & Maronna estimator \\
$\widetilde{\mathbf{C}}_{FP}$ & Toeplitz-rectified Maronna estimator\ (scaled) \\
$\check{\mathbf{C}}_{FP}$ & Toeplitz-rectified Maronna estimator\ of $\mathbf{C}$ (rescaled $\widetilde{\mathbf{C}}_{FP}$) \\
$\hat{\mathbf{C}}_{TYL}$ & Tyler $M$-estimator \\
$\check{\mathbf{C}}_{TYL}$ & Toeplitz-rectified Tyler est.\ of $\mathbf{C}$ \\
$\mathbf{Y}_{w}$ & Ideally whitened data \\
$\check{\mathbf{Y}}_{wSCM}$ & Whitened data (SCM-based) \\
$\check{\mathbf{Y}}_{wFP}$ & Whitened data (Maronna-based) \\
$\hat{\mathbf{Y}}_{w}$ & Whitened data (Tyler-based) \\

$\check{\boldsymbol{\Sigma}}_{SCM}$ & $M$-estimator on SCM-whitened data \\
$\check{\boldsymbol{\Sigma}}_{FP}$ & $M$-estimator on Maronna-whitened data \\
$\hat{\mathbf{S}}$ & Surrogate/oracle RMT matrix (decoupled weights) \\
$\hat{\mathbf{S}}_w$ & Surrogate RMT matrix for Tyler branch \\
$\check{\boldsymbol{\Sigma}}_{TYL}$ & Tyler estimator of Tyler-whitened data \\
\hline
\end{tabular}
\label{tab:matrices}
\end{table}

\subsection{Whitening Step}
We set the following two matrices:
\begin{equation}
\check{\mathbf{C}}_{SCM} = \mathcal{T}\!\left(\hat{\mathbf{C}}_{SCM}\right), \qquad
\hat{\mathbf{C}}_{SCM} = \frac{1}{N}\,\mathbf{Y}\,\mathbf{Y}^H\, .
\label{eq1}
\end{equation}

\begin{theorem}[Consistency of $\check{\mathbf{C}}_{SCM}$]
\label{thm:consist_scm}
Under Assumptions~\ref{ass:regime}--\ref{ass:signal},
\begin{equation}
\left\Vert \check{\mathbf{C}}_{SCM} - \mathbf{C} \right\Vert \overset{a.s.}{\longrightarrow} 0 \,.
\label{eqth1}
\end{equation}
The Toeplitz rectification asymptotically suppresses the low-rank signal contribution under Assumption~3(i), thereby preserving consistency even in the presence of sources.
\end{theorem}
\begin{IEEEproof}[Proof Outline]
See Supplementary Material S~A.1. The key step is that $\mathbf{M}\,\mathbf{S}\,\mathbf{S}^H\mathbf{M}^H/N$ is a rank-$p$ perturbation; Assumption~\ref{ass:signal}-(i) (absolute summability of the columns of $\mathbf{M}$) ensures that its Toeplitz rectification vanishes in spectral norm almost surely as $m, N\to\infty$ (see the fourth term in the decomposition of Supplementary Material S-A.1). Combined with the noise term convergence,
this gives $\check{\mathbf{C}}_{SCM}\overset{a.s.}{\to}
\mathbb{E}[\tau]\,\mathbf{C}=\mathbf{C}$ since $\mathbb{E}[\tau]=1$.
\end{IEEEproof}

The whitened observations are:
\begin{equation}
\check{\mathbf{Y}}_{wSCM} = \check{\mathbf{C}}_{SCM}^{-1/2} \, \mathbf{M} \, \boldsymbol{\Gamma}^{1/2}\boldsymbol{\delta}^H +
\check{\mathbf{C}}_{SCM}^{-1/2} \, \mathbf{C}^{1/2}\,\mathbf{X} \, \mathbf{T}^{1/2}\,.
\label{scm_blanc}
\end{equation}

\subsection{Estimation of the Scatter Matrix}

Having established that $\check{\mathbf{C}}_{SCM}$ consistently estimates $\mathbf{C}$, we apply Maronna's $M$-estimator \cite{Maronna76} to the whitened observations $\check{\mathbf{Y}}_{wSCM}$ in~\eqref{scm_blanc} to robustly estimate the scatter matrix of the decorrelated signal. The estimator $\check{\boldsymbol{\Sigma}}_{SCM}$ is the unique solution of:
\begin{equation}
\boldsymbol{\Sigma} = \frac{1}{N} \displaystyle\sum_{i=0}^{N-1}
u\!\left(\frac{1}{m}\,\check{\mathbf{y}}_{wSi}^H\,\boldsymbol{\Sigma}^{-1}\,
\check{\mathbf{y}}_{wSi}\right) \check{\mathbf{y}}_{wSi} \, \check{\mathbf{y}}_{wSi}^H \,,
\label{eqsigma1}
\end{equation}
where $\check{\mathbf{Y}}_{wSCM} = [\check{\mathbf{y}}_{wS0}, \ldots, \check{\mathbf{y}}_{wS\, N-1}]$ and $u:[0,+\infty)\to (0,+\infty)$ satisfies:
\begin{enumerate}[(i)]
\item $u$ is nonnegative, continuous, and non-increasing \cite{Mahot12};
\item $\phi(x) = x\,u(x)$ is increasing and bounded, with $\displaystyle\lim_{x\to\infty}\phi(x) = \Phi_{\infty} > 1$;
\item $c < \Phi_{\infty}^{-1}$ \quad (cf.\ Assumption~\ref{ass:model}-(v)).
\end{enumerate}

These conditions guarantee the existence and uniqueness of the solution \cite{Maronna76, Mahot12}.

\subsection{Model Order Selection}

\begin{figure*}[htbp]
\centering
\includegraphics[width=0.95\columnwidth]{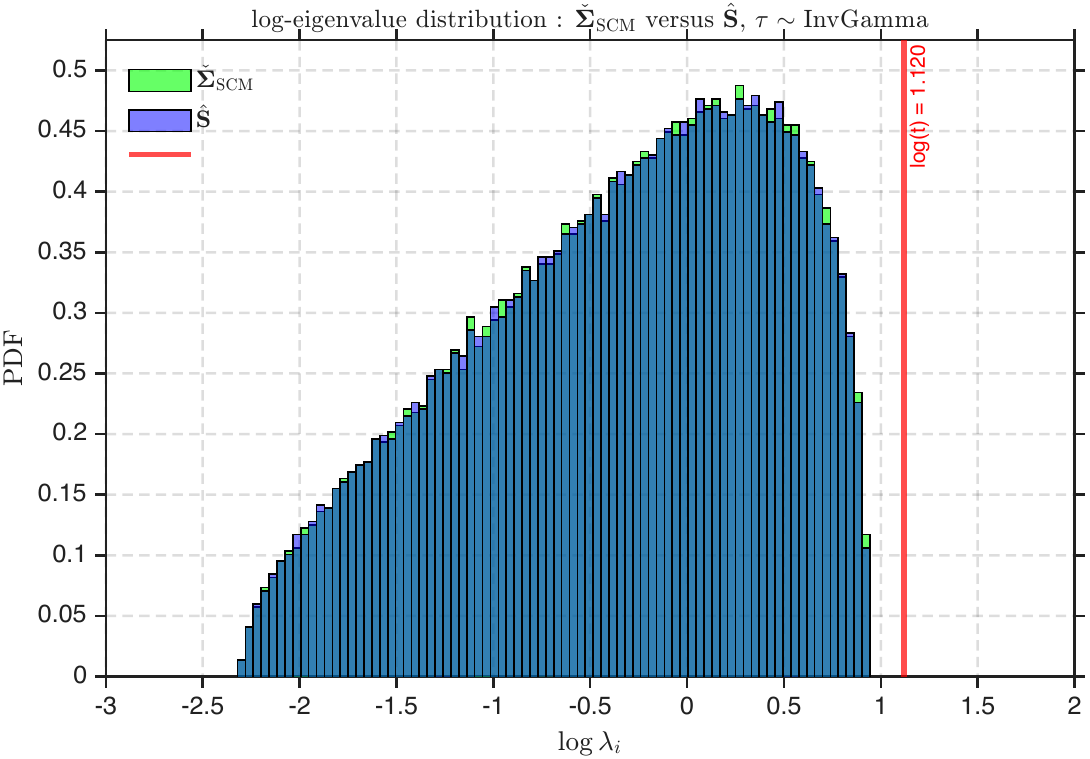} \hfill 
\includegraphics[width=0.95\columnwidth]{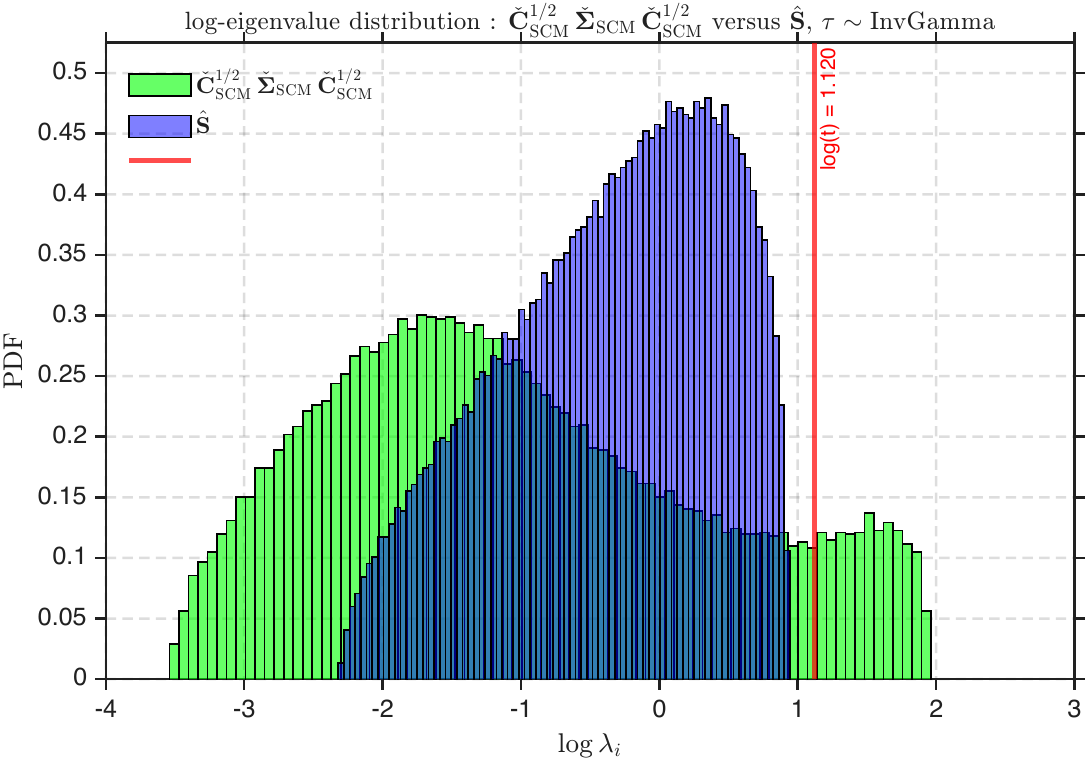}
\caption{Log-eigenvalue distributions ($m=900$, $N=2000$, $\rho=0.7$, 
$\tau\!\sim\!\mathrm{InvGamma}$, $\nu=0.1$). \textit{Left}: after whitening by $\check{\mathbf{C}}_{SCM}^{-1/2}$, the log-eigenvalue distribution of $\check{\boldsymbol{\Sigma}}_{SCM}$ (green) closely tracks that of $\hat{\mathbf{S}}$ (blue), and all eigenvalues lie below the RMT threshold $t$ (red), validating Theorem~\ref{thm:conv_scm} and the source-free upper bound~\eqref{eqseuil}. \textit{Right}: without whitening, the eigenvalues of $\check{\mathbf{C}}_{SCM}^{1/2}\, \check{\boldsymbol{\Sigma}}_{SCM}\,\check{\mathbf{C}}_{SCM}^{1/2}$ (green) spread well beyond $t$, confirming that the threshold has no theoretical justification unless the correlation structure is removed beforehand.}
\label{fig:whitening_comparison}
\end{figure*}

\paragraph{Threshold Derivation}
To count sources, we study the eigenvalue distribution of $\check{\boldsymbol{\Sigma}}_{SCM}$. A direct RMT analysis is intractable because the weight $u\!\left(\displaystyle\frac{1}{m}\,\check{\mathbf{y}}_{wSi}^H\,
\check{\boldsymbol{\Sigma}}_{SCM}^{-1}\,\check{\mathbf{y}}_{wSi}\right)$ depends on the very vector $\check{\mathbf{y}}_{wSi}$ it weights, creating a circular dependency. To break this dependency, we
introduce a surrogate matrix defined on the \emph{ideally} whitened signal. We therefore work with an auxiliary matrix defined on the \emph{ideally whitened} signal $\mathbf{Y}_w = \mathbf{C}^{-1/2}\, \mathbf{Y}=[\mathbf{y}_{w0}, \ldots, \mathbf{y}_{w N-1}]$:
\begin{equation}
\mathbf{Y}_w = \mathbf{C}^{-1/2}\,\mathbf{M}\,\boldsymbol{\Gamma}^{1/2}\, \boldsymbol{\delta}^H
+ \mathbf{X}\,\mathbf{T}^{1/2} \,.
\label{scm_normal}
\end{equation}

To overcome the dependence between the weights and the data, we introduce the surrogate matrix:
\begin{equation}
  \hat{\mathbf{S}}
  \triangleq \frac{1}{N}\sum_{i=0}^{N-1}
     v(\tau_i\,\xi)\,\mathbf{y}_{wi}\,\mathbf{y}_{wi}^H
   = \frac{1}{N}\sum_{i=0}^{N-1}
     v(\tau_i\,\xi)\,\tau_i\,\mathbf{x}_i\,\mathbf{x}_i^H.
  \label{matrixS}
\end{equation}
with $g(x) = x/(1-c\,\phi(x))$, which is strictly increasing on $[0,+\infty)$ under Assumption~2-(v) (since $c\,\phi(x) < c\,\Phi_\infty < 1$ for all $x$), and hence invertible on its range. Let $\psi(x)=x\,v(x)$, and let $\xi>0$ be the unique positive solution of
\begin{equation}
\displaystyle\frac{1}{N}\sum_{i=0}^{N-1} \frac{\psi(\tau_i\,\xi)}{1+c\,\psi(\tau_i\,\xi)} = 1\, .
\label{xi}
\end{equation}
Existence and uniqueness of $\xi$ follow from the monotonicity of the left-hand side: since $\psi(0)=0$ and $\psi(\tau_i\,\xi)\to\Phi_\infty$ as $\xi\to\infty$,
the left-hand side of~\eqref{xi} increases continuously from $0$ to $\displaystyle\frac{\Phi_\infty}{1+c\,\Phi_\infty}$, which exceeds $1$ under Assumption~\ref{ass:model}-(v) (since $c < \Phi_\infty^{-1}$ implies $\Phi_\infty(1-c)>1$), guaranteeing existence and uniqueness of~$\xi$ (see \cite{Couillet15b} for details). The weights $\{v(\tau_i\,\xi)\}$ are independent of $\{\mathbf{x}_i\}$, which makes $\hat{\mathbf{S}}$ more tractable. This decoupling is key to enabling a tractable RMT analysis.

It is proven in \cite{Couillet15b} that $\left\|\hat{\boldsymbol{\Sigma}} - \hat{\mathbf{S}}\right\|\overset{a.s.}{\to}0$, where $\hat{\boldsymbol{\Sigma}}$ is Maronna's $M$-estimator applied to the \emph{ideally} whitened data $\mathbf{Y}_w$. The key insight is that $\hat{\mathbf{S}}$ replaces the data-dependent weights
$u\!\left(\displaystyle\frac{1}{m}\,\mathbf{y}_{wi}^H\,\hat{\boldsymbol{\Sigma}}^{-1}\,\mathbf{y}_{wi}\right)$ by the \emph{oracle} weights $v(\tau_i\,\xi)$, which are independent of $\mathbf{x}_i$, making $\hat{\mathbf{S}}$ tractable for RMT analysis. The following theorem then extends this to the empirically whitened case.

\begin{theorem}[Convergence of $\check{\boldsymbol{\Sigma}}_{SCM}$]
\label{thm:conv_scm}
Under Assumptions~\ref{ass:regime}--\ref{ass:signal},
\begin{equation}
\left\Vert \check{\boldsymbol{\Sigma}}_{SCM} - \hat{\mathbf{S}} \right\Vert
\overset{a.s.}{\longrightarrow} 0 \,.
\label{eqcouilletmod}
\end{equation}
\end{theorem}
\begin{IEEEproof}[Proof outline]
See Supplementary Material S-B. Define $\widetilde{\boldsymbol{\Sigma}}\triangleq
\mathbf{C}^{-1/2}\,\check{\mathbf{C}}^{1/2}\,\check{\boldsymbol{\Sigma}}
\,\check{\mathbf{C}}^{1/2}\,\mathbf{C}^{-1/2}$.
Substituting $\check{\mathbf{y}}_{wi}=\check{\mathbf{C}}^{-1/2}\,\mathbf{C}^{1/2}\,\mathbf{y}_{wi}$
into the fixed-point equation of $\check{\boldsymbol{\Sigma}}$ shows that
$\widetilde{\boldsymbol{\Sigma}}$ satisfies the same equation as $\hat{\boldsymbol{\Sigma}}$
on $\mathbf{Y}_w$; uniqueness~\cite{Couillet15b} then gives $\check{\boldsymbol{\Sigma}}_{SCM}=\check{\mathbf{C}}^{-1/2}\,\mathbf{C}^{1/2}
\,\hat{\boldsymbol{\Sigma}}\,\mathbf{C}^{1/2}\,\check{\mathbf{C}}^{-1/2}$.
Since $\|\check{\mathbf{C}}-\mathbf{C}\|\xrightarrow{\mathrm{a.s.}}0$ (Theorem~1)
and $A\mapsto A^{-1/2}$ is Lipschitz on positive-definite matrices,
$\|\check{\mathbf{C}}^{-1/2}\,\mathbf{C}^{1/2}-\mathbf{I}\|\xrightarrow{\mathrm{a.s.}}0$,
and the result follows from $\|\hat{\boldsymbol{\Sigma}}_{SCM}-\hat{\mathbf{S}}\|
\xrightarrow{\mathrm{a.s.}}0$~\cite{Couillet15b}.
\end{IEEEproof}

\paragraph{Source-free upper bound} 
Since $\psi(x)=x\,v(x)$ by definition, one has $v(\tau_i\,\xi)\,\tau_i = v(\tau_i\,\xi)\,\tau_i\,\xi/\xi = \psi(\tau_i\,\xi)/\xi$. In the source-free case ($\mathbf{M}=\mathbf{0}$), Eq.~\eqref{matrixS} becomes
\begin{equation}
  \hat{\mathbf{S}}
  = \frac{1}{N}\sum_{i=0}^{N-1}
    \frac{\psi(\tau_i\,\xi)}{\xi}\,\mathbf{x}_i\mathbf{x}_i^H.
  \label{eq:Shat_sourcefree}
\end{equation}
Since $\{\tau_i\}$ are independent of $\{\mathbf{x}_i\}$ and $0 < \psi(\tau_i\,\xi)/\xi \leq \Phi_\infty/\xi$ almost surely, a standard comparison argument gives $\left\|\hat{\mathbf{S}}\right\| \leq \displaystyle\frac{\Phi_\infty}{\xi} \left\|\frac{1}{N}\,\mathbf{X}\mathbf{X}^H\,\right\|$ (see Supplementary Material~S-A.2). More precisely, since $v(\tau_i\,\xi)\leq\Phi_\infty/\xi$ and the
weights are independent of $\{\mathbf{x}_i\}$, the matrix inequality
$\hat{\mathbf{S}}\preceq \displaystyle\frac{\Phi_\infty}{\xi}\,\frac{1}{N}\,\mathbf{X}\,\mathbf{X}^H$
holds almost surely entry-wise in a quadratic form sense,
giving $\|\hat{\mathbf{S}}\|\leq \displaystyle\frac{\Phi_\infty}{\xi}
\left\|\frac{1}{N}\,\mathbf{X}\,\mathbf{X}^H\right\|$. Applying the Bai--Silverstein almost sure upper bound \cite{Bai98}, which guarantees that $\lambda_{\max}\!\left(\displaystyle\frac{1}{N} \,\mathbf{X}\, \mathbf{X}^H\right)\overset{a.s.}{\longrightarrow}\left(1+\sqrt{c}\right)^2$, then yields:
\begin{equation}
\left\Vert \hat{\mathbf{S}} \right\Vert \leq t, \qquad
t = \frac{\Phi_{\infty}\,(1+\sqrt{c})^2}{\xi\,(1-c\,\Phi_{\infty})} > 0\,,
\label{eqseuil}
\end{equation}
where $c\,\Phi_\infty<1$ is guaranteed by Assumption~\ref{ass:model}-(v). The additional factor $(1-c\,\Phi_\infty)^{-1}$ in~\eqref{eqseuil} 
arises from the definition $v = u\,\circ \,g^{-1}$ with 
$g(x) = x/(1-c\,\phi(x))$: one has $v(x)\leq \Phi_\infty/\xi$ 
but more precisely $\psi(x)=x\,v(x)\leq \Phi_\infty$ and 
$\xi\geq \xi_{\min}>0$, so that 
$v(\tau_i\,\xi)\,\tau_i = \psi(\tau_i\,\xi)/\xi \leq \Phi_\infty/\xi$, and the threshold follows from 
$\|\hat{\mathbf{S}}\|\leq (\Phi_\infty/\xi)(1+\sqrt{c})^2$ combined with the normalization $\xi\,(1-c\,\Phi_\infty)\geq \xi_{\min}\,(1-c\,\Phi_\infty)>0$ from the fixed-point equation~\eqref{xi}.
The bound~\eqref{eqseuil} holds \emph{almost surely} as $N,m\to\infty$;
for fixed $(N,m)$, a finite-sample correction may be needed.
In practice, $\xi$ is unknown; a consistent leave-one-out estimator is 
\[\hat{\xi} = \displaystyle \frac{1}{N}\sum_{i=0}^{N-1}\frac{1}{m}\, \mathbf{y}_i^H\,\check{\boldsymbol{\Sigma}}_{(i)}^{-1}\,\mathbf{y}_i\, ,\]
where $\check{\boldsymbol{\Sigma}}_{(i)}$ is the leave-one-out approximation of the scatter estimator with observation $i$ removed (see Section~\ref{sec::practicalestimation} for the explicit formula and efficient computation).

In practice, finite-sample exceedances above $t$ may occur with empirical whitening; Theorem~\ref{thm:consist_scm} guarantees asymptotic validity.

\paragraph{Model order estimator.} If sources are present with sufficiently high SNR, some eigenvalues of $\check{\boldsymbol{\Sigma}}_{SCM}$ exceed $t$. The estimated number of sources is defined as $\hat{p} = \#\bigl\{k : \lambda_k(\check{\boldsymbol{\Sigma}}_{SCM}) > t\bigr\}$ where $\lambda_1(\cdot) \geq \lambda_2(\cdot) \geq \cdots$ denotes the eigenvalues in decreasing order.

\medskip
\subsection{Results}
The \emph{source-free} configuration ($\mathbf{M}=\mathbf{0}$) is considered here in order to validate the theoretical noise eigenvalue upper bound \eqref{eqseuil} independently of signal subspace effects. The parameters are set to $c=0.45$, $m=900$, and $N=2000$. Thus, $\mathbf{Y} = \mathbf{C}^{1/2}\,\mathbf{X}\,\mathbf{T}^{1/2}$ with $\mathbf{C} = \mathcal{L}\left((\rho^0,\rho^1,\ldots,\rho^{m-1})^T\right)$ where $\rho=0.7$ and $\mathbf{X}$ is a zero-mean complex Gaussian noise with identity scatter matrix. The texture matrix $\mathbf{T}$ is a diagonal $N\times N$-matrix containing the $\{\tau_i\}_{i\in\{0,\ldots,N-1\}}$ on its diagonal where $\{\tau_i\}_i$ are i.i.d.\ inverse gamma distributed with mean equal to $1$ and with a shape parameter equal to $10$. The function $u$ is here defined as $u: x \mapsto \displaystyle\frac{1+\nu}{x+\nu}$ where $\nu$ is a fixed parameter equal to $0.1$ and corresponds to the Maximum Likelihood Estimation (MLE) for the chosen distribution of the $\tau$ with $\Phi_{\infty} = 1+\nu$. 

Fig.~\ref{fig:whitening_comparison} illustrates both the validity of Theorem~\ref{thm:conv_scm} and the critical role of the whitening step. In the top panel, the data $\mathbf{Y}$ have been whitened by $\check{\mathbf{C}}_{SCM}^{-1/2}$: the log-eigenvalue distribution of $\check{\boldsymbol{\Sigma}}_{SCM}$ (green) closely tracks that of the surrogate matrix $\hat{\mathbf{S}}$ (blue), in agreement with the almost sure convergence~\eqref{eqcouilletmod}. All eigenvalues remain below the RMT threshold $t$ (red vertical line), validating the source-free upper bound~\eqref{eqseuil} with $\Phi_\infty = 1+\nu$. The bottom panel shows the same data \emph{without} whitening: the eigenvalues of $\check{\mathbf{C}}_{SCM}^{1/2}\, \check{\boldsymbol{\Sigma}}_{SCM}\, \check{\mathbf{C}}_{SCM}^{1/2}$ (green) extend well beyond $t$, confirming that the threshold~\eqref{eqseuil} has no theoretical justification unless the correlation structure has been removed beforehand. Together, the two panels demonstrate that the robust $M$-estimation step effectively absorbs the impact of large texture values $\tau_i$, while the whitening step is indispensable for calibrating the eigenvalue threshold.

Fig.~\ref{fig3} presents the log-eigenvalue distributions of $\hat{\mathbf{S}}$ and $\check{\boldsymbol{\Sigma}}_{SCM}$ for samples following a different CES distribution. The texture matrix $\mathbf{T}$ is diagonal with i.i.d. entries ${\tau_i}_{i\in\{0,\ldots,N-1\}}$, where each $\tau_i=t^2$ and $t$ follows a Student-$t$ distribution with $100$ degrees of freedom and $\nu=0.1$. Compared to Fig.\ref{fig:whitening_comparison}-\textit{Left}, the eigenvalues are less concentrated around those of $\hat{\mathbf{S}}$ and move closer to the threshold $t$. As the distribution of $\tau$ deviates from the one for which $u$ corresponds to the MLE, the convergence in Theorem~\ref{thm:conv_scm} becomes slower, making the results less reliable for fixed $N$ and $m$. 

The key limitation of the SCM branch lies in the slow convergence rate of~\eqref{eqth1} when the texture distribution departs from the MLE design of $u$, as illustrated in Fig.~\ref{fig3}. This motivates Section~\ref{sec::4}, which replaces the SCM-based whitening by a Maronna $M$-estimator whitening, yielding improved finite-sample robustness at the cost of a mild additional computational overhead.

\begin{figure}
\centering\includegraphics[width=0.95\columnwidth]{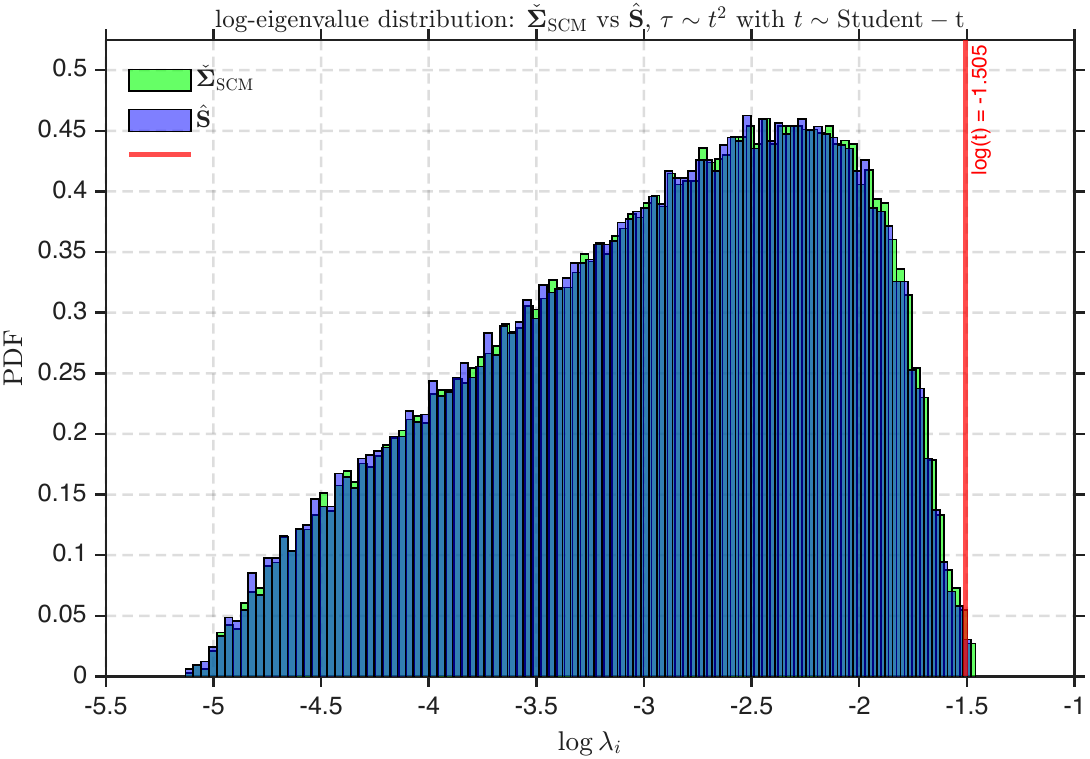}
\caption{Log-eigenvalue distribution of the scatter matrices $\check{\boldsymbol{\Sigma}}_{SCM}$ and $\hat{\mathbf{S}}$ when the signal is whitened by $\check{\mathbf{C}}_{SCM}$ and the calculated threshold ($\rho=0.7$, $m=900$, $N=2000$, $\tau\sim t^2$ with $t\sim \mathrm{Student-t}$, $\nu=0.1$).}
\label{fig3}
\end{figure}

\section{Robust Model Order Selection via $M$-Estimators}
\label{sec::4}

\subsection{Whitening Step}

Let $\widetilde{\mathbf{C}}_{FP}$ be a scaled estimator of the scatter matrix $\mathbf{C}$ defined by $\widetilde{\mathbf{C}}_{FP} = \mathcal{T}\left(\hat{\mathbf{C}}_{FP}\right)$, where $\hat{\mathbf{C}}_{FP}$ is the unique solution to the Maronna's $M$-estimator \cite{Maronna76}:
\begin{equation*}
\boldsymbol{\Sigma} = \frac{1}{N}\displaystyle\sum_{i=0}^{N-1}
u\!\left(\frac{1}{m}\,\mathbf{y}_i^H\,\boldsymbol{\Sigma}^{-1}\,\mathbf{y}_i\right)
\mathbf{y}_i\,\mathbf{y}_i^H \,.
\end{equation*}
As in the previous section, $u: [0,+\infty) \mapsto (0,+\infty)$ is nonnegative, continuous, and non-increasing. The following theorem stands for $\widetilde{\mathbf{C}}_{FP}$:

\begin{theorem}[Consistency of $\widetilde{\mathbf{C}}_{FP}$ and $\check{\mathbf{C}}_{FP}$]
\label{thm:consist_fp}
Let $\widetilde{\mathbf{C}}_{FP}$ be a fixed-point estimator of the scatter matrix $\mathbf{C}$ as defined above; the following result holds:
\begin{equation}
\left\Vert \widetilde{\mathbf{C}}_{FP} -
\mathbb{E}\left[v(\tau\,\xi)\,\tau\right]\,\mathbf{C}
\right\Vert \overset{a.s.}{\longrightarrow} 0 \,,
\label{eqth2}
\end{equation}
where $v$ and $\xi$ are defined as previously. The rescaled estimator $\check{\mathbf{C}}_{FP} = \displaystyle\frac{\widetilde{\mathbf{C}}_{FP}}{\mathbb{E}\left[v(\tau\,\xi)\,\tau\right]}$ is a consistent estimate of $\mathbf{C}$. When the texture distribution is unknown, the scaling factor $\mathbb{E}[v(\tau\,\xi)\,\tau]$ must be estimated from the data; see
Section~\ref{sec::practicalestimation}.
\end{theorem}

\begin{IEEEproof}[Proof Outline]
The proof, inspired by \cite{Vinogradova14, Vinogradova2015} and \cite{Loubaton16}, is provided in the Supplementary Material (Section~S-A.2). Write $\|\mathcal{T}(\hat{\mathbf{C}}_{FP})-\mathbb{E}[v(\tau\,\xi)\,\tau]\,\mathbf{C}\|
\le \eta_1+\eta_2$, where $\eta_1=\|\mathcal{T}(\hat{\mathbf{C}}_{FP}-\hat{\mathbf{S}})\|$
and $\eta_2=\|\mathcal{T}(\hat{\mathbf{S}})-\mathbb{E}[v(\tau\,\xi)\,\tau]\,\mathbf{C}\|$.
$\eta_1\xrightarrow{\mathrm{a.s.}}0$ because $\|\hat{\mathbf{C}}_{FP}-\hat{\mathbf{S}}\|
\xrightarrow{\mathrm{a.s.}}0$~\cite{Couillet15} and $\mathcal{T}$ is spectrally
contractive. For $\eta_2$: the bias $\mathbb{E}\left[\hat{\gamma}^{\hat{\mathbf{S}}}_m(\lambda)\right]
=\mathbb{E}[v(\tau\,\xi)\,\tau]\,\mathbf{d}_m^H(\lambda)\,\mathbf{C}\,\mathbf{d}_m(\lambda)$
vanishes identically, and the fluctuation term vanishes a.s.\ by the same
Hanson--Wright/Bernstein argument as Theorem~1 with
$\|\mathbf{D}\|_\infty\le\Phi_\infty/\xi<\infty$ replacing $\|\mathbf{T}\|_\infty$.
\end{IEEEproof}
 
\begin{remark}
The quantity $\xi$ in the scaling factor $\mathbb{E}[v(\tau\,\xi)\,\tau]$ is implicitly defined by~\eqref{xi} and depends on both the texture distribution and the function $u$. When the texture distribution is unknown, this expectation must be estimated from data (Section~\ref{sec::practicalestimation}); the resulting residual scale factor may slow convergence at finite sample sizes, while almost sure consistency is preserved. Note also that Theorem~\ref{thm:consist_scm} gives $\check{\mathbf{C}}_{SCM}\overset{a.s.}{\to}\mathbf{C}$, whereas Theorem~\ref{thm:consist_fp} gives $\widetilde{\mathbf{C}}_{FP}\overset{a.s.}{\to}\mathbb{E}[v(\tau\,\xi)\,\tau]\,\mathbf{C}$; hence $\check{\mathbf{C}}_{FP}=\widetilde{\mathbf{C}}_{FP}/\mathbb{E}[v(\tau\,\xi)\,\tau]$ also targets $\mathbf{C}$, confirming that the threshold $t$ in~\eqref{eqseuil} applies equally to $\check{\boldsymbol{\Sigma}}_{SCM}$ and $\check{\boldsymbol{\Sigma}}_{FP}$.
\end{remark}

As in the previous section, the samples $\mathbf{Y}$ can then be whitened thanks to $\check{\mathbf{C}}_{FP}^{-1/2}$. Let $\check{\mathbf{Y}}_{wFP} = [\check{\mathbf{y}}_{wF0},\ldots, \check{\mathbf{y}}_{wF\,N-1}]$ be the whitened samples:
\begin{equation*}
\check{\mathbf{Y}}_{wFP} =
\check{\mathbf{C}}_{FP}^{-1/2}\,\mathbf{M}\,\boldsymbol{\Gamma}^{1/2}\boldsymbol{\delta}^H +
\check{\mathbf{C}}_{FP}^{-1/2}\,\mathbf{C}^{1/2}\,\mathbf{X}\,\mathbf{T}^{1/2} \,.
\end{equation*}

\begin{figure*}[htbp]
\centering
\includegraphics[width=0.95\columnwidth]{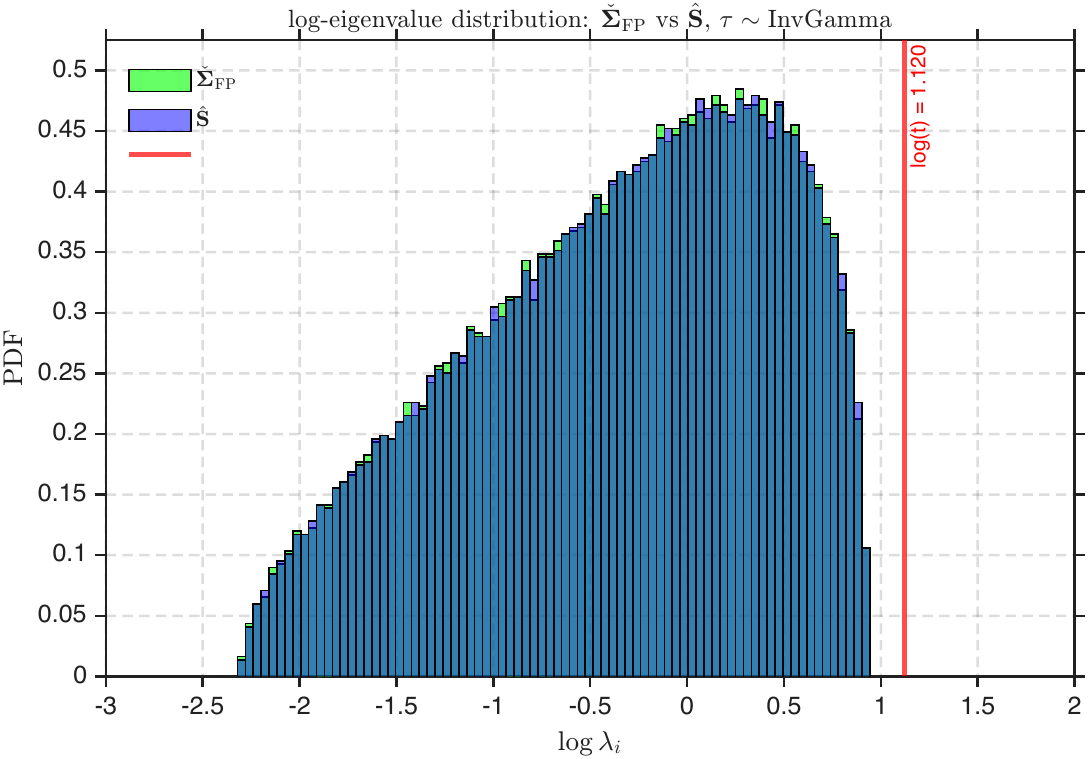}\hfill
\includegraphics[width=0.95\columnwidth]{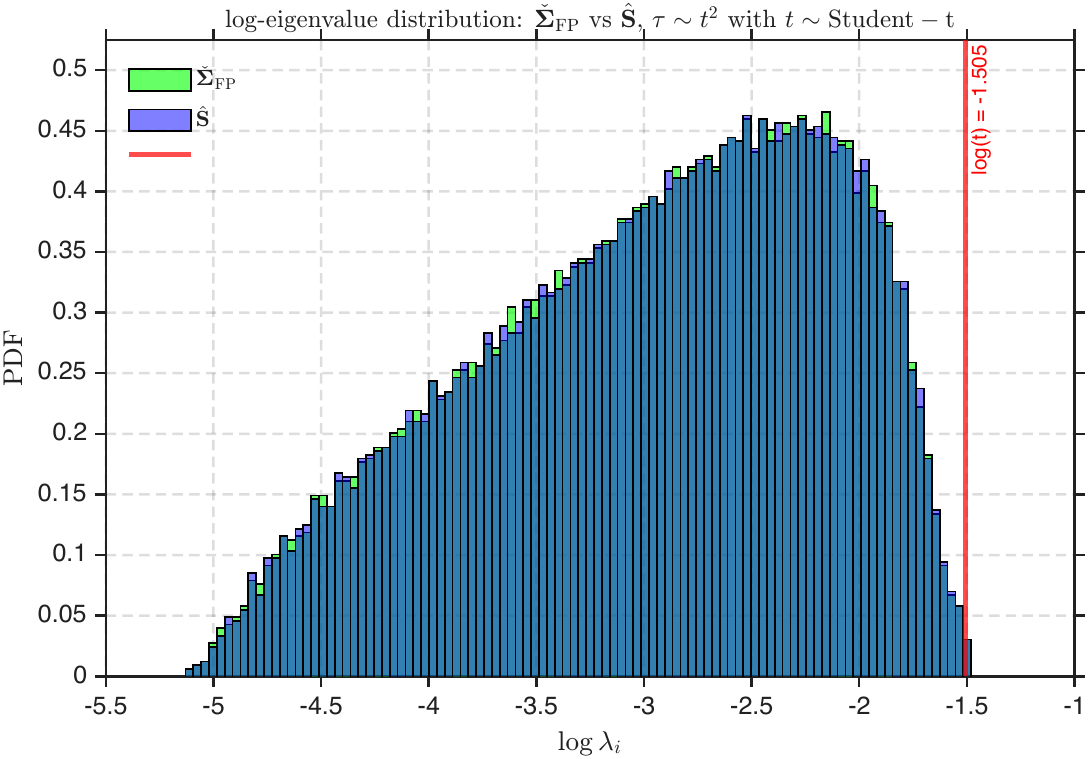}
\caption{Log-eigenvalue distributions, Maronna-based whitened ($\rho=0.7$, $m=900$, $N=2000$, $\nu=0.1$).
\textit{Left}: InvGamma textures; all eigenvalues below $t$, validating Theorem~\ref{thm:conv_fp}. \textit{Right}: Student-$t^2$ textures (100 d.o.f.); the Maronna branch maintains eigenvalue control under texture mismatch, unlike the SCM branch (cf.\ Fig.~\ref{fig3}).}
\label{fig45}
\end{figure*}

\subsection{Robust Estimation of the Scatter Matrix}

The robust estimation of the scatter matrix and the model order selection are performed as previously. The robust estimator of the scatter matrix of the whitened signal $\check{\mathbf{Y}}_{wFP}$ is a Fixed-Point estimator denoted by $\check{\boldsymbol{\Sigma}}_{FP}$ and defined as the unique solution of the equation:
\begin{equation}
\boldsymbol{\Sigma} = \frac{1}{N}\displaystyle\sum_{i=0}^{N-1} u\!\left(\frac{1}{m}\,\check{\mathbf{y}}_{wFi}^H\,\boldsymbol{\Sigma}^{-1}\, \check{\mathbf{y}}_{wFi}\right) \check{\mathbf{y}}_{wFi}\,\check{\mathbf{y}}_{wFi}^H \,.
\label{eqsigma2}
\end{equation}
Thus, $\check{\boldsymbol{\Sigma}}_{FP}$ is a robust estimator of the scatter matrix of the whitened signal. Through Theorem~\ref{thm:conv_fp} proposed below, Theorem~\ref{thm:conv_scm} extends the consistency result of $\check{\boldsymbol{\Sigma}}_{SCM}$ to the Maronna-whitened estimator $\check{\boldsymbol{\Sigma}}_{FP}$.

\begin{theorem}[Convergence of $\check{\boldsymbol{\Sigma}}_{FP}$]
\label{thm:conv_fp}
The following convergence holds:
\begin{equation}
\left\Vert \check{\boldsymbol{\Sigma}}_{FP} - \hat{\mathbf{S}} \right\Vert \overset{a.s.}{\longrightarrow} 0 \,.
\label{eqth2mod}
\end{equation}
where $\hat{\mathbf{S}}$ is defined in \eqref{matrixS}

\end{theorem}
\begin{IEEEproof}[Proof Outline]
The proof follows exactly the same perturbation argument as Theorem~\ref{thm:conv_scm} (Supplementary Material~S-B), with $\check{\mathbf{C}}_{SCM}^{-1/2}$ replaced by $\check{\mathbf{C}}_{FP}^{-1/2}$ throughout; 
Indeed, Theorem~\ref{thm:consist_fp} guarantees that $\|\check{\mathbf{C}}_{FP}-\mathbf{C}\|\overset{a.s.}{\to}0$: since $\check{\mathbf{C}}_{FP} = \widetilde{\mathbf{C}}_{FP}/ \mathbb{E}[v(\tau\,\xi)\,\tau]$ and $\mathbb{E}[v(\tau\,\xi)\,\tau]>0$ is a fixed constant (for a given $u$ and texture distribution), this follows immediately from~\eqref{eqth2}. The whitening operator $\check{\mathbf C}_{FP}^{-1/2}$ is asymptotically equivalent to the oracle whitening matrix $\mathbf C^{-1/2}$.
\end{IEEEproof}

\subsection{Model Order Selection}

The threshold $t$ given in Equation~\eqref{eqseuil} holds asymptotically, with the proof following the same perturbation argument as in Supplementary Material S-B, and can be used on the eigenvalues of $\check{\boldsymbol{\Sigma}}_{FP}$ to estimate $p$.

If sources with sufficiently high SNR are present, some eigenvalues of $\check{\boldsymbol{\Sigma}}_{FP}$ exceed the threshold $t$ defined in \eqref{eqseuil}. The estimated number of sources is $\hat{p} = \#\bigl\{k : \lambda_k(\check{\boldsymbol{\Sigma}}_{FP}) > t\bigr\}$.

\subsection{Results}
\label{sec::results}

Fig.~\ref{fig45} displays the log-eigenvalue distributions of $\hat{\mathbf{S}}$
and $\check{\boldsymbol{\Sigma}}_{FP}$ under the same simulation setup as
Section~\ref{sec::3} ($c=0.45$, $m=900$, $N=2000$, $\rho=0.7$, $\nu=0.1$),
with the same weight function $u(\cdot)$ and data whitened through
$\check{\mathbf{C}}_{FP}^{-1/2}$. In the left panel ($\mathrm{InvGamma}$ textures, $\Phi_\infty=1{+}\nu$), the eigenvalue distribution of $\check{\boldsymbol{\Sigma}}_{FP}$ closely tracks that of $\hat{\mathbf{S}}$ and all eigenvalues remain below the threshold $t$~\eqref{eqseuil}, validating Theorem~\ref{thm:conv_fp}.
The convergence is noticeably tighter than in Fig.~\ref{fig:whitening_comparison}-Top (SCM branch), consistent with the faster convergence rate of~\eqref{eqth2} relative to~\eqref{eqth1}.
The right panel ($\mathrm{Student-t^2}$ textures, $t\sim\mathrm{Student}$ with 100 degrees of freedom) shows that the Maronna-whitened estimator maintains reliable eigenvalue control even when the texture distribution departs from the MLE design of $u$---in contrast to the SCM branch (Fig.~\ref{fig3}), where the same mismatch causes eigenvalues to approach $t$. This robustness to texture distribution mismatch is a direct consequence of the down-weighting mechanism built into Maronna's $M$-estimator during the whitening step itself.

Fig.~\ref{fig6} presents the same setup as in Fig.~\ref{fig45}-\textit{left}, but with three sources having signal-to-noise ratios of $[5, 10, 15]\,$dB. It can be seen that three distinct groups of eigenvalues (10 Monte-Carlo runs) exceed the threshold, and that the noise eigenvalue distribution of $\check{\boldsymbol{\Sigma}}_{FP}$ closely matches that of $\hat{\mathbf{S}}$.

\begin{figure}[htbp]
\centering\includegraphics[width=0.95\columnwidth]{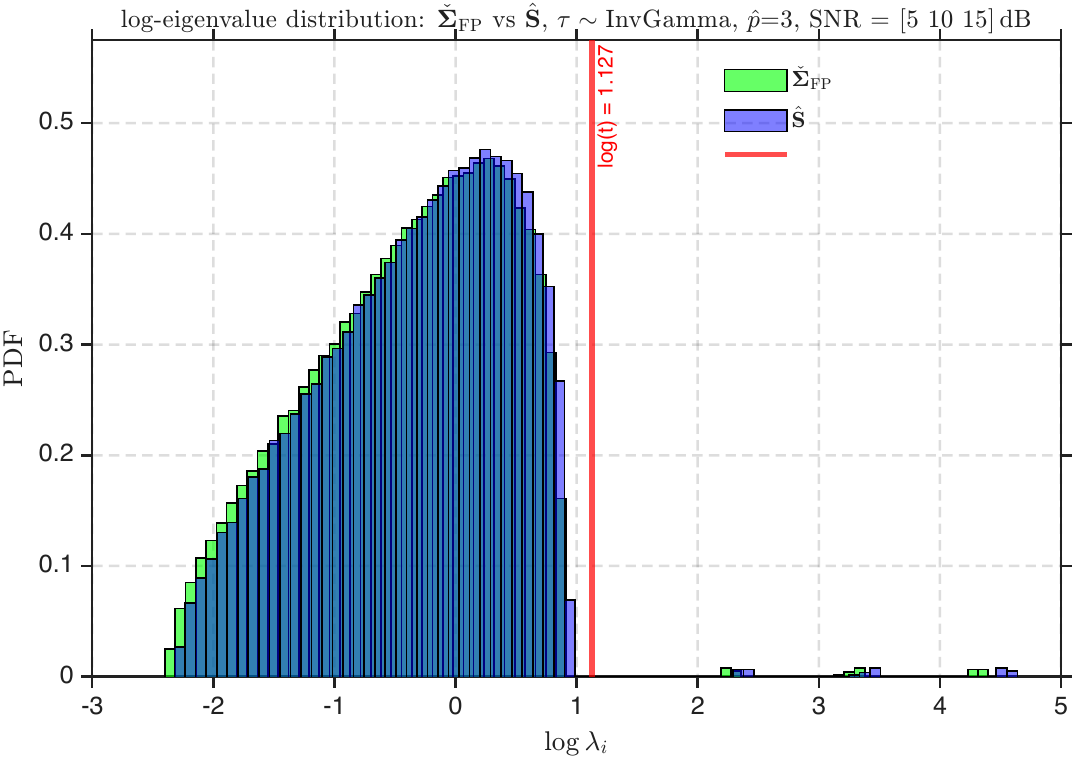}
\caption{Eigenvalues of the scatter matrices $\hat{\mathbf{S}}$ and $\check{\boldsymbol{\Sigma}}_{FP}$ for three sources with $\mathrm{SNR}=[5, 10, 15]\,$dB present in the samples and the calculated threshold ($\rho=0.7$, $m=900$, $N=2000$, $\tau\sim \mathrm{InvGamma}$). }
\label{fig6}
\end{figure}

The results obtained with the robust method improve upon those of the previous section (e.g., Fig.\ref{fig3}). The proposed approach is robust to the distribution of $\tau$: even when the texture distribution differs from that assumed in the design of $u$, the method remains reliable. This behavior stems from the robustness of the scatter matrix estimation. In contrast, without whitening, several eigenvalues exceed the threshold, preventing reliable conclusions or model order estimation. These results extend \cite{Vinogradova14, Vinogradova2015} to the case of left-correlated noise. Moreover, the $L_2$-norm of the difference between the estimated scatter matrix and the SCM tends to zero as $N,m\to\infty$ with constant ratio $c$. Since many signal subspace rank estimation methods rely on the eigenvalues of the scatter matrix, the proposed estimator improves their consistency.

While Sections~\ref{sec::3}--\ref{sec::4} require specifying a weight function $u$ tuned to the texture distribution---a requirement that may be difficult to fulfill in practice---Section~\ref{sec::5} eliminates this constraint by adopting Tyler's $M$-estimator, which is distribution-free and achieves a universal Mar\v{c}enko--Pastur threshold.

\begin{figure*}
\begin{tabular}{ccc}
\centering
\includegraphics[width=0.65\columnwidth]{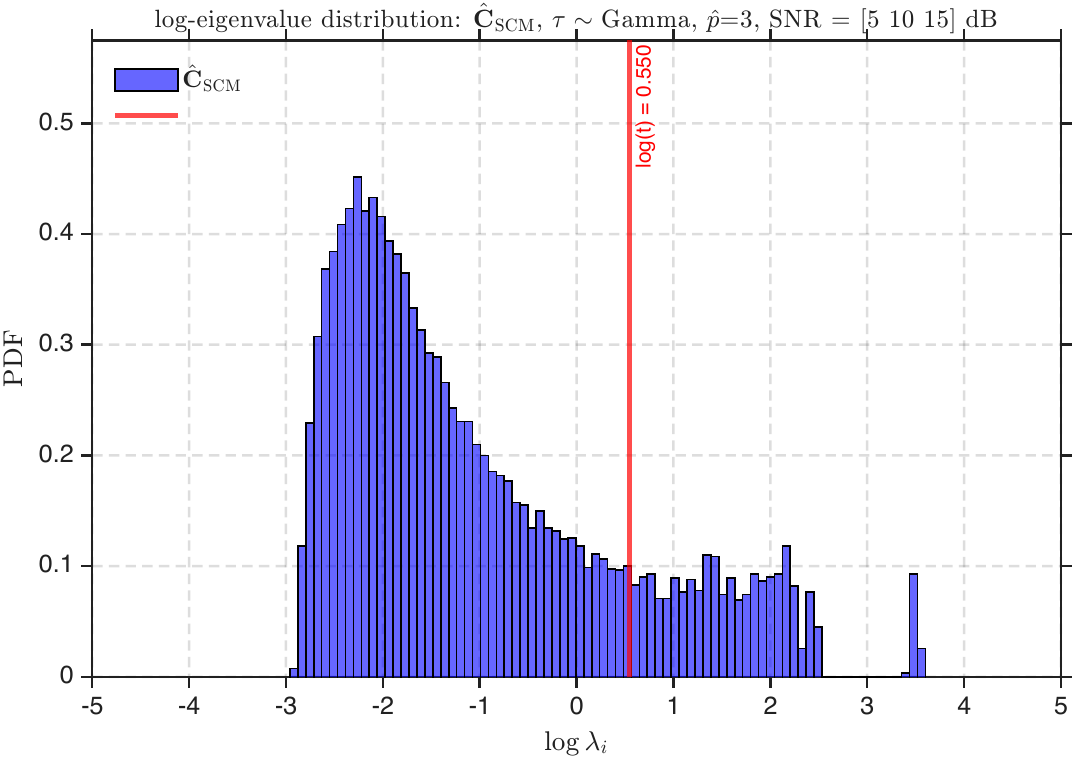}
&
\includegraphics[width=0.65\columnwidth]{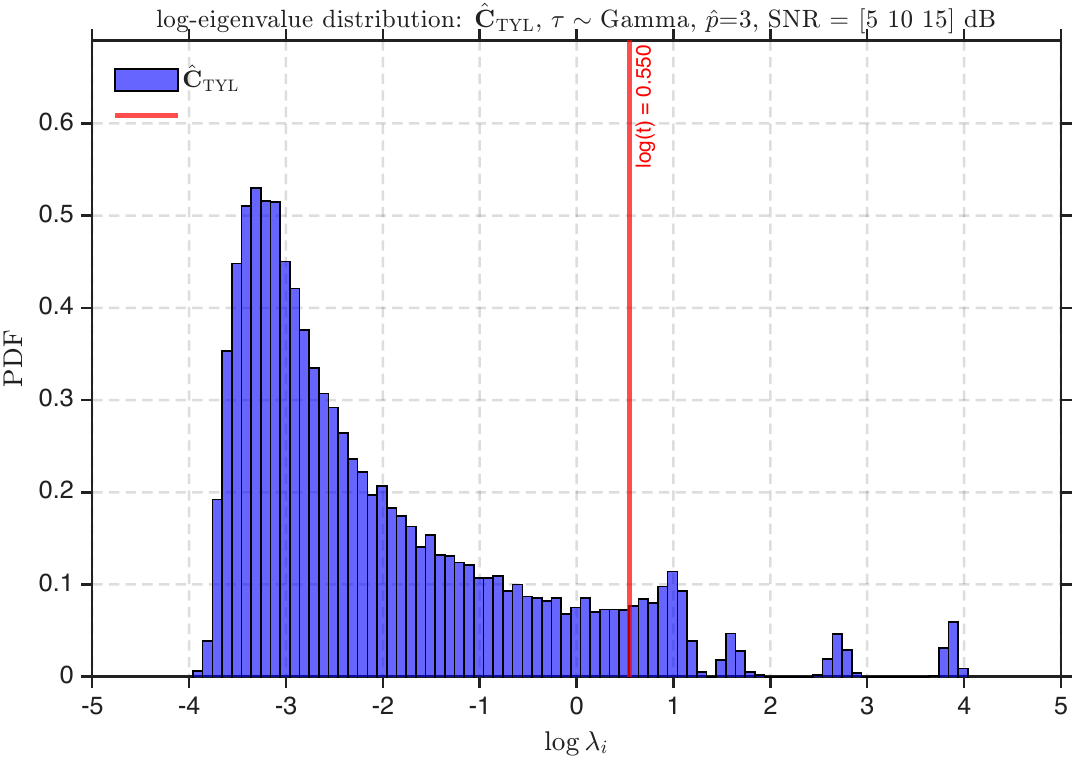}
&
\includegraphics[width=0.65\columnwidth]{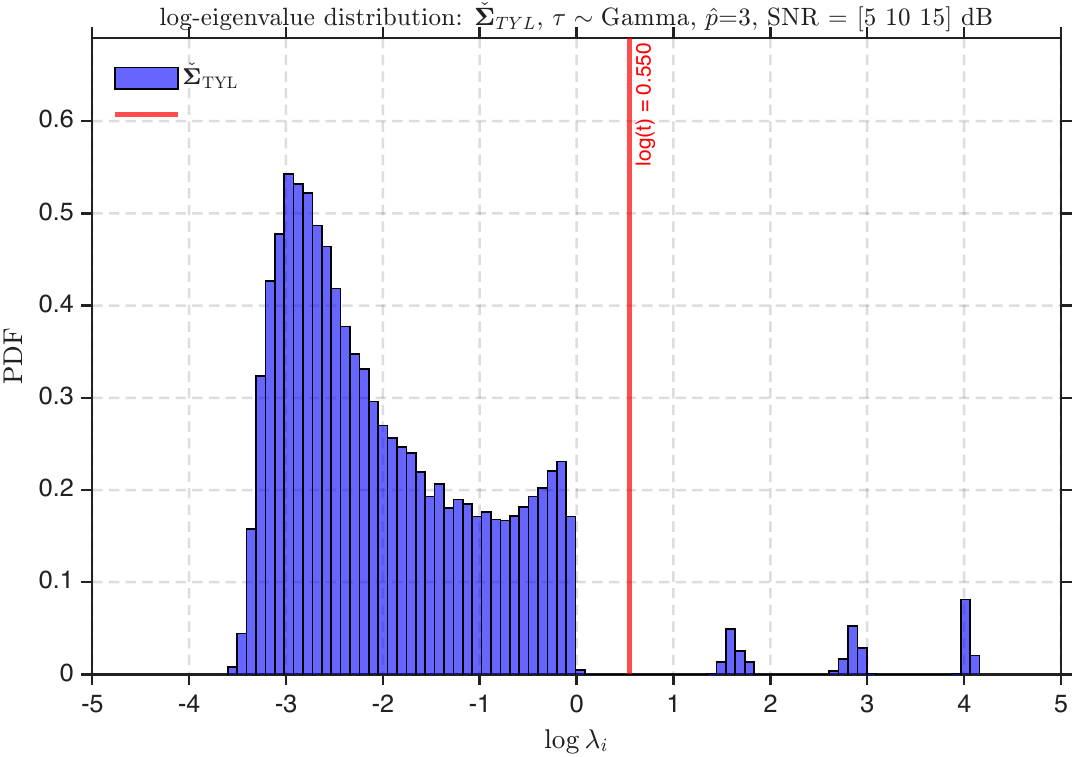}
\end{tabular}
\caption{Log-eigenvalue distributions for $m=100$, $N=1000$, $\rho=0.8$, $p=3$ sources (SNR $[5,10,15]$ dB), and K-distributed noise ($\nu=0.5$): \textit{Left}: SCM $\hat{\mathbf{C}}_{SCM}$. \textit{Middle}: Tyler’s estimator $\hat{\mathbf{C}}_{TYL}$. \textit{Right}: whitened Tyler’s estimator $\hat{\boldsymbol{\Sigma}}_{TYL}$. The Mar\v{c}enko--Pastur upper bound is $\log(\bar{\lambda}) = \log(1.7325)$.}
\label{allfig}
\end{figure*}

\section{Tyler $M$-Estimator: A Distribution-Free Approach}
\label{sec::5}

A major practical challenge in robust signal processing is the lack of prior knowledge regarding the noise texture distribution. To circumvent this, we specialize our framework to Tyler's $M$-estimator \cite{Tyler87}, which is distribution-free \cite{Ollila15}, and does not require specifying a weight function $u$. Unlike Maronna's framework of Sections~III--IV, it was proved in \cite{Zhang16} that the eigenvalue distribution of Tyler's $M$-estimator converges, in the white CES case, to the Mar\v{c}enko--Pastur law \cite{Marchenko67}; the detection threshold is therefore the upper edge $(1+\sqrt{c})^2$, which replaces~\eqref{eqseuil}. The proposed method proceeds in two steps: (i) whitening via a Toeplitz-rectified Tyler estimator of $\mathbf{C}$, and (ii) eigenvalue thresholding of the Tyler scatter matrix of the whitened signal at $(1+\sqrt{c})^2$ to infer the number of sources.

\subsection{Signal Whitening}

The first step of the procedure is to whiten the signal with a consistent estimation of $\mathbf{C}$, the scatter matrix of the CES noise. The proposed estimator $\hat{\mathbf{C}}_{TYL}$ is a Tyler $M$-estimator of the scatter matrix $\mathbf{C}$ enforced to be Toeplitz-structured with the operator $\mathcal{T} \left( \hat{\mathbf{C}}_{TYL} \right)$. The Tyler $M$-estimator $\hat{\mathbf{C}}_{TYL}$ is defined as the unique solution of \cite{Tyler87}: 
\begin{equation}
\boldsymbol{\Sigma} =  \frac{m}{N} \displaystyle \sum_{i=0}^{N-1} \frac{\mathbf{y}_i \,\mathbf{y}_i^H}{\mathbf{y}_i^H \,\boldsymbol{\Sigma}^{-1} \,\mathbf{y}_i} \, . 
\label{eqTyler}
\end{equation}
\begin{remark}
Tyler's $M$-estimator is unique only up to a positive scalar;  throughout this paper, uniqueness is enforced by the trace 
normalisation $\operatorname{tr}(\boldsymbol{\Sigma})=m$~\cite{Tyler87}.
This convention is applied consistently:
\begin{enumerate}[(i)]
  \item For $\hat{\mathbf{C}}_{TYL}$: the normalisation ensures that the unknown scale factor introduced during whitening in 
  $\hat{\mathbf{Y}}_w = \check{\mathbf{C}}_{TYL}^{-1/2}\mathbf{Y}$ 
 cancels out, since $\operatorname{tr}(\mathbf{C})=m$ holds for the true scatter matrix.
  \item For $\check{\boldsymbol{\Sigma}}_{TYL}$: the same normalisation ensures that the detection threshold $(1+\sqrt{c})^2$ corresponds exactly to the upper edge of the Mar\v{c}enko--Pastur law with ratio $c$; any other normalisation by a factor $\alpha \neq 1$ would shift 
  the threshold to $\alpha(1+\sqrt{c})^2$.
\end{enumerate}
\end{remark}

 This equation requires an iterative approach to compute $\hat{\mathbf{C}}_{TYL}$ as the final value of $\boldsymbol{\Sigma}$. Note that the $\mathbf{y}_i$ also contain the sources. The consistency of  $\mathcal{T} \left( \hat{\mathbf{C}}_{TYL} \right)$ is proven thanks to the following theorem: 
\begin{theorem}[Consistency of $\check{\mathbf{C}}_{TYL}$]
\label{thm:consist_tyl}
Under Assumptions~\ref{ass:regime}--\ref{ass:signal} (Assumption~\ref{ass:model}-(v) is not required), we have the following convergence:
 \begin{equation}
 \left\Vert \mathcal{T} \left( \hat{\mathbf{C}}_{TYL} \right) - \mathbf{C} \right\Vert \overset{a.s.}{\longrightarrow} 0 \, .
 \end{equation}
\end{theorem}
\begin{IEEEproof} See \cite{TerreauxICASSP18}. The proof follows the same four-term decomposition as Theorem~1, with two differences: the weight function is $u(x)=m/x$
(Tyler's choice), so the surrogate matrix has oracle weights $v(\tau_i\,\xi)\equiv 1/\xi$
that are bounded independently of the texture distribution, removing
Assumption~2-(v); and the signal term vanishes by the same absolute summability
argument (Assumption~3-(i)).
\end{IEEEproof}

Set $\mathbf{\check{C}}_{TYL}=\mathcal{T} \left( \hat{\mathbf{C}}_{TYL}\right)$. The Tyler-branch whitened data are:
\begin{align}
\hat{\mathbf{Y}}_w
    &= \mathbf{\check{C}}_{TYL}^{-1/2}\,\mathbf{Y} \\
    &= \mathbf{\check{C}}_{TYL}^{-1/2}\,\mathbf{M}\,\boldsymbol{\Gamma}^{1/2}\, \boldsymbol{\delta}^H + \mathbf{\check{C}}_{TYL}^{-1/2} \,\mathbf{C}^{1/2}\,\mathbf{X}\,\mathbf{T}^{1/2}.
\end{align}
with $\hat{\mathbf{Y}}_w = [ \, \hat{\mathbf{y}}_{w0}, \,  ... , \, \hat{\mathbf{y}}_{wN-1}]$.

\subsection{Estimation of the Tyler Scatter Matrix}

The signal being whitened, it is now possible to apply a Tyler $M$-estimator to threshold its eigenvalues. Let $\check{\boldsymbol{\Sigma}}_{TYL}$ denote a Tyler estimation of the scatter matrix of the whitened data $\hat{\mathbf{Y}}_w$, that is, $\check{\boldsymbol{\Sigma}}_{TYL}$ is the unique solution if it exists of: 
\begin{equation}
\boldsymbol{\Sigma} = \frac{m}{N} \displaystyle \sum_{i=0}^{N-1} \frac{\hat{\mathbf{y}}_{wi} \,\hat{\mathbf{y}}_{wi} ^H}{\hat{\mathbf{y}}_{wi} ^H \,\boldsymbol{\Sigma}^{-1} \, \hat{\mathbf{y}}_{wi} } \, .
\label{eqTylerbis}
\end{equation}
Existence and uniqueness of~\eqref{eqTylerbis} require $m < N$ (i.e., $c_N < 1$), which is guaranteed by Assumption~\ref{ass:regime} (specifically, $c < 1$, which may be assumed without loss of generality in the large-dimensional regime; no constraint from Assumption~\ref{ass:model}-(v) is needed here), as well as the almost sure general position of the whitened observations $\hat{\mathbf{y}}_{wi}$ in $\mathbb{C}^m$; since 
the entries of $\mathbf{X}$ have a density under Assumption~\ref{ass:model}-(iv), any $m$ of them are linearly independent almost surely~\cite{Tyler87}.

\subsection{Eigenvalue Thresholding via the Mar\v{c}enko--Pastur Law}  

\begin{theorem}[Convergence of $\check{\boldsymbol{\Sigma}}_{TYL}$]
\label{thm:tyler_est}
Let us define $\displaystyle \hat{\mathbf{S}}_w = \displaystyle \frac{1}{N} \,\mathbf{X} \,\mathbf{X}^H$. Under the same assumptions as previously,
\begin{equation}
\left \Vert \check{\boldsymbol{\Sigma}}_{TYL} - \hat{\mathbf{S}}_w \right\Vert \overset{a.s.}{\longrightarrow} 0\, .
\label{eqth_tyler}
\end{equation}
\end{theorem}
\begin{proof}
See \cite{TerreauxICASSP18}. By Theorem~5 and the same Lipschitz argument as Theorem~2,
$\|\check{\boldsymbol{\Sigma}}_{TYL}-\frac{1}{N}\mathbf{X}\mathbf{X}^H\|
\xrightarrow{\mathrm{a.s.}}0$; the Mar\v{c}enko--Pastur upper edge
$(1+\sqrt{c})^2$ then follows by the stability of ESD convergence under
spectral-norm perturbations.
 \end{proof}

Since $\hat{\mathbf{S}}_w = \displaystyle\frac{1}{N}\,\mathbf{X}\,\mathbf{X}^H$, the Mar\v{c}enko--Pastur theorem~\cite{Marchenko67} guarantees that the Empirical Spectral Distribution (ESD) of $\hat{\mathbf{S}}_w$ converges weakly to the Mar\v{c}enko--Pastur law with ratio $c$. Theorem~\ref{thm:tyler_est} gives $\left\|\check{\boldsymbol{\Sigma}}_{TYL} - \hat{\mathbf{S}}_w\right\| \overset{a.s.}{\to} 0$; since ESD convergence is continuous with respect to spectral-norm perturbations, the ESD of $\check{\boldsymbol{\Sigma}}_{TYL}$ converges to the same Mar\v{c}enko--Pastur law.

Eigenvalues of $\check{\boldsymbol{\Sigma}}_{TYL}$ that exceed $t=(1+\sqrt{c})^2$ are therefore attributed to signal components. Note that this threshold assumes the normalisation $\operatorname{tr}\left(\check{\boldsymbol{\Sigma}}_{TYL}\right)=m$, consistent with~\eqref{eqTyler} and~\eqref{eqTylerbis}.

If sources are present with sufficiently high SNR, some eigenvalues of $\check{\boldsymbol{\Sigma}}_{TYL}$ exceed $t$. The estimated number of sources is $\hat{p} = \#\bigl\{k : \lambda_k(\check{\boldsymbol{\Sigma}}_{TYL}) > t\bigr\}$.

\subsection{Simulations}
The efficiency of the whitening process is illustrated in Fig.~\ref{allfig}. We simulate $N=1000$ observations of dimension $m=100$ from a correlated $K$-distributed process with shape parameter $\nu=0.5$ and Toeplitz covariance matrix coefficient $\rho=0.8$, whose $(i,j)$-th entry is $\rho^{|i-j|}$ for $i,j=1,\ldots,m$. We then embed $p=3$ informative sources with SNR $=[5,10,15]\,$dB into the non-Gaussian correlated noise and compare the resulting eigenvalue distributions with the Mar\v{c}enko--Pastur upper bound. The eigenvalues are computed from: i)-\textit{Left} the sample covariance matrix $\hat{\mathbf{C}}_{SCM}$, ii)-\textit{Middle} the Tyler $M$-estimator $\hat{\mathbf{C}}_{TYL}$, and iii)-\textit{Right} the Tyler $M$-estimator $\check{\boldsymbol{\Sigma}}_{TYL}$ applied to whitened observations. The results show that the $p=3$ factors (10 Monte Carlo trials) are clearly identifiable only after whitening. 

\medskip
The following section validates the three branches on simulated and real-world data and compares them against baseline methods.

\section{Numerical Results and Comparisons}
\label{sec:applications}
The theoretical guarantees of Sections~\ref{sec::3}--\ref{sec::5} are validated through Monte Carlo simulations (Section~\ref{sec:simulations}), real hyperspectral images (Section~\ref{sec::Hyperspectral}), EEG source detection (Section~\ref{sec::EEG}), and portfolio optimization (Section~\ref{sec::finance}).

\begin{remark}
Although the paper is stated in the complex domain $\mathbb{C}^m$, all results extend verbatim to the real-valued case $\mathbb{R}^m$ by replacing Hermitian transposes $(\cdot)^H$ with transposes $(\cdot)^T$, complex Gaussian distributions $\mathcal{CN}(\mathbf{0}, \mathbf{I})$ with real Gaussian distributions $\mathcal{N}(\mathbf{0}, \mathbf{I})$, and the spectral integration domain $[0,2\pi)$ with $[0,\pi]$. The asymptotic eigenvalue thresholds~\eqref{eqseuil} and $(1+\sqrt{c})^2$ are unchanged; finite-sample fluctuations around these thresholds follow the Tracy-Widom$_1$ law in the real case rather than Tracy--Widom$_2$~\cite{Zhang16}, but this does not affect the almost sure consistency results. The EEG and finance applications fall in this real-valued setting.
\end{remark}

\subsection{Data-Driven Estimation of the Scaling Parameters}
\label{sec::practicalestimation}

In practice, neither the textures $\{\tau_i\}$ nor $\xi$ are directly accessible. We describe here how both $\xi$ and the scaling factor $\mathbb{E}[v(\tau\,\xi)\,\tau]$---required to form the consistent estimator $\check{\mathbf{C}}_{FP} = \widetilde{\mathbf{C}}_{FP}/\mathbb{E}[v(\tau\,\xi)\, \tau]$ of the true scatter matrix $\mathbf{C}$---can be estimated from the data alone.

\medskip
\subsubsection{Empirical estimator of $\xi$}

An empirical estimator of $\xi$ based on the leave-one-out approximation~\cite{Couillet15b} is used:
\begin{equation}
    \hat{\xi} = \frac{1}{N}\sum_{i=0}^{N-1}
\frac{1}{m}\,\mathbf{y}_i^H\,\check{\boldsymbol{\Sigma}}^{-1}_{(i)}\, \mathbf{y}_i \,,
  \label{eq:xi_loo}
\end{equation}
where $\check{\boldsymbol{\Sigma}}_{(i)}$ denotes the rank-one update of $\check{\boldsymbol{\Sigma}}_{FP}$ with observation $i$ removed:
\begin{equation}
\check{\boldsymbol{\Sigma}}_{(i)} = \check{\boldsymbol{\Sigma}}_{FP} - 
\frac{1}{N}\,u\!\left(\frac{1}{m}\,\mathbf{y}_i^H\, 
\check{\boldsymbol{\Sigma}}_{FP}^{-1}\,\mathbf{y}_i\right)
\mathbf{y}_i\,\mathbf{y}_i^H\, .
  \label{eq:Sigma_loo}
\end{equation}
It is shown in~\cite{Couillet15b} that $\xi - \hat{\xi} \overset{a.s.}{\longrightarrow} 0$.

Computing~\eqref{eq:xi_loo} naively would require inverting $\check{\boldsymbol{\Sigma}}_{(i)}$ for each $i$, at a cost of $O(Nm^3)$. Applying the matrix inversion lemma to~\eqref{eq:Sigma_loo} instead yields the closed-form expression
\begin{equation}
  \frac{1}{m}\,\mathbf{y}_i^H\,\check{\boldsymbol{\Sigma}}_{(i)}^{-1}\,
  \mathbf{y}_i
  \;=\; t_i
  \;+\; \frac{u(t_i)\,m^2\,t_i^2}{N - u(t_i)\,m\,t_i}\, ,
  \label{eq:SM_xi}
\end{equation}
where $t_i = \displaystyle\frac{1}{m}\,\mathbf{y}_i^H\,\check{\boldsymbol{\Sigma}}_{FP}^{-1}
\,\mathbf{y}_i$. 

The estimator~\eqref{eq:xi_loo} therefore requires a single matrix inversion $\check{\boldsymbol{\Sigma}}_{FP}^{-1}$, followed by $N$ scalar corrections via~\eqref{eq:SM_xi}, at a total cost of $O(m^3 + Nm^2)$.

\medskip
\subsubsection{Estimation of $\mathbb{E}[v(\tau\,\xi)\,\tau]$}
A key observation underlying this estimator is that $t_i$ converges almost surely to $\tau_i\,\xi$.
This can be established in four steps:

\medskip
\paragraph*{Step 1.}
The surrogate matrix $\hat{\mathbf{S}}$ defined in~\eqref{matrixS} admits the rank-one decomposition
\[
  \hat{\mathbf{S}} = \hat{\mathbf{S}}_{(i)}
  + \frac{v(\tau_i\,\xi)}{N}\,\mathbf{y}_{wi}\,\mathbf{y}_{wi}^H,
\]
where $\hat{\mathbf{S}}_{(i)}$ is the leave-one-out version of $\hat{\mathbf{S}}$ with observation $i$ removed. Applying the matrix inversion lemma to $\hat{\mathbf{S}}^{-1}$ yields
\begin{equation}
  t_i
  = \frac{\dfrac{1}{m}\,\mathbf{y}_{wi}^H\,
           \hat{\mathbf{S}}_{(i)}^{-1}\,\mathbf{y}_{wi}}
         {1 + \dfrac{v(\tau_i\xi)}{N}\,
              \dfrac{1}{m}\,\mathbf{y}_{wi}^H\,
              \hat{\mathbf{S}}_{(i)}^{-1}\,\mathbf{y}_{wi}}.
  \label{eq:ti_SM}
\end{equation}

\medskip
\paragraph*{Step 2.}
Since observation $i$ is absent from $\hat{\mathbf{S}}_{(i)}$, the speckle vector $\mathbf{x}_i$ (with $\mathbf{y}_{wi}=\sqrt{\tau_i}\,\mathbf{x}_i$) is independent of $\hat{\mathbf{S}}_{(i)}$. The Hanson--Wright inequality (Lemma~S-2 in the Supplementary Material; see also \cite{Vershynin2018,Rudelson2013}) gives, for any $\varepsilon>0$,
\begin{equation}
\begin{aligned}
\mathbb{P}\Biggl(
\biggl|
\frac{1}{m}\,\mathbf{x}_i^{H}\,\hat{\mathbf S}_{(i)}^{-1}\,\mathbf{x}_i
-\frac{1}{m}\operatorname{tr}\!\bigl(\hat{\mathbf S}_{(i)}^{-1}\bigr)
\biggr|
>\varepsilon
\Biggr)
\le
2\exp\!\Biggl(
-\frac{C\,m\,\varepsilon^2}
{\|\hat{\mathbf S}_{(i)}^{-1}\|^2}
\Biggr).
\end{aligned}
\end{equation}
Since $\|\hat{\mathbf{S}}_{(i)}^{-1}\|$ is almost surely bounded, the right-hand side is summable in $m$. The Borel--Cantelli lemma therefore implies
\begin{equation}
  \frac{1}{m}\,\mathbf{y}_{wi}^H\,\hat{\mathbf{S}}_{(i)}^{-1}\,
  \mathbf{y}_{wi}
  = \frac{\tau_i}{m}\,\mathbf{x}_i^H\,\hat{\mathbf{S}}_{(i)}^{-1}
  \,\mathbf{x}_i
  \overset{a.s.}{\longrightarrow}
  \frac{\tau_i}{m}
  \operatorname{tr}\!\left(\hat{\mathbf{S}}_{(i)}^{-1}\right).
  \label{eq:ti_LLN}
\end{equation}

\medskip
\paragraph*{Step 3.}
By Theorem~1 of~\cite{Couillet15b}, the normalized trace of $\hat{\mathbf{S}}^{-1}_{(i)}$ converges almost surely:
\begin{equation}
  \frac{1}{m}
  \operatorname{tr}\!\left(\hat{\mathbf{S}}^{-1}_{(i)}\right)
  \overset{a.s.}{\longrightarrow}
  \xi,
  \label{eq:trace_xi}
\end{equation}
where $\xi$ is the unique positive solution of~\eqref{xi}.

\medskip
\paragraph*{Step 4.}
Substituting~\eqref{eq:ti_LLN} and~\eqref{eq:trace_xi} into~\eqref{eq:ti_SM}, and using $\psi(x)=x\,v(x)$, the denominator converges to $1+c\,\psi(\tau_i\,\xi)> 0$ (since $v$ is bounded and continuous). The two factors cancel, giving $t_i \overset{a.s.}{\longrightarrow} \tau_i\,\xi$. Since $t_i \overset{a.s.}{\longrightarrow} \tau_i\,\xi$, we have $v(t_i) \overset{a.s.}{\longrightarrow} v(\tau_i\,\xi)$ and $t_i\,v(t_i) \overset{a.s.}{\longrightarrow} \xi\,v(\tau_i\,\xi)\,\tau_i$,
so that
\begin{equation}
\widehat{\rho} \;\triangleq\;
  \frac{1}{\hat{\xi}}\,\frac{1}{N}\sum_{i=0}^{N-1} \psi(t_i)
  \;\overset{a.s.}{\longrightarrow}\; \mathbb{E}[v(\tau\,\xi)\,\tau]
  \triangleq \rho\,,
  \label{eq:Evtau_est}
\end{equation}
where $t_i = \displaystyle\frac{1}{m}\,\mathbf{y}_i^H\,\check{\boldsymbol{\Sigma}}_{FP}^{-1}
\,\mathbf{y}_i$. 

\begin{figure*}[htbp]
\centering
\includegraphics[width=0.9\columnwidth]{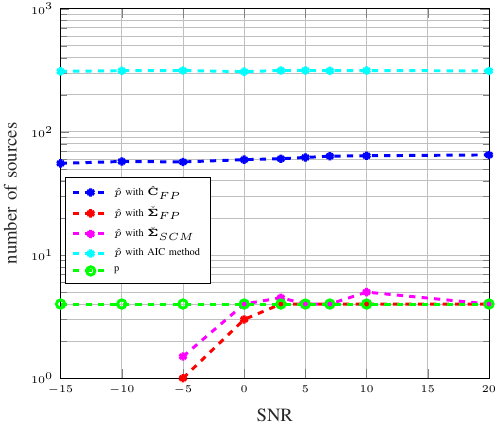}\hfill
\includegraphics[width=0.9\columnwidth]{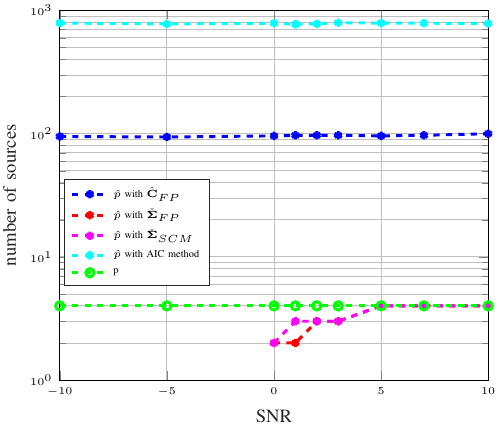}
\caption{Estimated $\hat{p}$ (log scale, 4 Monte-Carlo trials) vs.\ SNR for $p=4$ sources in correlated CES noise ($m=400$, $c=0.2$, $\rho=0.7$). \textit{Left}: inverse-gamma textures. \textit{Right}: Student-$t^2$ textures. The dashed line marks the true $p=4$.}
\label{fig:simul}
\end{figure*}

\begin{figure*}[h]
\centering
\begin{tabular}{cc}
\includegraphics[width=0.95\columnwidth]{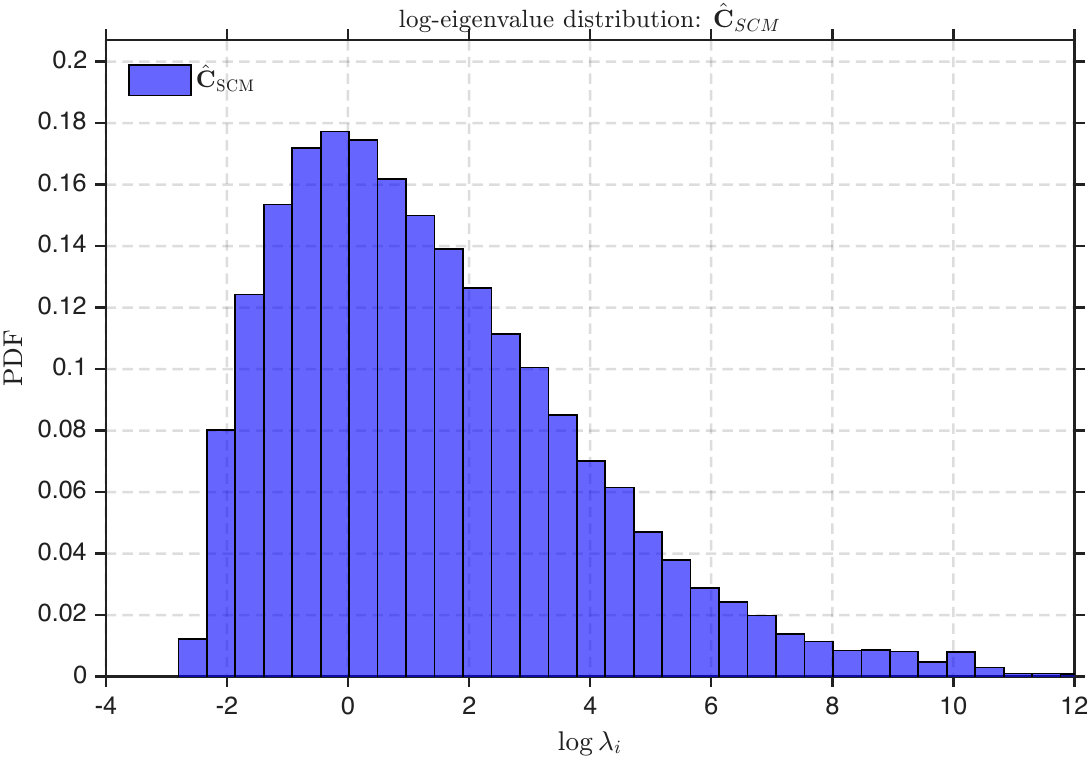}
&
\includegraphics[width=0.95\columnwidth]{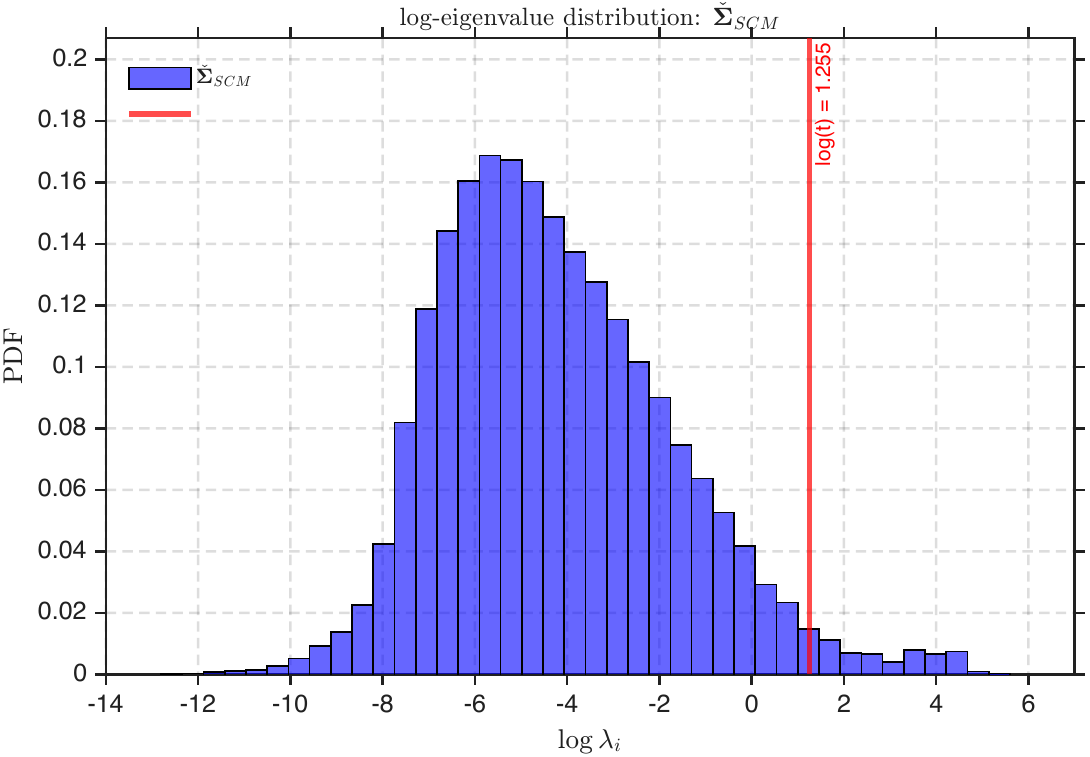} \\
\includegraphics[width=0.95\columnwidth]{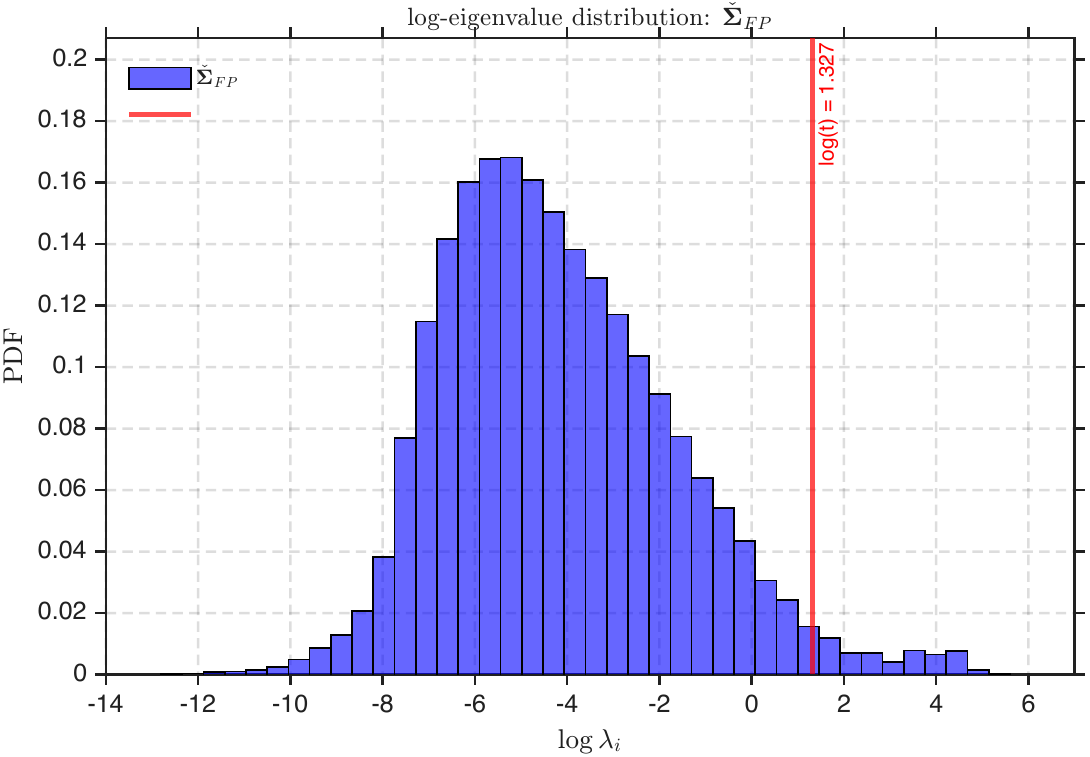}
&
\includegraphics[width=0.95\columnwidth]{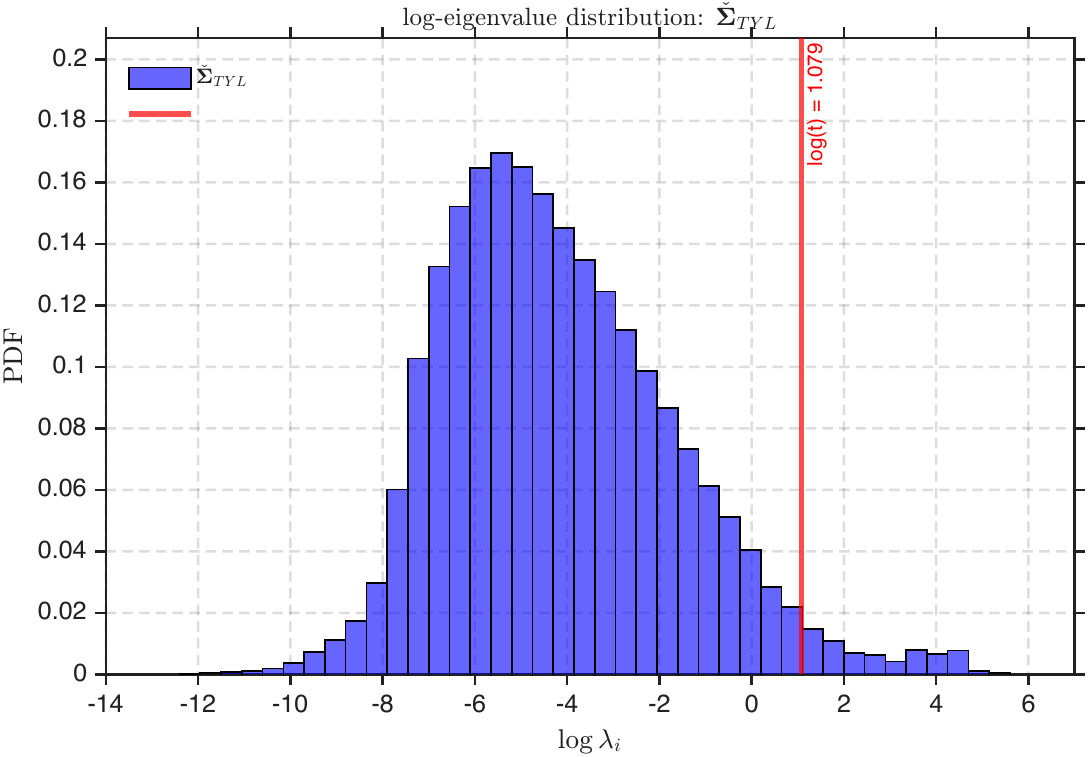}
\end{tabular}
\caption{Log-eigenvalue distributions of $\hat{\mathbf{C}}_{\mathrm{SCM}}$, 
$\check{\boldsymbol{\Sigma}}_{\mathrm{SCM}}$, $\check{\boldsymbol{\Sigma}}_{\mathrm{FP}}$, 
and $\check{\boldsymbol{\Sigma}}_{\mathrm{TYL}}$ for the MAMEM EEG SSVEP Dataset~I ($m=256$, $N=500$, $c=0.512$). The RMT thresholds $\log(t)$ differ across branches (SCM: $1.26$, Maronna: $1.33$, Tyler: $1.08$) and are shown as red vertical lines in each panel.}
\label{HS1bis}
\end{figure*}

\subsection{Simulated Data}
\label{sec:simulations}

We evaluate the proposed approach through Monte Carlo simulations by comparing four methods: M1 (SCM-based whitening, $\check{\boldsymbol{\Sigma}}_{SCM})$, M2 (Maronna-based whitening, $\check{\boldsymbol{\Sigma}}_{FP}$), M3 (direct application of Maronna’s M-estimator $\hat{\mathbf{C}}_{FP}$ to the unwhitened observations $\mathbf{Y}$), and the classical AIC criterion.

Simulations use $m=400$, $N=2000$ ($c=0.2$), Toeplitz correlation $\rho=0.7$, $p=4$ sources with SNR $\in[-15,20]$\,dB, and CES noise with inverse-gamma or Student-$t^2$ textures. Fig.~\ref{fig:simul} shows that M1 and M2 detect sources from $\approx{-5}$\,dB, with M2 achieving better accuracy at high SNR; both consistently outperform M3 and AIC across all SNR values.

\subsection{Hyperspectral Images Application}
\label{sec::Hyperspectral}

We use Indian Pines and SalinasA (AVIRIS) and PaviaU (ROSIS) \cite{siteinternet}. 

Results are in Table~\ref{table-Tableau1}. M1 and M2 dramatically outperform M3 and AIC. For PaviaU and Cars, M1/M2 underestimate $p$ due to (i) some sources having SNR below the threshold $t$, and/or (ii) the empirical estimate of $\mathbb{E}[v(\tau\,\xi)\,\tau]$ deviating from its true value due to the unknown texture distribution. 

\begin{remark}
The severe underestimation for PaviaU ($\hat{p}=1$, $p=9$) is attributable to the very large $N/m$ ratio ($207400/103 \approx 2014$), which pushes nearly all noise eigenvalues well below $t$, while the source SNRs for this dataset are comparably low. Moreover, when the texture distribution is unknown, the empirical estimator $\widehat{\rho}$ of $\mathbb{E}[v(\tau\,\xi)\,\tau]$, defined in~\eqref{eq:Evtau_est}, may differ from the true value, introducing a residual scale factor that shifts the threshold~$t$.
\end{remark}

\begin{table}[t]
\caption{Estimated number of sources $\hat{p}$ for four benchmark hyperspectral datasets. M1 and M2 (whitened) consistently outperform M3 (unwhitened) and AIC; the underestimation for PaviaU/Cars is discussed in the text.}
\label{table-Tableau1}
\begin{center}
\renewcommand{\arraystretch}{0.9}
{\setlength{\tabcolsep}{0.25cm}
\begin{tabular}{c|cccc}
\hline
Images & Indian Pines & SalinasA & PaviaU & Cars \\
$N$, $m$ & $21025$, $224$ & $7138$, $224$ & $207400$, $103$ & $40200$, $167$ \\
\hline
$p$         & 16  & 9   & 9   & 6   \\
$\hat{p}$ M1 & 11 & 9   & 1   & 3   \\
$\hat{p}$ M2 & 12 & 9   & 1   & 3   \\
$\hat{p}$ M3 & 220 & 204 & 103 & 1  \\
$\hat{p}$ AIC & 219 & 203 & 102 & 143 \\
\hline
\end{tabular}}
\end{center}
\end{table}

\subsection{Estimation of the Model Order: EEG data}
\label{sec::EEG}

To evaluate the proposed framework in a realistic high-dimensional setting, we consider the publicly available MAMEM EEG SSVEP Dataset~I \cite{moustakas2017mamem,mamemfigshare}. The dataset consists of high-density EEG recordings acquired from healthy subjects using \(m=256\) electrodes sampled at 250 Hz during steady-state visually evoked potential (SSVEP) experiments. Covariance matrices are estimated from non-overlapping 2-second windows (\(N=500\) samples), yielding a dimensional ratio \(c=m/N=0.512\), representative of modern high-dimensional EEG applications.

EEG signals exhibit spatial correlation, non-Gaussianity, and heavy-tailed fluctuations due to physiological and measurement artifacts. This motivates the use of robust covariance estimators and CES-type models, which have been shown to provide more accurate statistical descriptions of EEG observations in the presence of heavy-tailed fluctuations~\cite{Engemann2015,Li2024}.Although the Toeplitz structure only approximates EEG spatial covariance matrices, it provides an effective regularization for covariance estimation in short observation windows.

Each temporal segment is analyzed using the proposed SCM-, Maronna-, and Tyler-based estimators (Section~\ref{sec::5}), together with AIC, MDL~\cite{wax1985detection}, and the Kritchman--Nadler (KN) detector~\cite{kritchman2009non}. Table~\ref{tab:realdata} reports the empirical mean and variance of the estimated model order \(\hat p\).

\begin{table}[htbp]
\centering
\caption{Mean and variance of the estimated number of sources over multiple EEG temporal windows.}
\label{tab:realdata}
\begin{tabular}{|c|c|c|}
\hline
Method & Mean of $\hat p$ & Variance of $\hat p$ \\
\hline
$\check{\mathbf{\Sigma}}_{\mathrm{SCM}}$ & 6.80 & 1.50 \\
$\check{\mathbf{\Sigma}}_{\mathrm{FP}}$ & 6.75 & 1.56 \\
$\check{\mathbf{\Sigma}}_{\mathrm{TYL}}$ & 7.80 & 1.96 \\
AIC & 254.99 & 0.01 \\
MDL & 254.96 & 0.04 \\
KN & 197.98 & 50.68 \\
\hline
\end{tabular}
\end{table}

Figure~\ref{HS1bis} displays the eigenvalue distributions of the sample covariance matrix \(\hat{\mathbf{C}}_{\mathrm{SCM}}\) and the whitened covariance matrices \(\check{\mathbf{\Sigma}}_{\mathrm{SCM}}\), \(\check{\mathbf{\Sigma}}_{\mathrm{FP}}\), and \(\check{\mathbf{\Sigma}}_{\mathrm{TYL}}\). Despite the moderate sample size, the proposed asymptotic thresholds remain effective for \(c=0.512\).

For the \(K=235\) temporal windows, AIC and MDL systematically estimate a model order close to the observation dimension (\(\hat p \approx 255\)), while KN still yields unrealistically large estimates with high variability. These results reflect the mismatch between the assumptions underlying these methods and the correlated, non-Gaussian nature of EEG observations. In contrast, the proposed robust approaches provide stable estimates between 6.7 and 7.8 sources with variances below 2. As shown in Fig.~\ref{HS1bis}, robust whitening compresses the bulk eigenvalue distribution while preserving a small number of dominant outliers, thereby improving signal/noise subspace separation. Overall, the results demonstrate the effectiveness of robust covariance whitening for model-order estimation in high-dimensional EEG data.

\subsection{Estimation of the Model Order: Portfolio Optimization Application}
\label{sec::finance}

The proposed model-order selection framework has been successfully applied to portfolio optimization problems. In particular, \cite{Jay2018, Jay2020} show that replacing standard covariance-based factor selection with the proposed robust approach improves both maximum diversification and global minimum variance portfolios, leading to lower out-of-sample risk and improved risk-adjusted performance.

\medskip
\subsubsection{Financial factor model}

Let \(\mathbf{Y}=[\mathbf{y}_0,\ldots,\mathbf{y}_{N-1}]\) denote the \(m\times N\) matrix of asset returns. We assume that returns follow a \(p\)-factor model
\[
\mathbf{y}_t=\mathbf{M}\,\mathbf{s}_t+\sqrt{\tau_t}\,\mathbf{C}^{1/2}\,\mathbf{x}_t\, ,
\]
where \(\mathbf{M}\) is the factor loading matrix, \(\mathbf{s}_t\) contains the latent factors, \(\mathbf{x}_t\) is a zero-mean unitarily invariant random vector with \(\|\mathbf{x}_t\|^2=1\), \(\mathbf{C}\) is a Toeplitz-structured scatter matrix, and \(\tau_t\) models time-varying volatility. This formulation accounts for both cross-sectional dependence and heavy-tailed fluctuations commonly observed in financial returns.

Estimating the number of factors \(p\) is a central problem in portfolio construction and risk management, since an accurate separation of signal and noise subspaces directly impacts covariance estimation and portfolio allocation \cite{jayduvdar13, Jay11}.

\medskip
\subsubsection{Real financial data}

We consider the publicly available \emph{100 Portfolios Formed on Size and Book-to-Market} dataset from the Kenneth R. French Data Library. The dataset consists of 100 value-weighted portfolios obtained from the intersection of 10 market-capitalization deciles and 10 Book-to-Market deciles. Monthly returns are available over several decades, providing a high-dimensional dataset with strong cross-sectional dependence.

The \(m=100\) portfolio return series are analyzed over rolling windows of \(N=150\) observations, corresponding to a dimensional ratio \(c=m/N=0.67\). Each window is processed using the proposed SCM-, Maronna-, and Tyler-based estimators (Section~\ref{sec::5}), together with AIC, MDL~\cite{wax1985detection}, and the Kritchman--Nadler (KN) detector~\cite{kritchman2009non}. Table~\ref{tab:realdatafinance} reports the empirical mean and variance of the estimated model order \(\hat p\).

\begin{table}[htbp]
\centering
\caption{Mean and variance of the estimated number of sources over $N=150$ days for \(m=100\) portfolio return series.}
\label{tab:realdatafinance}
\begin{tabular}{|c|c|c|}
\hline
Method & Mean of $\hat p$ & Variance of $\hat p$ \\
\hline
$\check{\mathbf{\Sigma}}_{\mathrm{SCM}}$ & 2.43 & 0.44 \\
$\check{\mathbf{\Sigma}}_{\mathrm{FP}}$ & 2.39 & 0.43 \\
$\check{\mathbf{\Sigma}}_{\mathrm{TYL}}$ & 5.09 & 0.63 \\
AIC & 32.43 & 156.53 \\
MDL & 7.26 & 4.93 \\
KN & 28.09 & 49.17 \\
\hline
\end{tabular}
\end{table}

Figure~\ref{financeplots} displays the eigenvalue distributions of \(\hat{\mathbf{C}}_{\mathrm{SCM}}\), \(\check{\mathbf{\Sigma}}_{\mathrm{SCM}}\), \(\check{\mathbf{\Sigma}}_{\mathrm{FP}}\), and \(\check{\mathbf{\Sigma}}_{\mathrm{TYL}}\). The robust whitening procedures concentrate the bulk of the eigenvalues while preserving only a few dominant outliers associated with the factor structure.

In contrast with AIC, MDL, and KN, which tend to overestimate the number of factors, the proposed robust approaches consistently identify between three and five dominant factors with low variability across rolling windows. This behavior is consistent with the widely documented low-dimensional factor structure of equity returns and confirms the relevance of robust covariance whitening for model-order estimation in financial applications.

\begin{figure*}[htbp]
\centering
\begin{tabular}{cc}
\includegraphics[width=0.95\columnwidth]{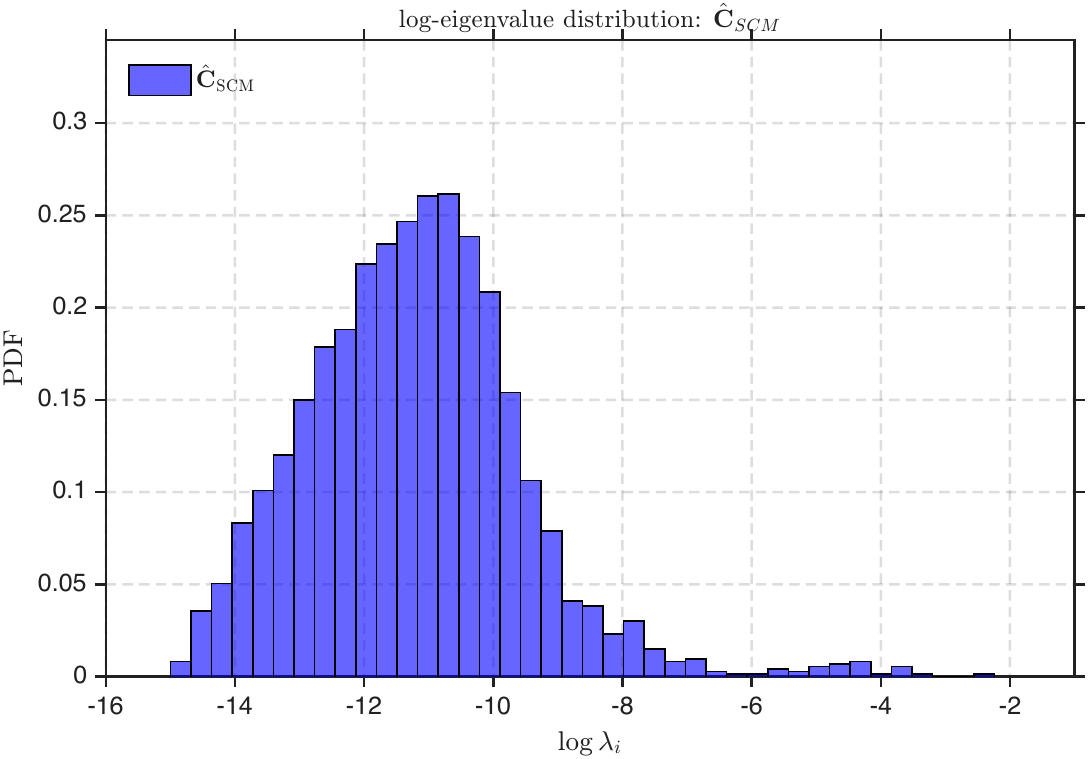}
&
\includegraphics[width=0.95\columnwidth]{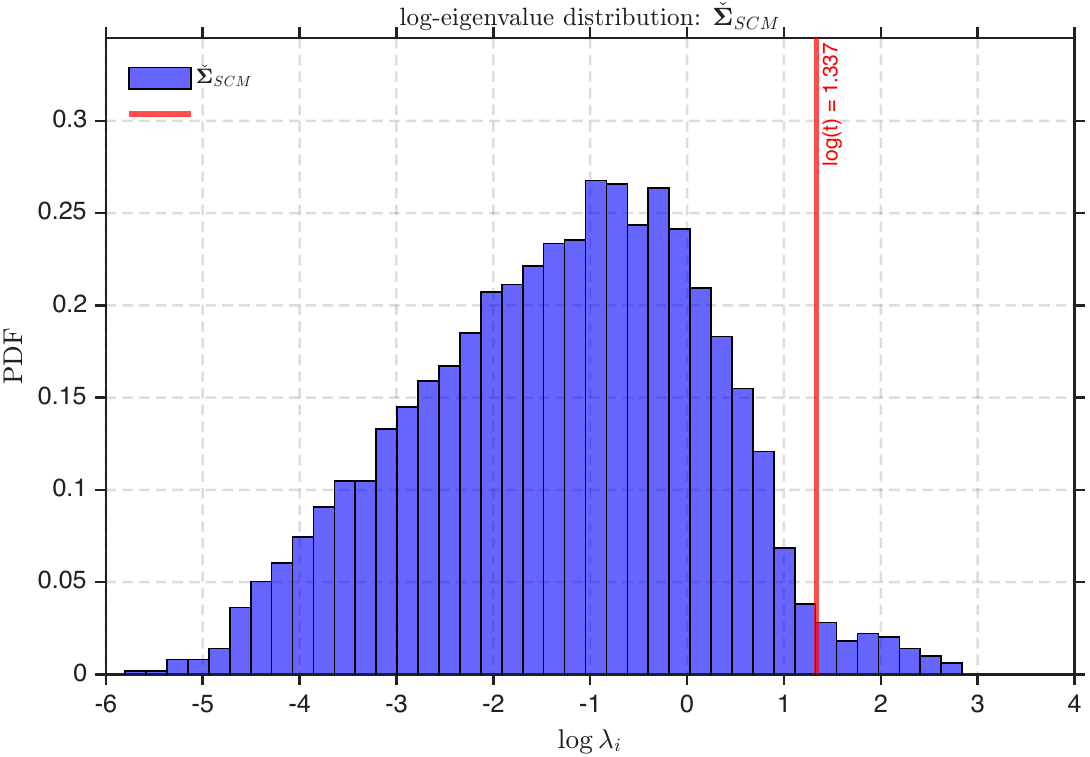} \\
\includegraphics[width=0.95\columnwidth]{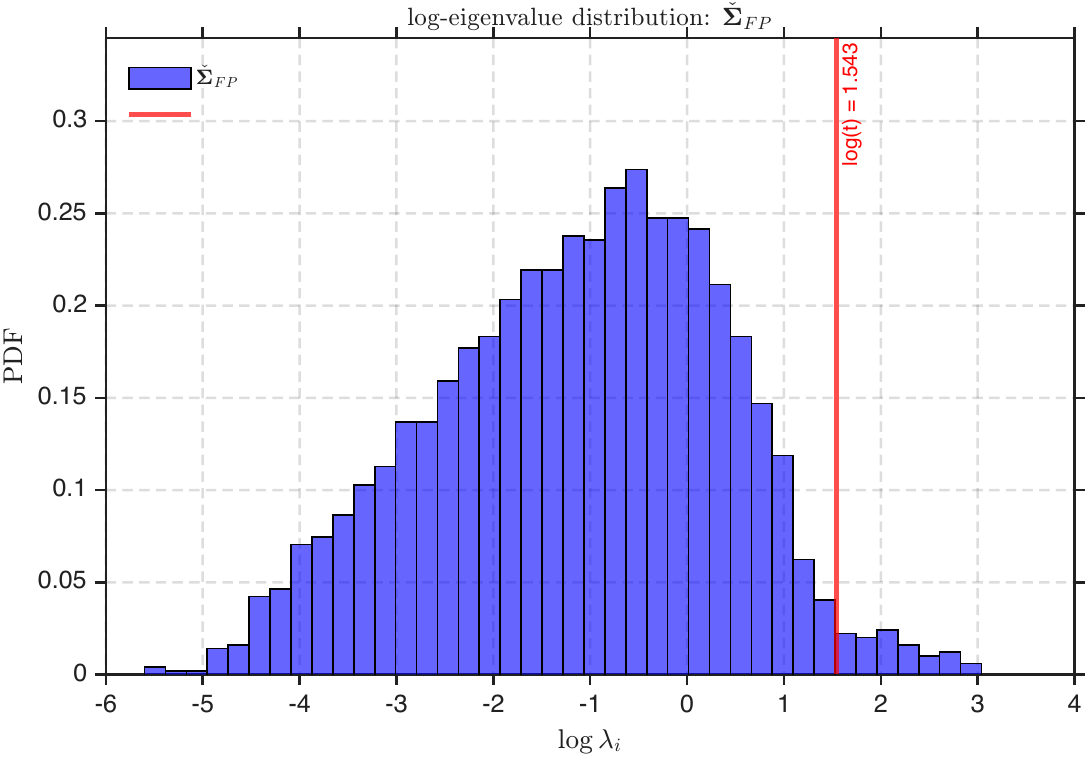}
&
\includegraphics[width=0.95\columnwidth]{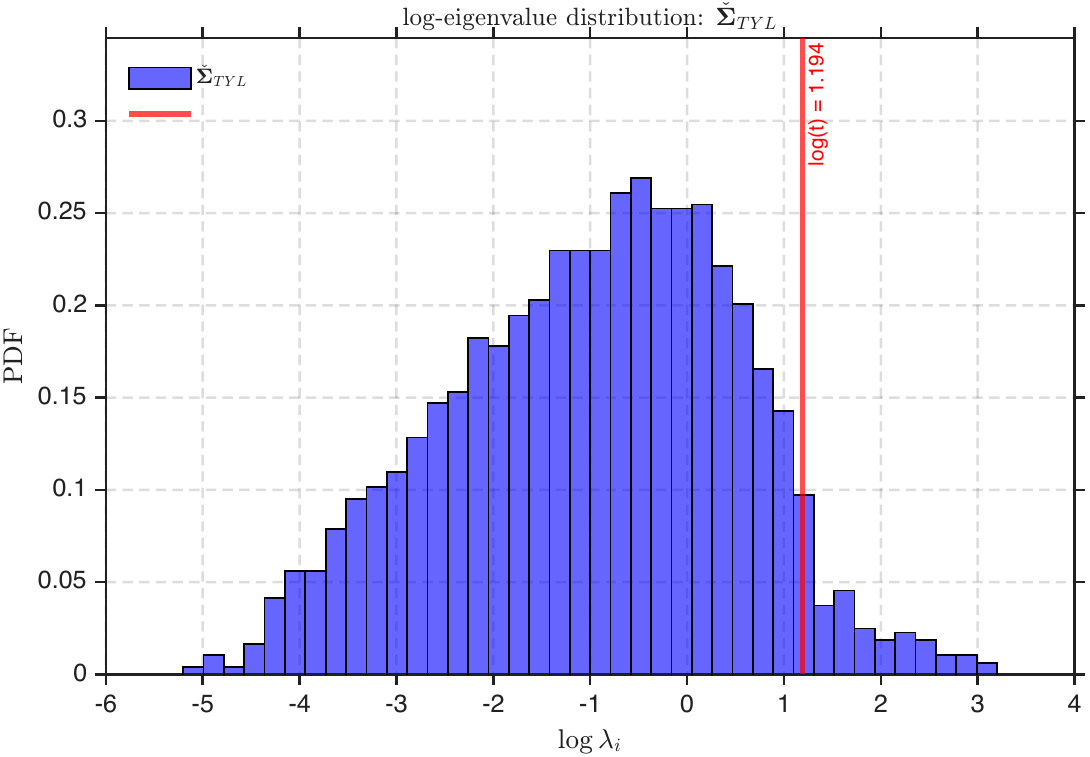}
\end{tabular}
\caption{Log-eigenvalue distributions of
\(\hat{\mathbf{C}}_{\mathrm{SCM}}\),
\(\check{\mathbf{\Sigma}}_{\mathrm{SCM}}\),
\(\check{\mathbf{\Sigma}}_{\mathrm{FP}}\),
and
\(\check{\mathbf{\Sigma}}_{\mathrm{TYL}}\)
for the 100 Fama--French portfolios dataset.
Here, \(m=100\), \(N=150\), and \(c=0.67\).}
\label{financeplots}
\end{figure*}

\section{Conclusion}

This paper has proposed and analysed a robust two-stage framework for model order selection under large-dimensional, correlated, non-Gaussian (CES) noise. Three estimation branches---SCM-based, Maronna $M$-estimator-based, and distribution-free Tyler-based---were developed within a unified Toeplitz whitening and RMT thresholding pipeline. Almost sure consistency of each branch was established, and an explicit closed-form eigenvalue threshold was derived, valid in the proportional-growth regime $m/N \to c > 0$.
 
The proposed approach demonstrates strong performance across all experimental settings---synthetic data, real hyperspectral images, EEG recordings, and financial data---highlighting its robustness to correlation and heavy-tailed noise distributions. In contrast with existing methods, the framework requires no prior knowledge of the noise structure and remains effective in challenging high-dimensional regimes.

Beyond hyperspectral imaging, the proposed methodology is particularly relevant for high-dimensional financial data, where asset returns exhibit latent factor structures embedded in correlated and non-Gaussian noise. In this context, it provides a principled tool for robust factor selection and covariance estimation, with potential applications to EEG source detection, portfolio optimization,
and risk management.

More broadly, this work opens several perspectives. On the theoretical side, extensions to non-Toeplitz structured noise (e.g., block-Toeplitz or low-rank-plus-Toeplitz covariance models) and to the spiked CES model with unknown $p$ and unknown $\mathbf{C}$ would be valuable. On the practical side, natural targets include radar clutter rank estimation, source localisation in antenna arrays, hyperspectral anomaly detection, and online (recursive) implementations for streaming financial data.



\putbib[bibliographieeug]
\end{bibunit}               

\clearpage
\appendices
\setcounter{page}{1}
\section*{Supplementary Material}
\setcounter{section}{0}
\setcounter{equation}{0}
\setcounter{theorem}{0}
\renewcommand{\thesection}{S\arabic{section}}
\renewcommand{\theequation}{S\arabic{equation}}
\renewcommand{\thetheorem}{S\arabic{theorem}}

\begin{bibunit}[IEEEtran]
This Supplementary Material provides detailed proofs of Theorems~1--4 stated in the main paper. We follow the same notations as in the main paper throughout. Theorem and lemma numbers in the Supplementary Material are prefixed by ``S-'' to distinguish them from the main paper. Reference numbers in this document are local to the Supplementary Material.

\setcounter{section}{0}
\renewcommand{\thesection}{S-\Alph{section}}
\renewcommand{\thetheorem}{S-\arabic{theorem}}
\renewcommand{\thelemma}{S-\arabic{lemma}}
 
\section{Proofs of Theorem~1 and Theorem~3}
\label{sec:supA}

The proofs follow~\cite{Vinogradova14}. Throughout, $c$ denotes the
limit of $c_N=m/N\to c>0$. We use the Toeplitz spectral-norm bound
from Lemma~4.1 of~\cite{Gray6}: for a Hermitian Toeplitz matrix
$\mathbf{A}=\mathcal{L}((a_0,\ldots,a_{m-1})^T)$ with absolutely
summable coefficients,
\begin{equation}
  \|\mathbf{A}\|\leq M_f
  =\sup_{\lambda\in[0,2\pi)}\!\left|\sum_{k=1-m}^{m-1}a_k e^{i\lambda k}\right|.
  \label{eq:toeplitz_bound}
\end{equation}

\subsection*{S-A.1\quad Proof of Theorem~1 (Consistency of
$\check{\mathbf{C}}_{SCM}$)}
\paragraph{Proof outline}
We bound $\left\|\mathcal{T}\left(\displaystyle\frac{1}{N}\mathbf{Y}\mathbf{Y}^H\right)) - \mathbf{C}\right\|$ into four spectral terms---see~\eqref{eq:thm1_four}---corresponding to: (1)~stochastic fluctuations of the noise, (2)~bias of the noise term, (3)~cross noise--signal interference, and (4)~pure signal contribution. Each term is shown to vanish almost surely via, respectively: (1)~Hanson--Wright concentration and Bernstein's inequality (Lemmas~\ref{lem:chi1_bound}--\ref{lem:chi3_bound}), (2)~direct computation, (3)~bilinearity and conditional centering, and (4)~absolute summability of the columns of $\mathbf{M}$ (Assumption~3-(i) of the main paper).

\paragraph{Proof}
Since $\mathbb{E}[\tau]=1$ (Assumption~2), the target identity is $\mathcal{T}\!\left(\displaystyle\frac{1}{N}\mathbf{Y}\mathbf{Y}^H\right) \overset{a.s.}{\to}\mathbf{C}$. We decompose the error into noise, cross, and signal terms and show each vanishes a.s.: the noise term via Hanson--Wright and Bernstein concentration, the cross term via bilinearity in $\mathbf{X}$ and $\boldsymbol{\delta}$, and the signal term via the absolute summability of $\mathbf{M}$ (Assumption~3). As in the main body, let
$\mathbf{Y}=\mathbf{M}\,\boldsymbol{\Gamma}^{1/2}\, \boldsymbol{\delta}^H +\mathbf{C}^{1/2}\,\mathbf{X}\,\mathbf{T}^{1/2}$ and let $\mathcal{T}$ be the Toeplitz operator defined in the introduction. Under Assumptions~1, 2 and~3, since $\mathcal{T}\!\left(\displaystyle \frac{1}{N}\mathbf{Y}\, \mathbf{Y}^H\right)$ and $\mathbb{E}[\tau]\,\mathbf{C}$ are both Toeplitz matrices, one can write thanks to~\eqref{eq:toeplitz_bound}:
\begin{equation}
\left\|\mathcal{T}\!\left(\displaystyle \frac{1}{N}\mathbf{Y}\, \mathbf{Y}^H\right) -\mathbb{E}[\tau]\,\mathbf{C}\right\|
\leq \!\sup_{\lambda\in[0,2\pi)} \left|\hat{\gamma}_m(\lambda) -\mathbb{E}[\tau]\,\gamma_m(\lambda)\right|\, ,
\label{eq:thm1_main}
\end{equation}
where $\gamma_m(\lambda)=\displaystyle\sum_{k=1-m}^{m-1}c_{k}\,e^{ik\lambda}$, $\hat{\gamma}_m(\lambda)=\displaystyle\sum_{k=1-m}^{m-1}\check{c}_{k}\,e^{ik\lambda}$ with $c_{-k}=c_k^\star$ and $\check{c}_{-k}=\check{c}_k^\star$. \\

The following lemma is essential for the development of the proof.
\begin{lemma}[Spectral representation]
\label{lem:gamma_hat}
The spectral density $\hat{\gamma}_m(\lambda)$ admits the representation:
\begin{equation}
\hat{\gamma}_m(\lambda) = \mathbf{d}_m^H(\lambda) \frac{\mathbf{Y} \mathbf{Y}^H}{N} \mathbf{d}_m(\lambda)\, ,
\label{eq:gamma_dm}
\end{equation}
with $\mathbf{d}_m(\lambda)=\displaystyle \frac{1}{\sqrt{m}}
\left(1,e^{-i\lambda},\ldots,e^{-i(m-1)\lambda}\right)^T$.
\end{lemma}

\begin{IEEEproof}
Equation~\eqref{eq:gamma_dm} can be rewritten as:
\[
\hat{\gamma}_m(\lambda) =\frac{1}{m N} \sum_{l, l^{\prime}=0}^{m-1} e^{-i\left(l^{\prime}-l\right) \lambda}\left[\mathbf{Y} \,\mathbf{Y}^H\right]_{l, l^{\prime}} =\sum_{k=1-m}^{m-1} \check{c}_k \, e^{-i k \lambda} \, ,
\]
where $\check{c}_k = \displaystyle\frac{1}{m N} \sum_{r=0}^{m-1} \sum_{j=0}^{N-1} y_{r, j} \, y_{r+k, j}^{\star} \, \mathbbm{1}_{0 \leq r+k \leq m-1}$.
\end{IEEEproof}
Decomposing $\hat{\gamma}_m(\lambda)$ according to the three contributions of
$\mathbf{Y}=\mathbf{M}\,\mathbf{S}+\mathbf{C}^{1/2}\,\mathbf{X}\,\mathbf{T}^{1/2}$, one defines $\hat{\gamma}_m(\lambda) = \hat{\gamma}^{\mathrm{sign}}_m(\lambda) + \hat{\gamma}^{\mathrm{cross}}_m(\lambda) + \hat{\gamma}^{\mathrm{noise}}_m(\lambda)$ with 
\begin{align*}
\hat{\gamma}^{\mathrm{sign}}_m(\lambda)
&= \frac{1}{N}\,\mathbf{d}_m^H(\lambda)\,\mathbf{M}\,
  \boldsymbol{\Gamma}^{1/2}\,\boldsymbol{\delta}^H\,\boldsymbol{\delta}\,
  \boldsymbol{\Gamma}^{1/2}\,\mathbf{M}^H\,\mathbf{d}_m(\lambda)\,,\\
\hat{\gamma}^{\mathrm{cross}}_m(\lambda)
&= \frac{2}{N}\,\mathcal{R}e\!\left(\mathbf{d}_m^H(\lambda)
  \,\mathbf{C}^{1/2}\,\mathbf{X}\,\mathbf{T}^{1/2}\,\boldsymbol{\delta}\,
  \boldsymbol{\Gamma}^{1/2}\,\mathbf{M}^H\,\mathbf{d}_m(\lambda)\right),\\
\hat{\gamma}^{\mathrm{noise}}_m(\lambda)
&= \frac{1}{N}\,\mathbf{d}_m^H(\lambda)\,
  \mathbf{C}^{1/2}\,\mathbf{X}\,\mathbf{T}\,\mathbf{X}^H\,\mathbf{C}^{1/2}
  \,\mathbf{d}_m(\lambda)\,.
\end{align*}

Equation~\eqref{eq:thm1_main} then leads to:
\begin{align}
\left\|\mathcal{T}\!\left(\frac{1}{N}\mathbf{Y}\mathbf{Y}^H\right)
-\mathbb{E}[\tau]\,\mathbf{C}\right\|
&\le
\underbrace{
\sup_\lambda\left|
\hat{\gamma}^{\mathrm{noise}}_m(\lambda)
-\mathbb{E}\!\left[
\hat{\gamma}^{\mathrm{noise}}_m(\lambda)
\right]
\right|
}_{\text{Term 1}}
\nonumber\\
&\hspace{-1.2cm}
+\underbrace{
\sup_\lambda\left|
\mathbb{E}\!\left[
\hat{\gamma}^{\mathrm{noise}}_m(\lambda)
\right]
-\mathbb{E}[\tau]\,\gamma_m(\lambda)
\right|
}_{\text{Term 2}}
\nonumber\\
&\hspace{-1.2cm}
+\underbrace{
\sup_\lambda\left|
\hat{\gamma}^{\mathrm{cross}}_m(\lambda)
\right|
}_{\text{Term 3}}
+\underbrace{
\sup_\lambda\left|
\hat{\gamma}^{\mathrm{sign}}_m(\lambda)
\right|
}_{\text{Term 4}}.
\label{eq:thm1_four}
\end{align}
We now analyse each term of~\eqref{eq:thm1_four} and show each term tends to zero almost surely.

\medskip
\noindent\textbf{Term 1:
$\displaystyle\sup_\lambda\left|\hat{\gamma}^{\mathrm{noise}}_m(\lambda)
-\mathbb{E}[\hat{\gamma}^{\mathrm{noise}}_m(\lambda)]\right|$.} We introduce auxiliary lemmas.

\begin{lemma}[Hanson--Wright inequality]
\label{lem:HW}
Let $\mathbf{x} \in \mathbb{C}^m$ be a sub-Gaussian isotropic random vector. Then for any deterministic matrix $\mathbf{A} \in \mathbb{C}^{m \times m}$ and any $t > 0$:
\[
\mathbb{P}\left(\left|\mathbf{x}^H \mathbf{A} \mathbf{x} - \operatorname{tr}(\mathbf{A})\right|\! > \!t\right)
\leq 2 \exp\!\left(\!-C \min\left(\frac{t^2}{\|\mathbf{A}\|_F^2}, \frac{t}{\|\mathbf{A}\|}\right)\!\!\right)\, .
\]
\end{lemma}

\begin{IEEEproof}
See \cite[Theorem 6.2.2]{Vershynin2018} and \cite{Rudelson2013}.
\end{IEEEproof}

\begin{lemma}[Tail control of $\|\mathbf{T}\|_\infty$]
\label{lem:gev_T}
Let $\{\tau_i\}_{i\geq 1}$ be i.i.d.\ positive random variables satisfying Assumption~2-(iii). Define the event 
$\mathcal{A}_N \triangleq \{\|\mathbf{T}\|_\infty \leq N^\kappa\}$, and denote by $\mathcal{A}_N^c \triangleq \{\|\mathbf{T}\|_\infty > N^\kappa\}$ its complement, where $\kappa$ satisfies
\begin{equation}
\label{eq:kappa_cond}
\frac{2}{\alpha_\tau} < \kappa < 1.
\end{equation}
Such a $\kappa$ exists since $\alpha_\tau > 2$ implies $2/\alpha_\tau < 1$, so the open interval $(2/\alpha_\tau, 1)$ is non-empty. The specific value of $\kappa$ does not affect the final result; any $\kappa \in (2/\alpha_\tau, 1)$ works, and the almost-sure conclusions below hold for every such choice. The upper bound $\kappa < 1$ ensures that $N^{1-\kappa} \to \infty$, which is needed for the Bernstein exponential bound to be summable (condition~(ii) below).

\medskip
\noindent\textbf{(i) Borel--Cantelli condition.}
Assume that $\displaystyle \sum_{N\ge 1} P(\mathcal A_N^c)<\infty$. By the Borel--Cantelli lemma, the event $\mathcal A_N$ occurs eventually almost surely. Hence, there exists a random integer $N_0$ such that $\|\mathbf T\|_\infty\le N^\kappa, N\ge N_0$.

\medskip
\noindent\textbf{(ii) Divergence condition.}
The condition $c_N N^{1-\kappa}\longrightarrow\infty$ holds. Indeed, since
$\kappa < 1$ by~\eqref{eq:kappa_cond}, the exponent $1-\kappa$ is strictly
positive, and since $c_N = m/N \to c > 0$ (Assumption~1):
\[
c_N N^{1-\kappa} \;\sim\; c\, N^{1-\kappa} \;\longrightarrow +\infty.
\]
This divergence condition will be invoked in 
Lemmas~\ref{lem:chi1_bound}--\ref{lem:chi3_bound} below to establish the 
almost sure convergence of the discretization remainder terms.

\end{lemma}

\begin{IEEEproof}
\noindent\textbf{Proof of~(i): summability of $P(\mathcal{A}_N^c)$.}

\noindent\textit{Step 1: pointwise tail bound.}
By Assumption~2-(iii), $P(\tau > x) = x^{-\alpha_\tau}\,\ell(x)$ where $\ell$ is slowly varying. Evaluating at $x = N^\kappa$, we obtain $P(\tau_1 > N^\kappa) = N^{-\alpha_\tau\kappa}\,\ell(N^\kappa)$.

\noindent\textit{Step 2: union bound.}
Since $\mathcal{A}_N^c = \left\{\displaystyle\max_{0\leq j \leq N-1} \tau_j > N^\kappa\right\}$ and the $\{\tau_j\}$ are i.i.d.:
\begin{equation}
\label{eq:union_bound}
P(\mathcal{A}_N^c) \leq N \, P(\tau_1 > N^\kappa) = N^{1-\alpha_\tau\kappa}\,\ell(N^\kappa)\, .
\end{equation}

\noindent\textit{Step 3: control of the slowly varying factor.}
By a standard property of slowly varying functions (see, e.g., \cite[Proposition 2.6 (Potter bounds)]{Resnick2007}), for any $\varepsilon > 0$, $\ell(x) = O(x^\varepsilon), \, x \to \infty$. Applying this with $x = N^\kappa$, we obtain
\[
\ell(N^\kappa) = O\big((N^\kappa)^\varepsilon\big) = O(N^{\kappa\, \varepsilon})\, .
\]
Let $\varepsilon \in (0, \alpha_\tau \,\kappa - 2)$. Renaming $\kappa \,\varepsilon$ as $\varepsilon$ (with a slight abuse of notation), we get $\ell(N^\kappa) = O(N^\varepsilon)$. Substituting into \eqref{eq:union_bound} yields
$\mathbb{P}(\mathcal{A}_N^c) = O\big(N^{1-\alpha_\tau \kappa + \varepsilon}\big)$. Define $\delta = \alpha_\tau \kappa - \varepsilon - 2 > 0$ which is positive by construction. Then $\mathbb{P}(A_N^c) = O\big(N^{-(1+\delta)}\big)$. Since $\delta > 0$, the series $\displaystyle\sum_N \mathbb{P}(A_N^c)$ converges, and the Borel--Cantelli lemma implies that $A_N$ holds eventually almost surely.
\end{IEEEproof}

Let $\lfloor\cdot\rfloor$ be the floor function. Choosing $\beta>3/2$, $I=\{0,\ldots,\lfloor N^\beta\rfloor-1\}$, and $\lambda_i=2\pi i/\lfloor N^\beta\rfloor$ for $i\in I$, one obtains:
\begin{equation}
\sup_{\lambda\in[0,2\pi)}
\left|
  \hat{\gamma}^{\mathrm{noise}}_m(\lambda)
  -\mathbb{E}\!\left[\hat{\gamma}^{\mathrm{noise}}_m(\lambda)\right]
\right| \leq \chi_1+\chi_2+\chi_3 \,,
\label{eq:chi_decomp}
\end{equation}
where
\begin{align*}
\chi_1
&\triangleq
\max_{i\in I}\sup_{\lambda\in[\lambda_i,\lambda_{i+1}]}
\left|
\hat{\gamma}^{\mathrm{noise}}_m(\lambda)
-\hat{\gamma}^{\mathrm{noise}}_m(\lambda_i)
\right|,
\\[0.5em]
\chi_2
&\triangleq
\max_{i\in I}
\left|
\hat{\gamma}^{\mathrm{noise}}_m(\lambda_i)
-\mathbb{E}\!\left[
\hat{\gamma}^{\mathrm{noise}}_m(\lambda_i)
\right]
\right|,
\\[0.5em]
\chi_3
&\triangleq
\max_{i\in I}\sup_{\lambda\in[\lambda_i,\lambda_{i+1}]}
\left|
\mathbb{E}\!\left[
\hat{\gamma}^{\mathrm{noise}}_m(\lambda_i)
\right]
-\mathbb{E}\!\left[
\hat{\gamma}^{\mathrm{noise}}_m(\lambda)
\right]
\right|.
\end{align*}
\begin{remark}[Role of Lemma~\ref{lem:gev_T} in the $\chi$ bounds]
\label{rem:chi_roles}
Lemma~\ref{lem:gev_T} supplies two ingredients will be used in 
Lemmas~\ref{lem:chi1_bound}--\ref{lem:chi3_bound}:
\begin{itemize}
  \item \emph{Part~(i)} (Borel--Cantelli for $\mathcal{A}_N$) is used in all 
  three lemmas to restrict to the high-probability event 
  $\mathcal{A}_N = \{\|\mathbf{T}\|_\infty \le N^\kappa\}$.
  \item \emph{Part~(ii)} (divergence of $c_N N^{1-\kappa}$) is used in 
  Lemma~\ref{lem:chi2_HW}: the Bernstein bound 
  $P(\chi_2>x) \le 2N^\beta \exp\!\bigl(-C x\,c_N N^{1-\kappa}/\|\mathbf{C}\|\bigr)$
  is summable precisely because $c_N N^{1-\kappa}\to+\infty$. It also 
  enters Lemma~\ref{lem:chi1_bound} via the decay 
  $N^{\kappa+1/2-\beta}\to 0$ (which uses $\kappa < 1$, a consequence 
  of~\eqref{eq:kappa_cond}).
\end{itemize}
\end{remark}

The convergence of each term is established through the following lemmas.

\begin{lemma}
\label{lem:chi1_bound}
For any $x>0$, one has $P[\chi_1>x]\xrightarrow{N\to\infty}0$, and the Borel--Cantelli lemma gives $\chi_1\xrightarrow{a.s.}0$.
\end{lemma}

\begin{IEEEproof}
Recall that:
\[
\chi_1 = \max_{i\in I}\sup_{\lambda\in[\lambda_i,\lambda_{i+1}]}
\left|\hat{\gamma}^{\mathrm{noise}}_m(\lambda)
-\hat{\gamma}^{\mathrm{noise}}_m(\lambda_i)\right|.
\]
Rewriting $\hat{\gamma}^{\mathrm{noise}}_m(\lambda)
-\hat{\gamma}^{\mathrm{noise}}_m(\lambda_i)$ as 
\begin{eqnarray*}
&\displaystyle\frac{1}{N}\left(\mathbf{d}_m^H(\lambda) - \mathbf{d}_m^H(\lambda_i)\right) \,
  \mathbf{C}^{1/2}\,\mathbf{X}\,\mathbf{T}\,\mathbf{X}^H\,\mathbf{C}^{1/2}
  \,\mathbf{d}_m(\lambda) &\\ 
   &+ \displaystyle\frac{1}{N}\,\mathbf{d}_m^H(\lambda_i)\,
  \mathbf{C}^{1/2}\,\mathbf{X}\,\mathbf{T}\,\mathbf{X}^H\,\mathbf{C}^{1/2}
  \,\left(\mathbf{d}_m^H(\lambda) - \mathbf{d}_m^H(\lambda_i)\right)\, .&
\end{eqnarray*}

\noindent\textit{Step 1: Lipschitz bound.}
Since $\lambda\mapsto\mathbf{d}_m(\lambda)$ is Lipschitz with constant $(m-1)/\sqrt{m}$, and using the sub-multiplicativity of the spectral norm together with the triangle inequality:
\begin{align}
\left|
\hat{\gamma}^{\mathrm{noise}}_m(\lambda)
-\hat{\gamma}^{\mathrm{noise}}_m(\lambda_i)
\right|
&\le
\frac{
2\,\|\mathbf{T}\|_\infty\,
\|\mathbf{C}\|\,
\|\mathbf{X}\mathbf{X}^H\|
}{N}
\nonumber\\
&\qquad\times
\frac{m-1}{\sqrt{m}}\,
|\lambda-\lambda_i|\,.
\label{eq:chi1_lip}
\end{align}
Since $|\lambda-\lambda_i|\leq 2\pi/\lfloor N^\beta\rfloor$ 
and $(m-1)/\sqrt{m} = O(\sqrt{N})$ (as $m = c_N \, N$):
\begin{equation}
\chi_1 \leq C\,\|\mathbf{T}\|_\infty\,\|\mathbf{C}\|\,
\frac{\|\mathbf{X}\mathbf{X}^H\|}{N}\,
\frac{\sqrt{N}}{N^\beta}\,,
\label{eq:chi1_bound_raw}
\end{equation}
where $C$ is a positive real constant. 
By the Bai--Silverstein theorem \cite{Bai98}, $\|\mathbf{X}\mathbf{X}^H\|/N \overset{a.s.}{\longrightarrow} (1+\sqrt{c})^2$. Substituting into~\eqref{eq:chi1_bound_raw} leads to:
\begin{equation}
\chi_1 \leq C\,\|\mathbf{C}\|\,
\|\mathbf{T}\|_\infty\, N^{1/2-\beta}
\quad \text{eventually a.s.}
\label{eq:chi1_bound_K}
\end{equation}
On the event $\mathcal{A}_N = \{\|\mathbf{T}\|_\infty 
\leq N^\kappa\}$, substituting into~\eqref{eq:chi1_bound_K}, we obtain:
\begin{equation}
\chi_1 \leq C\,\|\mathbf{C}\|\,
N^{\kappa+1/2-\beta}
\quad \text{on } \mathcal{A}_N\,.
\label{eq:chi1_final}
\end{equation}
Since $\kappa < 1$ and $\beta > 3/2$, one has $\beta > 3/2 > 1 > \kappa$, hence $\beta > \kappa + 1/2$ (because $3/2 > \kappa + 1/2 \iff \kappa < 1$). 
Therefore:
\begin{equation}
N^{\kappa+1/2-\beta} \longrightarrow 0
\quad \text{as } N\to\infty\,.
\end{equation}
Finally, via Borel--Cantelli, for any $x > 0$, we have:
\[
P[\chi_1 > x] 
\leq P(\mathcal{A}_N^c) 
+ P\!\left[C\,\|\mathbf{C}\|\, 
N^{\kappa+1/2-\beta} > x\right].
\]
The first term satisfies $\displaystyle\sum_N P(\mathcal{A}_N^c) < \infty$ by Lemma~\ref{lem:gev_T}-(i). The second term: since $N^{\kappa+1/2-\beta}\to 0$, the threshold $x/(C\|\mathbf{C}\|N^{\kappa+1/2-\beta}) \to +\infty$, so 
$P[C\,\|\mathbf{C}\|\, N^{\kappa+1/2-\beta}>x] \to 0$ and is summable. Hence $\displaystyle\sum_N \,P[\chi_1>x] < \infty$, and the Borel--Cantelli lemma gives $\chi_1\xrightarrow{a.s.}0$.
\end{IEEEproof}

\begin{lemma}
\label{lem:chi2_HW}
For any $x>0$, one has $P[\chi_2>x]\xrightarrow{N\to\infty}0$, and $\chi_2 \xrightarrow{a.s.} 0$.
\end{lemma}
\begin{IEEEproof}
Write $\hat{\gamma}^{\mathrm{noise}}_m(\lambda_i) = \displaystyle \frac{1}{N}\sum_{j=0}^{N-1}\tau_j\, \mathbf{x}_j^H\, \mathbf{A}(\lambda_i)\,\mathbf{x}_j$ for fixed $\lambda_i$,  with $\mathbf{A}(\lambda_i) \triangleq \mathbf{C}^{1/2}\,\mathbf{d}_m(\lambda_i)\,\mathbf{d}_m^H(\lambda_i)\,\mathbf{C}^{1/2}$. 
Note that $\|\mathbf{A}(\lambda_i)\|=\|\mathbf{A}(\lambda_i)\|_F\leq\|\mathbf{C}\|$ (rank-one) and $\operatorname{tr}(\mathbf{A}(\lambda_i))=\gamma_m(\lambda_i)$.
The fully centered quantity decomposes as:
\begin{gather*}
\hat{\gamma}^{\mathrm{noise}}_m(\lambda_i)
-\mathbb{E}\!\left[\hat{\gamma}^{\mathrm{noise}}_m(\lambda_i)\right] =
\\
\underbrace{
\frac{1}{N}\sum_j\tau_j
\Bigl(
\mathbf{x}_j^H\mathbf{A}(\lambda_i)\mathbf{x}_j
-\gamma_m(\lambda_i)
\Bigr)
}_{=:R_1}
+
\underbrace{
\frac{\gamma_m(\lambda_i)}{N}
\sum_j(\tau_j-\mathbb{E}[\tau])
}_{=:R_2}.
\end{gather*}

\medskip
\textit{Control of $R_1$, conditioning on $\mathcal{A}_N$.}
On $\mathcal{A}_N$, we have $\tau_j\leq N^\kappa$ for all $j$.
Conditionally on $\tau_j$, the vector $\mathbf{x}_j$ is standard Gaussian. Therefore, applying the Hanson--Wright inequality (Lemma~\ref{lem:HW}) to the quadratic form
$\mathbf{x}_j^H\mathbf{A}(\lambda)\mathbf{x}_j$ yields
\begin{gather}
P\!\left(
\left|
\mathbf{x}_j^H\,\mathbf{A}(\lambda)\,\mathbf{x}_j
-\gamma_m(\lambda)
\right|
> t
\right)
\nonumber\\
\le
2\exp\!\left(
-C\min\!\left(
\frac{t^2}{\|\mathbf{A}(\lambda)\|_F^2},
\frac{t}{\|\mathbf{A}(\lambda)\|}
\right)
\right).
\end{gather}
Moreover, $\gamma_m(\lambda) =
\operatorname{tr}\!\bigl(\mathbf{A}(\lambda)\bigr)
=
\mathbf{d}_m^H(\lambda)\,
\mathbf{C}\,
\mathbf{d}_m(\lambda)$. Since
\[
\mathbf{A}(\lambda)
=
\mathbf{C}^{1/2}
\mathbf{d}_m(\lambda)
\mathbf{d}_m^H(\lambda)
\mathbf{C}^{1/2}\, ,
\]
has rank one, its Frobenius and spectral norms coincide:
\[
\|\mathbf{A}(\lambda)\|_F
=
\|\mathbf{A}(\lambda)\|
=
\|\mathbf{C}^{1/2}\mathbf{d}_m(\lambda)\|^2
\leq
\|\mathbf{C}\|\,.
\]
Substituting these bounds into the Hanson--Wright inequality gives
\begin{gather}
P\!\left(
\left|
\mathbf{x}_j^H\,\mathbf{A}(\lambda)\,\mathbf{x}_j
-\gamma_m(\lambda)
\right|
> t
\right)
\nonumber\\
\le
2\exp\!\left(
-C\min\!\left(
\frac{t^2}{\|\mathbf{C}\|^2},
\frac{t}{\|\mathbf{C}\|}
\right)
\right).
\end{gather}

Setting $\varepsilon_{j} = \mathbf{x}_j^H\,\mathbf{A}(\lambda_i)\, \mathbf{x}_j - \gamma_m(\lambda_i)$, the variables $\{\tau_j\varepsilon_j\}$ are
independent conditionally on $\mathbf{T}$. To obtain a deterministic bound, we define the event $\mathcal{B}_N \triangleq \left\{\displaystyle\max_{0\leq j\leq N-1}\|\mathbf{x}_j\|^2 \leq 2m\right\}$. By standard sub-Gaussian concentration \cite{Vershynin2018}, $P(\mathcal{B}_N^c)\leq 2N e^{-Cc_N N}$, which is summable since $c_N\to c>0$ and 
$c_N N = m \to\infty$, so $\mathcal{B}_N$ holds eventually almost surely.

On the event $\mathcal{A}_N\cap\mathcal{B}_N$, one has $\|\mathbf{x}_j\|^2\leq 2m$ and $\tau_j\leq N^\kappa$ deterministically, so:
\[
|\tau_j\,\varepsilon_j| \leq \tau_j\,\|\mathbf{C}\|\,\|\mathbf{x}_j\|^2 \leq 2m \, N^\kappa\,\|\mathbf{C}\| \;\triangleq\; B_{N}\,,
\]
which is a deterministic bound. Since $\varepsilon_j$ is centered and $\tau_j\, \varepsilon_j$ is bounded by $B_N$ on $\mathcal{A}_N\cap
\mathcal{B}_N$, it is sub-exponential with parameter $b = B_N/N = 2\,c_N\,\|\mathbf{C}\|\,N^\kappa$. Let us define $\sigma^2 = \mathbb{E}\left[|\tau_j\, \varepsilon_j|^2\right]$. By the Bernstein inequality \cite{Vershynin2018}:
\begin{align*}
P\!\left[
|R_1|>x
\,\Big|\,
\mathbf{T},\mathcal{A}_N,\mathcal{B}_N
\right]
&\le
2\exp\!\left(
-C\,
\min\!\left(
\frac{N x^2}{\sigma^2},
\frac{N x}{B_N}
\right)
\right) ,
\\
&\le
2\exp\!\left(
-C\,
\frac{N x}{B_N}
\right).
\end{align*}
Indeed, $\sigma^2 \leq C\,\mathbb{E}[\tau^2]\,\|\mathbf{C}\|^2 < \infty$ (by Assumption~2-(iii), $\mathbb{E}[\tau^2]<\infty$), while $B_N = 2c_N\|\mathbf{C}\|N^\kappa \to \infty$, so $\sigma^2/B_N \to 0$ and the variance term in Bernstein is dominated by the linear term for $N$ large. By a union bound over the $\lfloor N^\beta\rfloor$ grid points:
\[
P[\chi_2>x,\,\mathcal{A}_N,\,\mathcal{B}_N]
\leq 2N^\beta\exp\!\left(-c_N N^{1-\kappa}\,
\frac{C\, x}{\|\mathbf{C}\|}\right).
\]
Since $c_N \,N^{1-\kappa}\to+\infty$ (Lemma~\ref{lem:gev_T}), this is summable. Combined with $\displaystyle\sum_N P(\mathcal{A}_N^c)<\infty$ and $\displaystyle \sum_N P(\mathcal{B}_N^c)<\infty$, the Borel--Cantelli lemma gives
$\chi_2\xrightarrow{a.s.} 0$.

\begin{remark}[Unconditional Bernstein bound]
The bound $B_N = 2m\,\|\mathbf{C}\|\, N^\kappa$ is deterministic on
$\mathcal{A}_N\cap\mathcal{B}_N$; the Bernstein inequality therefore
applies unconditionally on that event, with no further conditioning
required.
\end{remark}

\textit{Control of $R_2$.} By the strong Law of Large Numbers---valid since $\mathbb{E}[\tau]<\infty$ under Assumption~2-(iii) with $\alpha_\tau>2$---we have $\displaystyle\frac{1}{N}\sum_j(\tau_j-\mathbb{E}[\tau]) \xrightarrow{a.s.} 0$, so $R_2\xrightarrow{a.s.} 0$ uniformly over the finite grid.
\end{IEEEproof}

\begin{lemma}
\label{lem:chi3_bound}
$\chi_3\leq C\,\|\mathbf{C}\| \, N^{-\beta+1} \xrightarrow{N\to\infty}0$ for $\beta>1$.
\end{lemma}

\begin{IEEEproof}
The proof is the same as Lemma~6 in~\cite{Vinogradova14}. Note that on the event $\mathcal{A}_N$ the factor $\|\mathbf{T}\|_\infty\leq
N^\kappa$ is absorbed in the constant $C$, and the dominant decay is $N^{-\beta+1}$ regardless of $\|\mathbf{T}\|_\infty$.
\end{IEEEproof}

These three inequalities prove that, Term 1, $\forall\, x\in\mathbb{R}^{+\star}$ tends to zero almost surely:
\[
P\Bigl(
\sup_{\lambda\in[0,2\pi)}
\bigl|
\hat{\gamma}^{\mathrm{noise}}_m(\lambda)
-\mathbb{E}[\hat{\gamma}^{\mathrm{noise}}_m(\lambda)]
\bigr|
>x
\Bigr)
\xrightarrow[N\to\infty]{} 0 .
\]

While Term 1 describes the purely stochastic fluctuations of the noise, the second term, Term 2, captures the interaction between the signal signatures and the estimated noise subspace.

\medskip
\noindent\textbf{Term~2.}
By definition, $\hat{\gamma}^{\mathrm{noise}}_m(\lambda)$ involves only the noise
component $\mathbf{C}^{1/2}\mathbf{X}\mathbf{T}^{1/2}$ of $\mathbf{Y}$.
Since $\{\tau_j\}$ and $\{\mathbf{x}_j\}$ are independent and
$\mathbb{E}[\mathbf{x}_j\mathbf{x}_j^H]=\mathbf{I}_m$, a direct computation gives:
\begin{align*}
\mathbb{E}\!\left[\hat{\gamma}^{\mathrm{noise}}_m(\lambda)\right]
&\!= \!\frac{1}{N}\sum_{j=0}^{N-1}\!
   \mathbb{E}[\tau_j]
   \mathbf{d}_m^H(\lambda)\mathbf{C}^{1/2}
   \mathbb{E}\!\left[\mathbf{x}_j\mathbf{x}_j^H\right]
   \mathbf{C}^{1/2}\mathbf{d}_m(\lambda), \\
&= \mathbb{E}[\tau]\,
   \mathbf{d}_m^H(\lambda)\,\mathbf{C}\,\mathbf{d}_m(\lambda)
 = \mathbb{E}[\tau]\,\gamma_m(\lambda)\,,
\end{align*}
where we used $\mathbb{E}[\tau_j]=\mathbb{E}[\tau]=1$ (Assumption~2-(iii))
and $\mathbb{E}[\mathbf{x}_j\mathbf{x}_j^H]=\mathbf{I}_m$ (Assumption~2-(iv)).
Hence Term~2 vanishes identically.

\medskip
The analysis of Term 3 is more involved as it represents the residual interference after whitening. Its convergence is established by leveraging the concentration of quadratic forms and the properties of the Toeplitz rectification operator.

\medskip
\noindent\textbf{Term 3:
$\displaystyle\sup_\lambda|\hat{\gamma}^{\mathrm{cross}}_m(\lambda)|$.}

Recall:
\[
\hat{\gamma}^{\mathrm{cross}}_m(\lambda)
= \frac{2}{N}\,\mathcal{R}e\!\left(\mathbf{d}_m^H(\lambda)\,
\mathbf{C}^{1/2}\,\mathbf{X}\,\mathbf{T}^{1/2}\,\boldsymbol{\delta}\,
\boldsymbol{\Gamma}^{1/2}\,\mathbf{M}^H\,\mathbf{d}_m(\lambda)\right)\,.
\]
Define $\mathbf{H}(\lambda) \;=\; \mathbf{M}^H\, \mathbf{d}_m(\lambda)\,\mathbf{d}_m^H(\lambda)\, \mathbf{C}^{1/2} \in \mathbb{C}^{p \times m}$ and $\mathbf{K} \;=\; \mathbf{T}^{1/2}\,\boldsymbol{\Gamma}^{1/2} \in \mathbb{C}^{N \times p}$. Then:
\begin{align}
\hat{\gamma}^{\mathrm{cross}}_m(\lambda)
&=
\frac{2}{N}\,
\mathrm{Re}\!\left[
\operatorname{tr}\!\left(
\mathbf{H}(\lambda)\,
\mathbf{X}\,
\mathbf{K}\,
\boldsymbol{\delta}^H
\right)
\right]\, ,
\nonumber\\
&=
\frac{2}{N}\,
\mathrm{Re}\!\left[
\mathrm{vec}^{T}\!\left(
\mathbf{H}(\lambda)\,\mathbf{X}
\right)\,
\mathrm{vec}\!\left(
\mathbf{K}\,\boldsymbol{\delta}
\right)
\right].
\label{eq:cross_compact}
\end{align}
The key point is that this expression is \emph{bilinear} in the two independent random matrices $\mathbf{X}$ and $\boldsymbol{\delta}$; combined with $\mathbb{E}[\mathbf{X}]=\mathbf{0}$, this drives convergence to zero.

We decompose the proof into two parts: a Lipschitz reduction to a finite grid (control of $\chi_1$), and a pointwise then uniform almost sure convergence over the grid (control of $\chi_2$).

\medskip
\noindent\textit{Reduction to a finite grid.}\\
Let $\beta > 3/2$, $I = \{0,\ldots,\lfloor N^\beta\rfloor-1\}$, and $\lambda_i = 2\pi i/\lfloor N^\beta\rfloor$. Write:
\begin{equation}
\sup_\lambda
\left|\hat{\gamma}^{\mathrm{cross}}_m(\lambda)\right| \le \rho_1 + \rho_2\, ,
\end{equation}
where 
\begin{align}
&\rho_1 = \max_{i\in I}
\sup_{\lambda\in[\lambda_i,\lambda_{i+1}]}
\left|
\hat{\gamma}^{\mathrm{cross}}_m(\lambda)
-\hat{\gamma}^{\mathrm{cross}}_m(\lambda_i)
\right|\, ,
\nonumber\\
&\rho_2 = 
\max_{i\in I}
\left|
\hat{\gamma}^{\mathrm{cross}}_m(\lambda_i)
\right|\, .
\label{eq:cross_grid}
\end{align}
\begin{lemma}[Lipschitz control]
\label{lem:cross_lip}
We have $\rho_1\xrightarrow{a.s.}0$.
\end{lemma}
\begin{IEEEproof}
Since $\lambda \mapsto \mathbf{d}_m(\lambda)$ is Lipschitz with constant $(m-1)/\sqrt{m}=O(\sqrt{m})$, and using the sub-multiplicativity of the spectral norm:
\[
\begin{aligned}
\left| \hat{\gamma}^{\mathrm{cross}}_m(\lambda)
- \hat{\gamma}^{\mathrm{cross}}_m(\lambda_i) \right|
\le\;& \frac{C\,
\|\mathbf{C}\|^{1/2}
\|\mathbf{M}\|
\|\mathbf{X}\|
}{N\sqrt{m}\,\lfloor N^\beta\rfloor}
\\
&\times
\|\mathbf{T}\|_\infty^{1/2} \|\boldsymbol{\Gamma}\|^{1/2}(m-1)\,.
\end{aligned}
\]
On $\mathcal{A}_N$: $\|\mathbf{T}\|_\infty^{1/2}\leq N^{\kappa/2}$. By \cite{Bai2010}, $\|\mathbf{X}\|=O\left(\sqrt{N}\right)$ a.s. Since $m=c_N N$, we get $(m-1)/\sqrt{m}=O(\sqrt{N})$, hence:
\[
\rho_1 \leq C' \, N^{\kappa/2-\beta} \xrightarrow{N\to\infty} 0\,,
\]
for $C'>0$,  $\beta > \kappa/2$ (satisfied since $\beta>3/2>\kappa/2$ for $\kappa<1$). Borel--Cantelli, using summability of $P(\mathcal{A}_N^c)$, gives $\rho_1\xrightarrow{a.s.}0$.
\end{IEEEproof}

\begin{lemma}[Pointwise and uniform a.s.\ convergence]
\label{lem:cross_pw}
For any fixed $\lambda\in[0,2\pi)$, $\hat{\gamma}^{\mathrm{cross}}_m(\lambda)\xrightarrow{a.s.}0$. Consequently, $\rho_2\xrightarrow{a.s.}0$.
\end{lemma}

\begin{IEEEproof}
Write $\hat{\gamma}^{\mathrm{cross}}_m(\lambda) = \displaystyle\frac{2}{N}\,\mathrm{Re}\!\Bigl(\sum_{j=0}^{N-1} f_j(\lambda)\Bigr)$
where $f_j(\lambda) = \mathbf{x}_j^H\,\mathbf{H}^H(\lambda)\,\mathbf{j}_j$, with $\mathbf{j}_j = \tau_j^{1/2}\,\boldsymbol{\Gamma}^{1/2}\,\boldsymbol{\delta}_j$ the $j$-th column of $\mathbf{K}\boldsymbol{\delta}$. Since $\mathbf{X}$ and $(\boldsymbol{\delta},\mathbf{T})$ are independent and $\mathbb{E}[\mathbf{x}_j]=\mathbf{0}$:
\[
\mathbb{E}[f_j(\lambda)\mid\mathbf{T},\boldsymbol{\delta}] = 0\,.
\]
Thus $\{f_j\}$ are conditionally centered and independent given $(\mathbf{T},\boldsymbol{\delta})$. The proof proceeds in three steps.

\medskip
\textit{Step 1: High-probability events.}
Define:
\[
\begin{gathered}
\mathcal{A}_N
=
\{\|\mathbf{T}\|_\infty \le N^\kappa\},
\qquad
\mathcal{B}_N
=
\Bigl\{
\max_{0\le j\le N-1}\|\mathbf{x}_j\|^2 \le 2m
\Bigr\},
\\[0.3em]
\text{and } \mathcal{C}_N
=
\Bigl\{
\max_{0\le j\le N-1}\|\boldsymbol{\delta}_j\|^2
\le 4\log N
\Bigr\}\,.
\end{gathered}
\]
By Lemma~\ref{lem:gev_T}, $\displaystyle\sum_N P(\mathcal{A}_N^c)<\infty$. By sub-Gaussian concentration \cite{Vershynin2018}, $P(\mathcal{B}_N^c)\leq 2N e^{-Cm}$, which is summable. For $\mathcal{C}_N$: since $\boldsymbol{\delta}_j\sim\mathcal{CN}(\mathbf{0},\mathbf{I}_p)$, the variable $\|\boldsymbol{\delta}_j\|^2$ is sub-exponential with mean $p$ (i.e., $P(\|\boldsymbol{\delta}_j\|^2) > t)=2\, e^{-C(t-p)})$. For $N$ large enough that $4\log N > p$, a union bound and standard sub-exponential tail bounds give:
\[
P(\mathcal{C}_N^c)
\le
2N\,e^{-C(4\log N - p)}
=
2\,e^{Cp}\,N^{1-4C}.
\]
which is $O\left(N^{1-4C}\right)$ for $N$ large; for some positive constant $C$ depending only on the sub-Gaussian norm, this is summable. Hence, $\displaystyle\sum_N P(\mathcal{C}_N^c) < \infty$.

\medskip
\textit{Step 2: Deterministic bounds on $\mathcal{A}_N\cap\mathcal{B}_N\cap\mathcal{C}_N$.}
On this event, for each $j$:
\[
\|f_j(\lambda)\| \leq \|\mathbf{H}(\lambda)\|\,\|\mathbf{x}_j\|\,\|\mathbf{j}_j\| \leq K_1\sqrt{2m}\, N^{\kappa/2} K_2^{1/2}\sqrt{4\log N}\,,
\]
where $K_1=\|\mathbf{M}\|\|\mathbf{C}\|^{1/2}$ and $K_2=\|\boldsymbol{\Gamma}\|$. Hence $|f_j(\lambda)|\leq b_N \triangleq C\,K_1 K_2^{1/2}\, N^{(1+\kappa)/2}\sqrt{\log N}$.

For the conditional variance:
\[
\mathbb{E}\!\left[|f_j(\lambda)|^2\,\Big|\,\mathbf{j}_j\right] = \mathbf{j}_j^H\,\mathbf{H}(\lambda)\,\mathbf{H}^H(\lambda)\,\mathbf{j}_j \leq K_1^2\,\|\mathbf{j}_j\|^2\,,
\]
so on $\mathcal{A}_N\cap\mathcal{B}_N\cap\mathcal{C}_N$:
\[
\begin{aligned}
\sigma_N^2
&\triangleq
\frac{1}{N^2}
\sum_{j=0}^{N-1}
\mathbb{E}\!\left[
|f_j(\lambda)|^2
\,\Big|\,
\mathbf{T},\boldsymbol{\delta}
\right]
\le
\frac{K_1^2 K_2 N^\kappa}{N^2}
\sum_{j=0}^{N-1}
\|\boldsymbol{\delta}_j\|^2\, ,
\\
&\le
\frac{
4 K_1^2 K_2 N^\kappa \log N
}{N}
\le
4 K_1^2 K_2
\,N^{\kappa-1}\log N\,.
\end{aligned}
\]

\textit{Step 3: Bernstein inequality and Borel--Cantelli.}
The Bernstein inequality \cite[Theorem.~2.8.1]{Vershynin2018} applied to $\mathrm{Re}[f_j(\lambda)]$ (which are real, centered, bounded by $b_N$, conditionally on $\mathbf{T},\boldsymbol{\delta}$) gives, for any $x>0$:
\[
\begin{aligned}
&P\!\left(
\left|\hat{\gamma}^{\mathrm{cross}}_m(\lambda)\right| > x
\,\Big|\,
\mathbf{T},
\boldsymbol{\delta},
\mathcal{A}_N\cap\mathcal{B}_N\cap\mathcal{C}_N
\right)
\\
&\qquad\le
2\exp\!\Biggl(
-\frac{N^2\,x^2/2}{
N^2\,\sigma_N^2+N\,x\,b_N/3
}
\Biggr).
\end{aligned}
\]
Fix $\varepsilon \in \left(\frac{1+\kappa}{2},\, 1\right)$ 
(which is non-empty since $\kappa < 1$ implies 
$\frac{1+\kappa}{2} < 1$):
\begin{itemize}
\item Variance term: setting $x = N^{-(1-\varepsilon)}$ so that $x^2 = N^{-2(1-\varepsilon)}$:
\end{itemize}
\begin{equation}
\frac{N^2 \,x^2}{N^2\,\sigma_N^2}
= \frac{x^2}{\sigma_N^2}
\geq \frac{N^{-2(1-\varepsilon)}}{C\,N^{\kappa-1}\log N}
= C'\, \frac{N^{2\varepsilon - 1 - \kappa}}{\log N}\,.
\end{equation}
This diverges to $+\infty$ since $2\varepsilon > 1+\kappa$.
\begin{itemize}
\item Single-term bound: Similarly,
\end{itemize}
\begin{align*}
N\,x\,b_N
&=
N\,N^{-(1-\varepsilon)}
\,O\!\left(
N^{(1+\kappa)/2}\sqrt{\log N}\, ,
\right)
\\
&=
O\!\left(
N^{\varepsilon+\frac{1+\kappa}{2}}
\sqrt{\log N}
\right).
\end{align*}
Hence, 
$\displaystyle\frac{N^2 \,x^2}{N \, x \,b_N}
= O\!\left(N^{\varepsilon - \frac{1+\kappa}{2}} \sqrt{\log N}\right)$ which diverges to $+\infty$ as soon as $\varepsilon > \displaystyle\frac{1+\kappa}{2}$. Under the condition $\varepsilon > \displaystyle\frac{1+\kappa}{2}$, both terms in the denominator of Bernstein's inequality are dominated by $N^2 \,x^2$, so that:
\begin{equation}
P\!\left(\left|\hat{\gamma}^{\mathrm{cross}}_m(\lambda)\right| > N^{-(1-\varepsilon)}\right)
\leq 2 \exp\!\left(-C \, N^{\delta}\right),
\end{equation}
for some constants $C>0$ and $\delta=\varepsilon - \displaystyle\frac{1+\kappa}{2} >0$. Integrating over $(\mathbf{T},\boldsymbol{\delta})$ leads to:
\begin{align*}
P\!\left(
\left|\hat{\gamma}^{\mathrm{cross}}_m(\lambda)\right|
> N^{-(1-\varepsilon)}
\right)
&\le
P(\mathcal{A}_N^c)
+ P(\mathcal{B}_N^c)
+ P(\mathcal{C}_N^c)
\\
&\quad
+ 2\,e^{-C\,N^{\delta}}\,.
\end{align*}

For the uniform convergence over the grid $\{\lambda_i\}_{i\in I}$ of size $\lfloor N^\beta\rfloor$, a union bound over the grid multiplies the bound by $N^\beta$, yielding $C N^\beta e^{-N^{2\varepsilon}/C}$, which is still summable for any $\beta$. Hence $\rho_2\xrightarrow{a.s.}0$.
\end{IEEEproof}

Combining Lemmas~\ref{lem:cross_lip} and~\ref{lem:cross_pw} in~\eqref{eq:cross_grid}:
\[
\sup_{\lambda\in[0,2\pi)}\left|\hat{\gamma}^{\mathrm{cross}}_m(\lambda)\right|
\leq \rho_1+\rho_2\xrightarrow{a.s.}0\,,
\]
completing the proof that Term~3 tends to zero almost surely.

\medskip
\noindent\textbf{Term 4:
$\displaystyle\sup_\lambda|\hat{\gamma}^{\mathrm{sign}}_m(\lambda)|$.} We recall that:
\[\hat{\gamma}^{\mathrm{sign}}_m(\lambda)
= \frac{1}{N}\,\mathbf{d}_m^H(\lambda)\,\mathbf{M}\, \boldsymbol{\Gamma}^{1/2}\,\boldsymbol{\delta}^H\,\boldsymbol{\delta}\,
  \boldsymbol{\Gamma}^{1/2}\,\mathbf{M}^H\,\mathbf{d}_m(\lambda)\, .
\]
where $\mathbf{M}\in\mathbb{C}^{m\times p}$, $\boldsymbol{\delta}\in\mathbb{C}^{N\times p}$, $\boldsymbol{\Gamma}\in\mathbb{C}^{p\times p}$ is Hermitian nonnegative definite, and $\mathbf{d}_m(\lambda)\in\mathbb{C}^m$ with $\|\mathbf{d}_m(\lambda)\|=1$.

The key structural difference with the cross term is that $\hat{\gamma}^{\mathrm{sign}}_m(\lambda)$ is quadratic in $\boldsymbol{\delta}$ (and not bilinear in $\mathbf{X}$ and $\boldsymbol{\delta}$). There is therefore only a single source of randomness ($\boldsymbol{\delta}$), and the reasoning is more straightforward but different.

\medskip
\noindent\textit{Step 1: reduction to a sum of i.i.d.\ complex
products.} Define the deterministic vectors: $\mathbf{h}(\lambda) \triangleq \mathbf{M}^H\mathbf{d}_m(\lambda) \in\mathbb{C}^p$ and $\mathbf{w}(\lambda) \triangleq \boldsymbol{\Gamma}\,\mathbf{h}(\lambda) = \boldsymbol{\Gamma}\,\mathbf{M}^H\,\mathbf{d}_m(\lambda)\in\mathbb{C}^p$, satisfying $\|\mathbf{h}(\lambda)\|\leq\|\mathbf{M}\|<\infty$ and $\|\mathbf{w}(\lambda)\|\leq\|\boldsymbol{\Gamma}\|\,\|\mathbf{M}\|
\triangleq K < \infty$ uniformly in $\lambda$ (Assumption~3).

Denote by $\boldsymbol{\delta}_j\in\mathbb{C}^p$ the $j$-th row of $\boldsymbol{\delta}$, viewed as a column vector. The
$\{\boldsymbol{\delta}_j\}_{j=0}^{N-1}$ are i.i.d. $\mathcal{CN}(\mathbf{0},\mathbf{I}_p)$.
Define the scalar random variables $w_j(\lambda) \triangleq \boldsymbol{\delta}_j^T\,\mathbf{w}(\lambda)
= \displaystyle\sum_{i=1}^p \delta_{ji}\,w_i(\lambda)$ and $
z_j(\lambda) \triangleq \boldsymbol{\delta}_j^T\,\mathbf{h}(\lambda)
= \displaystyle\sum_{i=1}^p \delta_{ji}\,h_i(\lambda)$. 
These are two distinct linear forms of the same vector $\boldsymbol{\delta}_j$,
hence correlated. Decomposing $\boldsymbol{\delta}^H\,\boldsymbol{\Gamma}\,
\boldsymbol{\delta} = \displaystyle\sum_{j=0}^{N-1}\boldsymbol{\delta}_j^*\,\boldsymbol{\delta}_j^T\,
\boldsymbol{\Gamma}$ and applying to $\mathbf{h}(\lambda)$:
\begin{equation}
\hat{\gamma}^{\mathrm{sign}}_m(\lambda)
= \frac{1}{N}\,\mathbf{h}^H(\lambda)\,\boldsymbol{\delta}^H\,
  \boldsymbol{\Gamma}\,\boldsymbol{\delta}\,\mathbf{h}(\lambda)
= \frac{1}{N}\sum_{j=0}^{N-1}
  z_j^\ast(\lambda)\,w_j(\lambda)\,.
\label{eq:sign_product}
\end{equation}
Note that $\hat{\gamma}^{\mathrm{sign}}_m(\lambda)\in\mathbb{R}$ since it is a diagonal entry of the Hermitian matrix
$\displaystyle\frac{1}{N}\,\mathbf{d}_m^H(\lambda)\, \mathbf{Y}\,\mathbf{Y}^H\, \mathbf{d}_m(\lambda)$.

Having expressed $\hat\gamma^{\mathrm{sign}}_m(\lambda)$ as an average of i.i.d. complex products~\eqref{eq:sign_product}, we next compute their joint distribution, which will reveal the mean and provide the sub-exponential bound needed for concentration.

\medskip
\noindent\textit{Step 2: joint distribution of $(z_j, w_j)^T$.} Since $z_j$ and $w_j$ are both linear forms of $\boldsymbol{\delta}_j\sim\mathcal{CN}(\mathbf{0},\mathbf{I}_p)$:
\[
\begin{pmatrix}z_j(\lambda)\\ w_j(\lambda)\end{pmatrix}
= \begin{pmatrix}\mathbf{h}^T(\lambda)\\ \mathbf{w}^T(\lambda)\end{pmatrix}
\boldsymbol{\delta}_j
\sim \mathcal{CN}\!\left(\mathbf{0},\,\boldsymbol{\Sigma}_{zw}\right),
\]
where the joint covariance matrix is:
\[
\boldsymbol{\Sigma}_{zw} =
\begin{pmatrix}
\|\mathbf{h}(\lambda)\|^2 & \mathbf{h}^H(\lambda)\,\boldsymbol{\Gamma}\,\mathbf{h}(\lambda) \\
\mathbf{h}^H(\lambda)\,\boldsymbol{\Gamma}^H\,\mathbf{h}(\lambda) &
\|\mathbf{w}(\lambda)\|^2
\end{pmatrix}\,.
\]
In particular, the cross-covariance satisfies:
\begin{align}
\mathbb{E}\!\left[z_j^\ast(\lambda)\,w_j(\lambda)\right]
&=
\mathbf{h}^H(\lambda)\,
\boldsymbol{\Gamma}\,
\mathbf{h}(\lambda)
\nonumber\\
&=
\mathbf{d}_m^H(\lambda)\,
\mathbf{M}\,
\boldsymbol{\Gamma}\,
\mathbf{M}^H\,
\mathbf{d}_m(\lambda).
\label{eq:sign_mean}
\end{align}
We now show this expectation vanishes as $m, N\to\infty$.

\medskip
\noindent\textit{Step 3: the mean tends to zero.} By~\eqref{eq:sign_mean}, we have: $\mathbb{E}\!\left[\hat{\gamma}^{\mathrm{sign}}_m(\lambda)\right]
= \mathbf{d}_m^H(\lambda)\,\mathbf{M}\,\boldsymbol{\Gamma}\,
  \mathbf{M}^H\,\mathbf{d}_m(\lambda)$. 
Since the columns of $\mathbf{M}$ are absolutely summable (Assumption~3), for each $j \in \{1,\ldots,p\}$, we can rewrite $\left|[\mathbf{M}^H\,\mathbf{d}_m(\lambda)]_j\right|^2$ as:
\begin{equation}
\frac{1}{m}\left|\sum_{k=0}^{m-1}[\mathbf{M}]_{k,j}\,e^{-ik\lambda}\right|^2\, 
\leq \frac{1}{m}\left(\sum_{k=0}^{m-1}|[\mathbf{M}]_{k,j}|\right)^2
\xrightarrow{m\to\infty} 0\,,
\end{equation}
since $\displaystyle\sum_k |[\mathbf{M}]_{k,j}| < \infty$ by Assumption~3-(i), so $\bigl(\displaystyle\sum_k |[\mathbf{M}]_{k,j}|\bigr)^2$ is a finite constant independent of $m$. Therefore, $\left|[\mathbf{M}^H\,\mathbf{d}_m(\lambda)]_j\right|^2 \to 0$ uniformly in $\lambda$, and:
\[
\left|\mathbb{E}\!\left[\hat{\gamma}^{\mathrm{sign}}_m(\lambda)\right]\right|
\leq \|\boldsymbol{\Gamma}\|\,\|\mathbf{M}^H\, \mathbf{d}_m(\lambda)\|^2
\xrightarrow{N\to\infty} 0\,,
\]
uniformly in $\lambda$. The mean, therefore, tends to zero, but we must also control the fluctuations $\displaystyle\frac{1}{N}\sum_j \xi_j(\lambda)$ around this vanishing mean.  The key insight is that each $\xi_j(\lambda)$ is sub-exponential (as a centered product of two sub-Gaussian variables), enabling a Bernstein bound.

\medskip
\noindent\textit{Step 4: Sub-exponential concentration.} Recall $\|X\|_{\psi_r}=\inf\{t>0:\mathbb{E}[e^{|X|^r/t^r}]\leq 2\}$ for $r\in\{1,2\}$;
$X$ is sub-Gaussian iff $\|X\|_{\psi_2}<\infty$.

Let $\xi_j(\lambda) \triangleq z_j^\ast(\lambda)\,w_j(\lambda) - \mathbb{E}[z_j(\lambda)^\ast \, w_j(\lambda)]$. Since $z_j$, $w_j$ are Gaussian linear forms (Step~2), they are sub-Gaussian with $\|z_j(\lambda)\|_{\psi_2}\leq C\|\mathbf{M}\|$ and $\|w_j(\lambda)\|_{\psi_2}\leq C\,\|\boldsymbol{\Gamma}\|\|\mathbf{M}\|$ uniformly in $\lambda$ \cite[Ex.~2.5.8]{Vershynin2018}. Their product $\xi_j(\lambda)$ is sub-exponential \cite[Lemma~2.7.7]{Vershynin2018}:
\begin{align}
\|\xi_j(\lambda)\|_{\psi_1}
&\le b
\triangleq
C\,\|\mathbf{M}\|^2\,
\|\boldsymbol{\Gamma}\|,
\\
\mathbb{E}\!\left[|\xi_j(\lambda)|^2\right]
&\le \nu^2
\triangleq
C\,\|\mathbf{M}\|^4\,
\|\boldsymbol{\Gamma}\|^2.
\end{align}
uniformly in $j$ and $\lambda$. Bernstein's inequality \cite[Theorem~2.8.1]{Vershynin2018} then gives for any $t>0$:
\begin{equation}\label{eq:sign_bernstein_final}
P\!\left(\left|\frac{1}{N}\sum_{j=0}^{N-1}\xi_j(\lambda)\right|>t\right)
\leq 2\exp\!\left(-CN\min\!\left(\frac{t^2}{\nu^2},\frac{t}{b}\right)\right)\,,
\end{equation}
which is summable in $N$ for any fixed $t>0$. Recalling that $\hat{\gamma}^{\mathrm{sign}}_m(\lambda)
= \displaystyle \frac{1}{N}\sum_{j=0}^{N-1} \xi_j(\lambda)
+ \mathbb{E}\!\left[\hat{\gamma}^{\mathrm{sign}}_m(\lambda)\right]$, and using Step~3, which ensures that $\mathbb{E}\!\left[\hat{\gamma}^{\mathrm{sign}}_m(\lambda)\right]
\longrightarrow 0$, we obtain that, for any fixed $\lambda$, $\hat{\gamma}^{\mathrm{sign}}_m(\lambda) \xrightarrow{a.s.} 0$ by the Borel--Cantelli lemma. 

Pointwise almost sure convergence is established. We extend it uniformly to $\lambda\in[0,2\pi)$ via a discretisation argument analogous to those used above.

\medskip
\noindent\textit{Step 5: uniform control over $[0,2\pi)$.} For $|\lambda-\lambda'|\leq 2\pi/\lfloor N^\beta\rfloor$ and by setting $A_j = |z_j^\ast(\lambda)|
\,|w_j(\lambda)-w_j(\lambda')|$ and $B =|z_j^\ast(\lambda)-z_j^\ast(\lambda')|
\,|w_j(\lambda')|$, we can rewrite $\left|
\hat{\gamma}^{\mathrm{sign}}_m(\lambda)
-
\hat{\gamma}^{\mathrm{sign}}_m(\lambda')
\right|$ as 
\begin{equation*}
\Bigl|
\frac{1}{N}
\sum_{j=0}^{N-1}
\bigl(
z_j^\ast(\lambda)w_j(\lambda)
-
z_j^\ast(\lambda')w_j(\lambda')
\bigr)
\Bigr| \le \frac{1}{N} \sum_{j=0}^{N-1}(A_j+B_j)\, .
\end{equation*}
Since $z_j(\lambda)-z_j(\lambda') = \boldsymbol{\delta}_j^T\,\mathbf{M}^H \left(\mathbf{d}_m(\lambda) -\mathbf{d}_m(\lambda')\right)$ and similarly for $w_j$, and since $\lambda\mapsto\mathbf{d}_m(\lambda)$ is Lipschitz with constant
$(m-1)/\sqrt{m}$:
\[
\begin{aligned}
&\max\!\left(
\left|z_j(\lambda)-z_j(\lambda')|,
|w_j(\lambda)-w_j(\lambda')\right|
\right)
\\
&\qquad\le
K\,
\|\boldsymbol{\delta}_j\|
\|\mathbf{M}\|
\frac{m-1}{\sqrt m}
\,|\lambda-\lambda'|\,.
\end{aligned}
\]
Bounding $|z_j(\lambda)|\leq\|\boldsymbol{\delta}_j\|\|\mathbf{M}\|$ and $|w_j(\lambda)|\leq\|\boldsymbol{\delta}_j\|K$, one obtains:
\[
\left|\hat{\gamma}^{\mathrm{sign}}_m(\lambda)
-\hat{\gamma}^{\mathrm{sign}}_m(\lambda')\right|
\leq \frac{C\,\|\mathbf{M}\|^2 K}{N}
\,\frac{(m-1)}{\sqrt{m}\,\lfloor N^\beta\rfloor}
\,\frac{1}{N}\sum_{j=0}^{N-1}\|\boldsymbol{\delta}_j\|^2\,.
\]
Since $\sum_j\|\boldsymbol{\delta}_j\|^2/N\overset{a.s.}{\to}p$ (Law of Large Numbers, each $\|\boldsymbol{\delta}_j\|^2$ has mean $p$) and $(m-1)/(\sqrt{m}\lfloor N^\beta\rfloor)=O(N^{1/2-\beta})\to 0$ for $\beta>1/2$, 
$\displaystyle \max_{i\in I}\sup_{\lambda\in[\lambda_i,\lambda_{i+1})}
\left|\hat{\gamma}^{\mathrm{sign}}_m(\lambda)
-\hat{\gamma}^{\mathrm{sign}}_m(\lambda_i)\right|
\overset{a.s.}{\longrightarrow} 0$. Combining with the pointwise almost sure convergence of Step~4 and a union bound over the $\lfloor N^\beta\rfloor$ grid points (valid since the bound in~\eqref{eq:sign_bernstein_final} is summable in $N$), we obtain $\displaystyle \sup_{\lambda\in[0,2\pi)}\hat{\gamma}^{\mathrm{sign}}_m(\lambda)
\overset{a.s.}{\longrightarrow} 0$.

\subsection*{S-A.2\quad Proof of Theorem~3 (Consistency of $\widetilde{\mathbf{C}}_{FP}$ and $\check{\mathbf{C}}_{FP}$)}

The proof follows the same idea. Recall that $\widetilde{\mathbf{C}}_{FP}=\mathcal{T}
\left(\hat{\mathbf{C}}_{FP}\right)$ and 
$\check{\mathbf{C}}_{FP}=\widetilde{\mathbf{C}}_{FP}/
\mathbb{E}[v(\tau\,\xi)\,\tau]$. 
The equation to prove is $\left\|\widetilde{\mathbf{C}}_{FP}
-\mathbb{E}[v(\tau\,\xi)\,\tau]\,\mathbf{C}\right\|
\xrightarrow{a.s.}0$, from which $\|\check{\mathbf{C}}_{FP} - \mathbf{C}\| \xrightarrow{a.s.}0$ follows immediately by dividing by the positive constant $\mathbb{E}[v(\tau\,\xi)\,\tau]$. This splits as:
\begin{align}
\left\|\mathcal{T}\left(\hat{\mathbf{C}}_{FP}\right)
-\mathbb{E}[v(\tau\,\xi)\,\tau]\,\mathbf{C}\right\|
& \leq \left\|\mathcal{T}\left(\hat{\mathbf{C}}_{FP}-\hat{\mathbf{S}}\right)\right\| \nonumber \\
& +\left\|\mathcal{T}\left(\hat{\mathbf{S}}\right)
-\mathbb{E}[v(\tau\,\xi)\,\tau]\,\mathbf{C}\right\| \, , \nonumber \\
 &\triangleq  \, \eta_1+\eta_2\, ,
\label{eq:thm3_split}
\end{align}
where $\hat{\mathbf{S}}=\displaystyle \frac{1}{N}\sum_{i=0}^{N-1} v(\tau_i\, \xi)\,\mathbf{y}_{wi}\, \mathbf{y}_{wi}^H$.

\noindent\textbf{Part 1: convergence of
$\eta_1=\left\|\mathcal{T}\left(\hat{\mathbf{C}}_{FP}
-\hat{\mathbf{S}}\right)\right\|$.}

It is proven in~\cite{Couillet15b} that
$\left\|\hat{\mathbf{C}}_{FP}-\hat{\mathbf{S}}\right\|\xrightarrow{a.s.}0$. Since the Toeplitz rectification satisfies 
$\|\mathcal{T}(\mathbf{A})\| \leq \displaystyle\sup_\lambda |\mathbf{d}_m^H(\lambda)\mathbf{A}
\mathbf{d}_m(\lambda)| \leq \|\mathbf{A}\|$ for any Hermitian matrix $\mathbf{A}$, we have:
\begin{align*}
\left\|\mathcal{T}\left(\hat{\mathbf{C}}_{FP}-\hat{\mathbf{S}}\right)\right\|
&\leq\sup_{\lambda\in[0,2\pi)}\left|
\mathbf{d}_m^H(\lambda)\,\left(\hat{\mathbf{C}}_{FP}-\hat{\mathbf{S}}\right)\,
\mathbf{d}_m(\lambda)\right| \, , \\
& \leq\left\|\hat{\mathbf{C}}_{FP}-\hat{\mathbf{S}}\right\|
\xrightarrow{a.s.}0\, .
\end{align*}

\noindent\textbf{Part 2: convergence of $\eta_2 = \left\|\mathcal{T}\!\left(\hat{\mathbf{S}}\right)
- \mathbb{E}[v(\tau\,\xi)\,\tau]\,\mathbf{C}\right\|$.}

\begin{lemma}
\label{lem:gamma_shat}
One has:
\begin{align}
\hat{\gamma}^{\hat{\mathbf{S}}}_m(\lambda)
&=\mathbf{d}_m^H(\lambda)\,\hat{\mathbf{S}}\,\mathbf{d}_m(\lambda)\, ,
\label{eq:gamma_shat} \\
\mathbb{E}\!\left[\hat{\gamma}^{\hat{\mathbf{S}}}_m(\lambda)\right]
&=\mathbb{E}[v(\tau\,\xi)\,\tau]\,\mathbf{d}_m^H(\lambda)\,\mathbf{C}\,
\mathbf{d}_m(\lambda)\, .
\label{eq:mean_shat}
\end{align}
\end{lemma}

\begin{IEEEproof}
Equation~\eqref{eq:gamma_shat} follows by the same computation as Lemma~\ref{lem:gamma_hat}. For~\eqref{eq:mean_shat}, letting
$\mathbf{D}$ be the diagonal matrix with entries $\{v(\tau_i\,\xi)\}$:
\[
\mathbb{E}\!\left[\hat{\gamma}^{\hat{\mathbf{S}}}_m(\lambda)\right] =\mathbf{d}_m^H(\lambda)\,\mathbb{E}\!\left[ \frac{\mathbf{Y}_w\,\mathbf{D}\,\mathbf{Y}_w^H}{N}\right] \mathbf{d}_m(\lambda)\, .
\]
Since the signal term $\mathbf{C}^{-1/2}\mathbf{M}\,\boldsymbol{\Gamma}^{1/2}\,\boldsymbol{\delta}^H$ 
in $\mathbf{Y}_w$ is independent of $\mathbf{D}$ (which depends only on $\{\tau_i\}$), and since 
$\mathbf{D}$ weights only $\mathbf{X}$ (cf.\ the definition of $\hat{\mathbf{S}}$ in Eq.~(7)), the cross terms vanish and
$\left[\mathbb{E}\left[\mathbf{Y}_w\,\mathbf{D}\,\mathbf{Y}_w^H\right]\right]_{ij}
= \displaystyle \sum_{n=0}^{N-1}\mathbb{E}[v(\tau_n\, \xi)\tau_n]
= N\,\mathbb{E}[v(\tau\xi)\tau]$\,,
so Equation~\eqref{eq:mean_shat} follows. We split $\eta_2$ as
\begin{equation}
\begin{aligned}
\eta_2 &= \underbrace{ \sup_{\lambda}\Bigl| \hat{\gamma}^{\hat{\mathbf S}}_m(\lambda) - \mathbb{E}\left[ \hat{\gamma}^{\hat{\mathbf S}}_m(\lambda)
\right] \Bigr| }_{\eta_{11}}
\\
&\quad+ \underbrace{ \sup_{\lambda}\Bigl| \mathbb{E}\!\left[ \hat{\gamma}^{\hat{\mathbf S}}_m(\lambda) \right] - \mathbb{E}[v(\tau\,\xi)\,\tau]\, \gamma_m(\lambda) \Bigr| }_{\eta_{12}} .
\end{aligned}
\end{equation}

The term $\eta_{11}\xrightarrow{a.s.}0$ follows by the same concentration argument as Lemmas~S-4 and~S-5, with $\|\mathbf{T}\|_\infty$ replaced by $\|\mathbf{T}\|_\infty\,\|\mathbf{D}\|_\infty$, where $\mathbf{D}=\mathrm{diag}(v(\tau_0\,\xi),\ldots,v(\tau_{N-1}\,\xi))$ satisfies $\|\mathbf{D}\|_\infty \leq \Phi_\infty/\xi < \infty$ almost surely. The term $\eta_{12}=0$ follows directly from Lemma~S-9 and Eq.~\eqref{eq:mean_shat}, since the bias of $\hat{\gamma}^{\hat{\mathbf{S}}}_m(\lambda)$ equals $\mathbb{E}[v(\tau\,\xi)\,\tau]\,\gamma_m(\lambda)$, which is exactly the target.
\end{IEEEproof}

Since $\eta_1,\eta_2\xrightarrow{a.s.}0$, Theorem~3 follows.\qed

\section{Proofs of Theorems~2 and~4}
\label{sec:supB}
\begin{IEEEproof}
The proofs for $\check{\boldsymbol{\Sigma}}_{SCM}$ and $\check{\boldsymbol{\Sigma}}_{FP}$ are identical; let $\check{\boldsymbol{\Sigma}}$ and $\check{\mathbf{C}}$ denote either pair.

Since $\check{\mathbf{y}}_{wi}=\check{\mathbf{C}}^{-1/2}\,\mathbf{y}_i$, the estimator $\check{\boldsymbol{\Sigma}}$ satisfies a fixed-point equation in $\check{\mathbf{C}}^{-1/2}\,\mathbf{y}_i$. Define $\widetilde{\boldsymbol{\Sigma}} \triangleq
\mathbf{C}^{-1/2}\,\check{\mathbf{C}}^{1/2}\,\check{\boldsymbol{\Sigma}}
\,\check{\mathbf{C}}^{1/2}\,\mathbf{C}^{-1/2}$.
Using the identity $\check{\mathbf{y}}_{wi} = \check{\mathbf{C}}^{-1/2} \,\mathbf{C}^{1/2}\,\mathbf{y}_{wi}$ (since $\check{\mathbf{y}}_{wi} = \check{\mathbf{C}}^{-1/2}\mathbf{y}_i$ and $\mathbf{y}_{wi} = \mathbf{C}^{-1/2}\mathbf{y}_i$), a direct substitution into the fixed-point equation of $\check{\boldsymbol{\Sigma}}$ shows that $\widetilde{\boldsymbol{\Sigma}}$ satisfies the \emph{same} fixed-point equation as $\hat{\boldsymbol{\Sigma}}$ on $\mathbf{Y}_w$.

Since $\check{\mathbf{C}}$ is positive definite (a.s.), $\widetilde{\boldsymbol{\Sigma}}$ inherits the required positivity and trace normalisation, so uniqueness of the fixed point~\cite{Couillet15b} gives $\widetilde{\boldsymbol{\Sigma}} = \hat{\boldsymbol{\Sigma}}$, which, upon rearranging, gives
\begin{equation}
\check{\boldsymbol{\Sigma}}
= \check{\mathbf{C}}^{-1/2}\,\mathbf{C}^{1/2}\,\hat{\boldsymbol{\Sigma}}
  \,\mathbf{C}^{1/2}\,\check{\mathbf{C}}^{-1/2}\,.
\label{eq:Sigma_check_hat}
\end{equation}

Since $\|\hat{\boldsymbol{\Sigma}}-\hat{\mathbf{S}}\|\xrightarrow{a.s.}0$ by \cite{Couillet15b}, it
suffices to show that $\|\check{\boldsymbol{\Sigma}}-\hat{\boldsymbol{\Sigma}}\|\xrightarrow{a.s.}0$.
From~\eqref{eq:Sigma_check_hat} and the triangle inequality:
\begin{equation}
\begin{aligned}
\left\|
\check{\boldsymbol{\Sigma}}
-
\hat{\boldsymbol{\Sigma}}
\right\|
&\le
\left\|
\check{\mathbf C}^{-1/2}\mathbf C^{1/2}
-\mathbf I
\right\|
\left\|
\hat{\boldsymbol{\Sigma}}
\right\|
\left\|
\mathbf C^{1/2}\check{\mathbf C}^{-1/2}
\right\|
\\
&\quad+
\left\|
\hat{\boldsymbol{\Sigma}}
\right\|
\left\|
\mathbf C^{1/2}\check{\mathbf C}^{-1/2}
-\mathbf I
\right\|.
\end{aligned}
\end{equation}

By Theorems~1/3, $\left\|\check{\mathbf{C}}-\mathbf{C}\right\| \xrightarrow{a.s.} 0$. Since the coefficients $\{c_k\}$ are absolutely summable (Assumption~2-(ii)), $\mathbf{C}$ has spectrum bounded in a compact interval $(0,\infty)$, and for $N$ large enough $\check{\mathbf{C}}$ shares this property a.s. Since the map $\mathbf{A}\mapsto\mathbf{A}^{-1/2}$ is Lipschitz on any compact subset of positive-definite matrices \cite[Theorem~X.1.1]{bhatia1997}, and since $\check{\mathbf{C}}$ is a.s.\ positive definite with spectrum in a fixed compact interval
$(0,\infty)$ for $N$ large enough (by Theorems~S-1 and~S-3 above), we obtain:
\begin{equation}
\left\|\check{\mathbf{C}}^{-1/2}\,\mathbf{C}^{1/2}-\mathbf{I}\right\|
\leq C\,\left\|\check{\mathbf{C}}-\mathbf{C}\right\|
\xrightarrow{a.s.} 0\,,
\end{equation}
for some Lipschitz constant $C < \infty$.

The norms $\left\|\hat{\boldsymbol{\Sigma}}\right\|$ and $\left\|\mathbf{C}^{1/2}\, \check{\mathbf{C}}^{-1/2}\right\|$ are almost surely bounded: the former by \cite{Couillet15b}, the latter because $\left\|\check{\mathbf{C}}^{-1}\right\|$ is bounded as $\check{\mathbf{C}}\to\mathbf{C}$ which has bounded away-from-zero spectrum. Hence $\left\|\check{\boldsymbol{\Sigma}}-\hat{\boldsymbol{\Sigma}}\right\|\xrightarrow{a.s.} 0$, completing the proof. 
\end{IEEEproof}
\putbib[bibliographieeug]
\end{bibunit}

\end{document}